\documentclass[11pt,a4paper]{article}%
\usepackage{amssymb,amsmath, amsfonts}
\usepackage{graphicx,graphics}
\usepackage{mathtools}
\usepackage{bbold}
\usepackage[english]{babel}
\usepackage[utf8]{inputenc}
\usepackage{epsfig,url}
\usepackage{bbm,theorem}
\usepackage{a4wide}
\usepackage{color}
\usepackage{enumerate}
\usepackage{calrsfs}
\usepackage[bookmarks=true]{hyperref}
\usepackage{bookmark}
\usepackage{amsmath}
\usepackage{amsfonts}
\usepackage{amssymb}
\usepackage{verbatim}
\usepackage{graphicx}%
\providecommand{\U}[1]{\protect\rule{.1in}{.1in}}
\DeclareMathAlphabet{\pazocal}{OMS}{zplm}{m}{n}
\newtheorem{theorem}{Theorem}[section]
\newtheorem{definition}[theorem]{Definition}

\newtheorem{lemma}[theorem]{Lemma}
\newtheorem{proposition}[theorem]{Proposition}

{\theorembodyfont{\upshape}
\newtheorem{remark}[theorem]{Remark}

}
\numberwithin{equation}{section}
\numberwithin{theorem}{section}
\newcommand{\qed}{\hfill$\Box$}
\newenvironment{proof}{\begin{trivlist}\item[]{\em Proof:}\/}{\qed\end{trivlist}}
\newenvironment{proofof}[1][Proof]{\noindent \textit{#1.} }{\ \qed}

\newcommand{\R}{{\mathbb R}}

\DeclareMathOperator{\Tr}{Tr}

\newcommand{\beq}{\begin{equation}}
\newcommand{\eeq}{\end{equation}}
\newcommand{\beqs}{\begin{eqnarray}}
\newcommand{\eeqs}{\end{eqnarray}}

\newcounter{jlisti}

\title{Self-similar profiles for homoenergetic solutions of the Boltzmann equation: particle velocity distribution and entropy}
\author{Richard D. James \thanks{\emailrichard} , Alessia Nota \thanks{\emailalessia} , Juan J. L. Vel\'azquez \thanks{\emailjuan} \\[1em]
$\,^*$\UMaddress \\[0.5em]
$ ^\dag\,^\ddag$\UBaddress}

\date{\today}

\newcommand{\email}[1]{E-mail: \tt #1}
\newcommand{\emailrichard}{\email{james@umn.edu}}
\newcommand{\emailalessia}{\email{nota@iam.uni-bonn.de}}
\newcommand{\emailjuan}{\email{velazquez@iam.uni-bonn.de}}

\newcommand{\UBaddress}{\em University of Bonn, Institute for Applied Mathematics\\
\em Endenicher Allee 60, D-53115 Bonn, Germany}
\newcommand{\UMaddress}{\em Department of Aerospace Engineering and Mechanics, University of Minnesota \\
\em 107, Akerman Hall, Minneapolis, MN 55455, USA}

\begin{document}

\maketitle

\begin{abstract}
In this paper we study a class of solutions of
the Boltzmann equation which have the form $f\left(  x,v,t\right)  =g\left(
v-L\left(  t\right)  x,t\right)  $ where $L\left(  t\right)  =A\left(
I+tA\right)  ^{-1}$ with the matrix $A$ describing a shear flow or a
dilatation or a combination of both. These solutions are known as
homoenergetic solutions.   We prove existence of homoenergetic solutions 
for a large class of initial data.  For different choices for the matrix
$A$ and for different homogeneities of the collision kernel, we characterize the
long time asymptotics of the velocity distribution for the corresponding
homoenergetic solutions.  For a large class of choices of $A$ we then prove rigorously
the existence of self-similar solutions of the Boltzmann equation.  The latter are non
Maxwellian distributions and describe far-from-equilibrium flows.  For
Maxwell molecules we obtain exact formulas for the $H$-function for some of these
flows.  These formulas show that  in some cases, despite being very
far from equilibrium, the relationship between density, temperature and entropy
is exactly the same as in the equilibrium case.   We make conjectures about the asymptotics of
homoenergetic solutions that do not have self-similar profiles.  

\end{abstract}
\tableofcontents
\section{Introduction. \label{sect1}}

In this paper we study homoenergetic solutions of the
Boltzmann equation.  Our approach is motivated by
an invariant manifold of solutions of the equations of classical molecular dynamics with 
certain symmetry properties (\cite{md, viscometry}).  

Briefly and formally, this manifold can be described as follows.  Choose a matrix 
$A\in M_{3\times3}\left(  \mathbb{R}\right) $, let $e_1, e_2 ,e_3$ be linearly
independent vectors in $\mathbb R^3$,  and consider
a time interval $[0,a)$ such that $\det (I + t A)>0$ for $ t \in [0, a)$ with $a>0$.  Consider any
number of atoms labeled $1, \dots, M$ with positive masses $m_1, \dots, m_M$  and any initial conditions 
\begin{align}
{y}_{k}(0)    = {y}^{0}_{k}, \quad
\dot{y}_{k}(0)   = {v}^{0}_{k},\quad  k = 1, \dots, M.  \label{md_ic}
\end{align}
Call these $M$ atoms the {\it simulated atoms}.   The simulated atoms will be
subject to the equations of molecular dynamics (to be stated presently) with the
initial conditions (\ref{md_ic}), yielding solutions $y_k(t) \in \mathbb R^3, \  0 \le t < a,\ k = 1, \dots , M $.
In addition there will be {\it non-simulated atoms} with time-dependent positions
$y_{\nu,k}(t)$, indexed by a triple of integers $\nu = (\nu_1, \nu_2, \nu_3) \in \mathbb Z^3, \ \nu \ne (0,0,0)$ and 
$ k = 1, \dots, M$.  The nonsimulated atom $(\nu,k)$ will have mass $m_k$.  The
positions of the nonsimulated atoms will be given by the following explicit formulas based on the
positions of the simulated atoms:
\beq
y_{\nu,k}(t) = y_k(t)  + (I + t A)(\nu_1 e_1 + \nu_2 e_2 +  \nu_3 e_3),   \quad \nu = (\nu_1, \nu_2, \nu_3) \in \mathbb Z^3, \ k = 1, \dots, M.  \label{nonsim}
\eeq
For $k = 1, \dots, M$ let 
$f_{k}: \cdots \mathbb R^3 \times \mathbb R^3 \times {\mathbb R^3}  \cdots \to \mathbb R$ be the force
on simulated atom $k$.  Naturally, the force on simulated atom $k$ depends on the positions of all the atoms.  
This force is required to satisfy the standard conditions of frame-indifference and permutation invariance
\cite{md}.  Formally, the equations of molecular dynamics for the simulated atoms are:
\beqs
m_k \ddot{y}_k  &=& f_k( \dots, y_{\nu_1, 1}, \dots, y_{\nu_1, M}, \dots, y_{\nu_2, 1}, \dots, y_{\nu_2, M}, \dots), \label{md_sim}  \\
{y}_{k}(0)    &=& {y}^{0}_{k}, \quad
\dot{y}_{k}(0)   = {v}^{0}_{k},\quad  k = 1, \dots, M.  \nonumber 
\eeqs
Note that these are ODEs in standard form for the motions of the simulated atoms since, for the nonsimulated
atoms, we assume that the formulas (\ref{nonsim}) have been substituted into the right hand side
of (\ref{md_sim}).  It is shown in \cite{md} and \cite{viscometry} that, 
even though the motions of the nonsimulated atoms are only given by formulas, the equations of molecular dynamics are exactly satisfied for each
nonsimulated atom.  

While this is stated formally here, if conditions are given on the $f_k$ such that the standard existence and
uniqueness theorem holds for the initial value problem (\ref{md_ic}), (\ref{md_sim}), then the result holds rigorously.  
 The proof is a simple consequence of frame-indifference and permutation invariance of atomic forces.
The result can be rephrased as the existence of a certain family of time-dependent invariant manifolds of 
molecular dynamics.

These results on molecular dynamics have a simple interpretation in terms of the molecular density function of
the kinetic theory.  Consider a molecular dynamics simulation of the type described above.  
Consider a ball $B_r(x)$ of any radius $r$ centered at $x = (I + t A)(\nu_1 e_1 + \nu_2 e_2 +  \nu_3 e_3)$, $(\nu_1, \nu_2, \nu_3) \in \mathbb Z^3$.  The
ansatz (\ref{nonsim}) implies that, the velocities of all atoms in the ball $B_r(x)$ are completely determined by those
in the ball $B_r(0)$.  But the molecular density function $f(t, x, v)$ of the kinetic theory is supposed to describe
the probability density of finding velocities in the small neighborhood of a point $x$ at time $t$.  Thus,
the ansatz associated to this observation about balls can be
immediately written down based on (\ref{nonsim}) and its time-derivative.  It is
\beq
f(t, x, v) = g(t, v - A(I + tA)^{-1}x).  \label{ansatz}
\eeq
(The emergence of the quantity $A(I + t A)^{-1}$ arises from conversion to the Eulerian form of the 
kinetic theory.)  

Besides the reasons mentioned below, the study of these solutions is interesting
from the general perspective of non-equilibrium statistical mechanics.  Essentially, we show for broad classes
of choices of $A$, there exist solutions of the Boltzmann equation satisfying (\ref{ansatz}).  This means
that, in a precise sense, this invariant manifold of molecular dynamics is inherited by the
Boltzmann equation.  This is true despite the fact that the Boltzmann equation is time irreversible, while molecular dynamics is
time reversible.  It is then particularly interesting to look at the form of the entropy (minus the $H$-function)
in these cases.  We give explicit relations satisfied by the entropy in some cases, that can be
considered as derived constitutive relations.  It would now be
extremely interesting to study these relations in molecular dynamics.  
Besides the entropy, our results give new insight into the relation between atomic
forces and nonequilibrium behavior. 

An alternative viewpoint leading to the same result is presented in Section \ref{hom}.  That
derivation is based on the viewpoint of {\it equidispersive solutions}, i.e., an ansatz of the form
\beq
f\left(  t,x,v\right)  =g\left(  t,w\right)  \text{ \ \ with }w=v-\xi\left(
t,x\right). 
\eeq
Under mild conditions of smoothness, this ansatz is found to reduce the Boltzmann equation if and 
only if $\xi (t, x) = A (I + t A)^{-1}x$.  

Formally, if  $f$ is a solution of the Boltzmann equation (\ref{A0_0}) of the form (\ref{ansatz}) 
the function $g$ satisfies
\begin{equation}
\partial_{t}g-L\left(  t\right)  w\cdot\partial_{w}g=\mathbb{C}g\left(
w\right)  \label{D1_0}%
\end{equation}
where the collision operator $\mathbb{C}$ is defined as in \eqref{A0_0}. 
These solutions are called {\it homoenergetic solutions} and were introduced 
by Galkin \cite{Galkin1} and Truesdell  \cite{T}.

Homoenergetic solutions of the Boltzmann equation have been studied in
\cite{Bobylev75}, \cite{Bobylev76}, \cite{BobCarSp}, \cite{CercArchive},
\cite{Cerc2000}, \cite{Cerc2002}, \cite{Galkin1}, \cite{Galkin2}, \cite{Galkin3}, \cite{garzo},
\cite{Nikol1}, \cite{Nikol2}, \cite{T}, \cite{TM}. Details about the precise contents of these papers will be given later in the corresponding sections where similar results appears.  To summarize this literature, we refer to interaction
potentials of the form $V\left(  x\right)  =\frac{1}{\left\vert x\right\vert
^{\nu-1}}$, which  have {\it homogeneity} of the kernel $\gamma=\frac{\nu-5}{\nu-1}$.  The case of Maxwell molecules corresponds to the case $\nu=5$, that is,  homogeneity $\gamma = 0$. 
In this case the moments $M_{j,k}$ associated to the
function $g$ defined in (\ref{B1_0}) satisfy a system of linear equations.
As in the original work of Galkin \cite{Galkin1} and Truesdell  \cite{T},
 most previous work is concerned with the computation of the
evolution of the moments $M_{j,k}=\int v_{j}v_{k}gdv$, as well as higher order moments,  in the case in which the
kernel $B$ in (\ref{A0_0}) has homogeneity $\gamma = 0$.  The
evolution equation for the moments yields a huge amount of information about
quantities like the typical deviation of the velocity and similar quantities
(\cite{BGP, garzo,TM}). 

Referring to these studies of the equations of the moments,  Truesdell and Muncaster \cite{TM} say, 
``To what extent the exact solutions in the class here exhibited
correspond to solutions of the Maxwell-Boltzmann equation is not yet established $\dots$ It is not clear whether [the moments] correspond to a molecular density''.  In this paper, although we
will use at several places the information provided by the moments, 
we will be mostly concerned with a detailed description of the distribution of velocities
and other quantities such as the $H$-function that are not accessible from the moment equations.

The initial value problem associated to (\ref{D1_0}) has been considered by Cercignani
\cite{CercArchive} for a particular choice of $L(t)$.  More precisely, Cercignani in \cite{CercArchive} (see also \cite{Cerc2000}) considered homoenergetic affine flows for the Boltzmann equation in the case of simple shear (cf., Theorem \ref{ClassHomEne}, case \ref{T1E4}) proving existence in $L^1$ of the distribution function for a large class of interaction potentials which include hard sphere and angular cut-off interactions. These solutions are in general not self-similar. 

In this paper we first prove the existence of a large class of homoenergetic solutions
and we study  their long time asymptotics.  Their behavior strongly depends
on the homogeneity of the collision kernel $B$ and on the particular form of
the hyperbolic terms,  namely
$L\left(  t\right)  w\cdot\partial_{w}g$.   We find that, depending on the
homogeneity of the kernel, we have different behaviors of the
solutions of the Boltzmann equations for large times.  Indeed, we prove the existence of 
self-similar profiles for Maxwell molecules, when the
hyperbolic part of the equation and the collision term are of the same order
of magnitude as $t\rightarrow\infty$. The resulting self-similar solutions
are different from the Maxwellian distributions. Indeed, they reflect a
nonequilibrium regime due to the balance between the hyperbolic part of the
equation (which reflects effects like shear, dilatation) and the collision
term.

\bigskip

The plan of the paper is the following. In Section \ref{hom} 
we describe the main properties of homoenergetic solutions of the Boltzmann equation. 
In Section \ref{ss:classeqsol} we characterize the long time asymptotics of 
$\xi\left(  t,x\right)  =L\left(  t\right)  x=\left(  I+tA\right)  ^{-1}Ax$,
restricting ourselves to the case in which $\det (I + t A)>0$ holds for all
$t\geq0$.   In Section \ref{ss:genprop} we prove well posedness for homoenergetic flows, and
we  prove existence of self-similar homoenergetic solutions. 
In Section \ref{SelfSimEx} we apply the general theory of Section \ref{ss:genprop} to various
homoenergetic flows described in Section
\ref{ss:classeqsol}.  In Section \ref{conj} we propose some conjectures on solutions which cannot described by
self-similar profiles.  These correspond to cases of  $L(t)$ and homogeneity $\gamma$ such that the 
collision term and the hyperbolic term do not balance.  Some of these conjectures were arrived at by
careful study of the corresponding formal Hilbert expansion, which is presented in a forthcoming paper \cite{JMNV}.
 
An important comment on the solutions discussed in this paper concerns the thermodynamic entropy. Indeed, as we point out in Section 7, there are many analogies with the corresponding formulas for the entropy for equilibrium distributions, in spite of the fact that the distributions obtained in this paper concern non-equilibrium situations.  For example, if we identify the entropy density $s$ with minus the $H$-function, then our asymptotic formulas for self-similar solutions yield the identity
\begin{equation}
\frac{s}{\rho} = \log \frac{e^{3/2}}{\rho} + C_G.
\end{equation}
But despite the fact that $s, \rho$ and $e$ can be rapidly changing functions of time for self-similar homoenergetic solutions,
{\it the relation between them is asymptotically the same as in the equilibrium case (Maxwellian distribution)}, except for one important fact. That is, the constant $C_G$ is not the same as the constant as in the equibrium case: $C_G< C_M$ where
$C_M$ is the corresponding value for the Maxwellian distribution.

Another interesting consequence of our results is further insight into the possibility (discussed in \cite{TM}) that our solutions for simple shear exhibit non-zero heat flux despite having zero temperature gradient, in contradiction to most versions of continuum thermodynamics.
A conjectured  scenario under which this could occur is described in Section \ref{HeatFluxes}.

At the end of the paper we give a table which summarizes the rigorous results and conjectures (see Section \ref{sec:tableresults}) and in Section \ref{sec:conclusions}  we conclude with a discussion to clarify the state-of-the-art of the analysis of homoenergetic solutions for the Boltzmann equation and we give some further perspectives.

\section{Homoenergetic solutions of the Boltzmann equation}
\label{hom}

The classical Boltzmann equation has the form
\begin{align}
\partial_{t}f+v\partial_{x}f  &  =\mathbb{C}f\left(  v\right)
\ \ ,\ \ f=f\left(  t,x,v\right) \nonumber\\
\mathbb{C}f\left(  v\right)   &  =\int_{\mathbb{R}^{3}}dv_{\ast}\int_{S^{2}%
}d\omega B\left(  n\cdot\omega,\left\vert v-v_{\ast}\right\vert \right)
\left[  f^{\prime}f_{\ast}^{\prime}-f_{\ast}f\right],  \ \label{A0_0}%
\end{align}
where $S^{2}$ is the unit sphere in $\mathbb{R}^{3}$ and $n=n\left(
v,v_{\ast}\right)  =\frac{\left(  v-v_{\ast}\right)  }{\left\vert v-v_{\ast
}\right\vert }.$ Here $(v,v_{\ast})$ is a pair of velocities in incoming
collision configuration (see Figure \ref{fig1}) and $(v^{\prime},v_{\ast
}^{\prime})$ is the corresponding pair of outgoing velocities defined by the
collision rule
\begin{align}
v^{\prime}  &  =v+\left(  \left(  v_{\ast}-v\right)  \cdot\omega\right)
\omega, \label{CM1}\\
v_{\ast}^{\prime}  &  =v_{\ast}-\left(  \left(  v_{\ast}-v\right)  \cdot
\omega\right)  \omega. \label{CM2}%
\end{align}
The unit vector  $\omega=\omega(v,V)$ bisects the angle between the
incoming relative velocity $V=v_{\ast}-v$ and the outgoing relative velocity
$V^{\prime}=v_{\ast}^{\prime}-v^{\prime}$ as specified in Figure \ref{fig1}.
The collision kernel $B\left(  n\cdot\omega,\left\vert v-v_{\ast}\right\vert
\right)  $ is proportional to the cross section for the scattering problem
associated to the collision between two particles. We use the conventional notation in kinetic theory,
$f=f\left(t, x, v\right)  ,\ f_{\ast
}=f\left( t, x, v_{\ast}\right)  ,\ f^{\prime}=f\left( t, x, v^{\prime}\right)
,\ \ f_{\ast}^{\prime}=f\left( t, x,  v_{\ast}^{\prime}\right)  $.

\begin{figure}[th]
\centering
\includegraphics [scale=0.3]{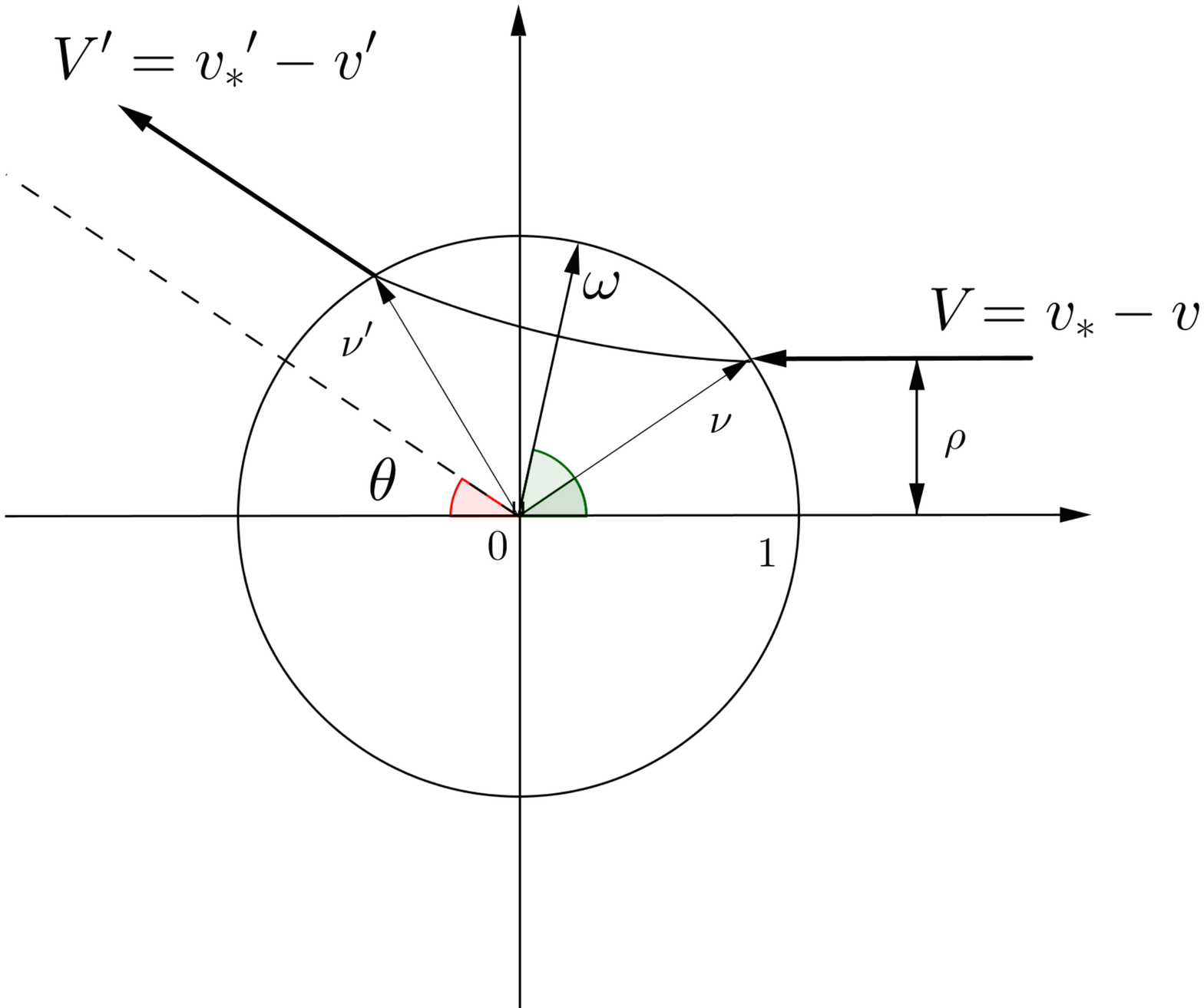}\caption{The two-body scattering.
The solution of the two-body problem lies in a plane, which is taken to be the plane
of the page, and the motion of molecule $\ast$ is plotted relative to the unstarred
molecule.  The scalar $\rho$ is the impact parameter expressed in microscopic units,
$\rho\in[-1,1]$, and  $\theta=\theta(\rho, |V|)$ is the scattering angle.  The
scattering vector of (\ref{CM1}), (\ref{CM2}) is the unit 
vector $\omega=\omega(v,V)$. \label{fig1} }%
\end{figure}

We will assume that the kernel $B$ is homogeneous with respect to the variable
$\left\vert v-v_{\ast}\right\vert $ and we will denote its homogeneity by
$\gamma,\ $i.e.,
\begin{equation}
B\left(  \frac{\left(  v-v_{\ast}\right)  \cdot\omega}{\left\vert v-v_{\ast
}\right\vert },\lambda\left\vert v-v_{\ast}\right\vert \right)  =\lambda
^{\gamma}B\left(  \frac{\left(  v-v_{\ast}\right)  \cdot\omega}{\left\vert
v-v_{\ast}\right\vert },\left\vert v-v_{\ast}\right\vert \right)
 ,\ \ \lambda>0. \label{S8E7}%
\end{equation}

Given $f\left(  t,x,v\right)  $ we can compute the density $\rho$,
the average velocity $V$ and the internal energy $\varepsilon$ at each point
$x$ and time $t$ by means of%
\begin{equation}
\rho\left(  t,x\right)  =\int_{\mathbb{R}^{3}}f\left(  t,x,v\right)
dv,\ \ \rho\left(  t,x\right)  V\left(  t,x\right)  =\int_{\mathbb{R}^{3}%
}f\left(  t,x,v\right)  vdv. \label{S8E3}%
\end{equation}
The internal energy $\varepsilon\left(  t,x\right)  $ (or temperature) is given by%
\[
\rho\left(  t,x\right)  \varepsilon\left(  t,x\right)  =\int_{\mathbb{R}^{3}%
}f\left(  t,x,v\right)  \left(  v-V\left(  t,x\right)  \right)  ^{2}dv.
\]

Homoenergetic solutions of (\ref{A0_0}) defined in \cite{Galkin1} and
\cite{T} (cf., also \cite{TM}) are solutions of the Boltzmann equation having
the form%
\begin{equation}
f\left(  t,x,v\right)  =g\left(  t,w\right)  \text{ \ \ with }w=v-\xi\left(
t,x\right).  \label{B1_0}%
\end{equation}

Notice that, under suitable integrability conditions, every solution of
(\ref{A0_0}) with the form (\ref{B1_0}) yields only time-dependent
internal energy and density%
\begin{equation}
\varepsilon\left(  t,x\right)  =\varepsilon\left(  t\right)  \ \ ,\ \ \rho
\left(  t,x\right)  =\rho\left(  t\right).  \label{eq:Hom1}%
\end{equation}
However, we have $V\left(  t,x\right)
=\xi\left(  t,x\right)  $ and therefore the average velocity 
depends also on the position.

A direct computation shows that in order to have solutions of (\ref{A0_0})
with the form (\ref{B1_0}) for a sufficiently large class of initial data we
must have%
\begin{equation}
\frac{\partial\xi_{k}}{\partial x_{j}}\text{ independent on }x\text{ and
}\partial_{t}\xi+\xi\cdot\nabla\xi=0 .\label{B2_0}%
\end{equation}

The first condition implies that $\xi$ is an affine function on $x$. However,
we will restrict attention in this paper  to the case in which $\xi$ is a
linear function of $x,$ for simplicity,  whence%
\begin{equation}
\xi\left(  t,x\right)  =L\left(  t\right)  x, \ \label{B4_0}%
\end{equation}
where $L\left(  t\right)  \in M_{3\times3}\left(  \mathbb{R}\right)  $ is a
$3\times3$ real matrix. The second condition in (\ref{B2_0}) then implies
that%
\begin{equation}
\frac{dL\left(  t\right)  }{dt}+\left(  L\left(  t\right)  \right)
^{2}=0, \quad L (0) = A, \label{B3_0}%
\end{equation}
where we have added an initial condition.

Classical ODE theory shows that there is a unique continuous solution of (\ref{B3_0}),
\begin{equation}
L\left(  t\right)  =\left(  I+tA\right)  ^{-1}A=A\left(  I+tA\right)
^{-1},\ \label{B7_0}%
\end{equation}
defined on a maximal interval of existence $[0,a)$.
On the interval $[0,a)$, $\det \left( I+tA \right)>0 $.

\section{Characterization of homoenergetic solutions defined for arbitrary
large times. \label{ss:classeqsol}}

\bigskip

In this section we describe the long time asymptotics of the function
$\xi\left(  t,x\right)  =L\left(  t\right)  x=\left(  I+tA\right)  ^{-1}Ax$
(cf. (\ref{B4_0}) and (\ref{B7_0})).
There are interesting choices of $A\in M_{3\times3}\left(  \mathbb{R}\right)  $
for which $L\left(  t\right)  $ blows up in finite time, but we will restrict attention in this paper to the
case in which the matrix $\det (I + t A) >0$ for all $t \ge 0$.   We will use in the rest of the paper the
following norm in $M_{3\times3}\left(  \mathbb{R}\right)  :$%
\begin{equation}
\left\Vert M\right\Vert =\max_{i,j}\left\Vert m_{i,j}\right\Vert
\ \ \text{with }M=\left(  m_{i,j}\right)  _{i,j=1,2,3} .\label{T3E1}%
\end{equation}

\begin{theorem}\label{ClassHomEne}  Let $A \in M_{3 \times 3}(\mathbb R) $ satisfy $\det (I + t A) >0$ for
$t\ge 0$ and let $L(t) = (I + tA)^{-1}A$.  Assume $L$ does not vanish identically.
Then, there is an orthonormal basis  (possibly different in each case) such that
the matrix of $L(t)$ in this basis has one of the following forms:

\vspace{2mm}
\noindent Case (i) Homogeneous dilatation: 
\beq
L(t)  = 
\frac{1}{t} I + O\bigg( \frac{1}{t^2} \bigg)  \quad {\rm as}\ t \to \infty.
 \label{T1E1}
 \eeq
\noindent Case (ii) Cylindrical dilatation (K=0), or Case (iii) Cylindrical dilatation and shear ($K \ne 0$):
\beq
L(t) = \frac{1}{t} \left( \begin{array}{ccc} 1& 0 & K \\
0 & 1 & 0 \\ 0 & 0 & 0 \end{array} \right) + O\bigg( \frac{1}{t^2} \bigg)  \quad {\rm as}\ t \to \infty.
 \label{T1E2}
\eeq
\noindent Case (iv). Planar shear:
\beq
L(t) = \frac{1}{t} \left( \begin{array}{ccc} 0& 0 & 0 \\
0 & 0 & K \\ 0 & 0 & 1 \end{array} \right) + O\bigg( \frac{1}{t^2} \bigg)  \quad {\rm as}\ t \to \infty. 
 \label{T1E3}
\eeq
\noindent Case (v). Simple shear:
\beq
L(t) = \left( \begin{array}{ccc} 0& K & 0 \\
0 & 0 & 0 \\ 0 & 0 & 0 \end{array} \right) , \quad K \ne 0.
 \label{T1E5}
\eeq
\noindent Case (vi). Simple shear with decaying planar dilatation/shear:
\beq
L(t) = \left( \begin{array}{ccc} 0& K_2 & 0 \\
0 & 0 & 0 \\ 0 & 0 & 0 \end{array} \right) + \frac{1}{t} \left( \begin{array}{ccc} 0 & K_1 K_3 & K_1 \\
0 & 0 & 0 \\ 0 & K_3 & 1 \end{array} \right) +    O\bigg( \frac{1}{t^2} \bigg)  , \quad K_2 \ne 0.
 \label{T1E6}
\eeq
\noindent Case (vii). Combined orthogonal shear:
\beq
L(t) =   \left( \begin{array}{ccc} 0 & K_3 & K_2 - t K_1 K_3  \\
0 & 0 & K_1 \\ 0 & 0 & 0 \end{array} \right), \quad K_1 K_3 \ne 0.
\label{T1E4}
\eeq
\label{ClassHomEne}
\end{theorem}

\noindent {\it Proof of Theorem \ref{ClassHomEne}}  The Jordan Canonical Form for real $3 \times 3$ matrices
says that there exists an orthonormal basis $e_1, e_2, e_3$ and a real invertible  matrix $P$ such that $A = P J P^{-1}$,
where the matrix $J$  has one of the following forms in this basis:
\beq
\left( \begin{array}{ccc} \alpha & \beta & 0 \\
-\beta & \alpha & 0 \\ 0 & 0 & \gamma \end{array} \right), 
\left( \begin{array}{ccc} \lambda_1 & 0  & 0 \\
0 & \lambda_2 & 0 \\ 0 & 0 & \lambda_3 \end{array} \right), 
\left( \begin{array}{ccc} \xi & 1 & 0 \\
0 & \xi & 0 \\ 0 & 0 & \eta \end{array} \right), 
\left( \begin{array}{ccc} \lambda & 1 & 0 \\
0 & \lambda  & 1 \\ 0 & 0 & \lambda  \end{array} \right). 
\eeq
All entries are real and $\beta \ne 0$.  In these four cases, respectively, we have
that $\det (I + tA) $ is 
\beq
((1+ \alpha t)^2 + t^2 \beta^2)(1+ t \gamma), \quad (1 + t \lambda_1)(1 + t \lambda_2)(1 + t \lambda_3),  \quad
(1 + t \eta)(1 + t \xi)^2,  \quad (1 + t \lambda)^3 \label{dets}
\eeq
Therefore, necessary and sufficient conditions that $\det (I + tA)>0 $ for $t \ge 0$ are, respectively,
\beq
\gamma \ge 0, \quad \quad  \lambda_1 \ge 0, \lambda_2 \ge 0, \lambda_3 \ge 0,  \quad  \quad 
\eta \ge 0, \xi \ge 0,  \quad \quad   \lambda \ge 0.
\eeq
Again in this basis, we have that $L(t) = P (I + t J)^{-1} J P^{-1}$ where, respectively, 
\beqs
(I + t J)^{-1} J  &=& 
\left( \begin{array}{ccc} \frac{\alpha + t (\alpha^2 + \beta^2)}{(1 + t \alpha)^2 + t^2 \beta^2}   & \frac{\beta}{(1 + t \alpha)^2 + t^2 \beta^2} & 0 \\
\frac{-\beta}{(1 + t \alpha)^2 + t^2 \beta^2} & \frac{\alpha + t (\alpha^2 + \beta^2)}{(1 + t \alpha)^2 + t^2 \beta^2} & 0 \\ 0 & 0 & \frac{\gamma}{1 + t \gamma} \end{array} \right), 
\left( \begin{array}{ccc} \frac{\lambda_1}{1 + t \lambda_1} & 0  & 0 \\
0 & \frac{\lambda_2}{1 + t \lambda_2} & 0 \\ 0 & 0 & \frac{\lambda_3}{1 + t \lambda_3}\end{array} \right), 
\nonumber \\
& & \left( \begin{array}{ccc} \frac{\xi}{1 + t \xi } & \frac{1}{(1 + t \xi )^2} & 0 \\
0 & \frac{\xi}{1 + t \xi } & 0 \\ 0 & 0 & \frac{\eta}{1 + t \eta }  \end{array} \right), 
\left( \begin{array}{ccc} \frac{\lambda}{1 + t \lambda} &  \frac{1}{(1 + t \lambda)^2}  & \frac{-t}{(1 + t \lambda)^3}  \\
0 & \frac{\lambda}{1 + t \lambda}   & \frac{1}{(1 + t \lambda)^2}  \\ 0 & 0 & \frac{\lambda}{1 + t \lambda}   \end{array} \right)  \label{cases}
\eeqs
The first matrix in (\ref{cases})  with $\gamma > 0$ gives (\ref{T1E1}), and with $\gamma = 0$ gives (\ref{T1E2}).
To see the latter, note first that 
\beq
\lim_{t \to \infty} P \left( \begin{array}{ccc} \frac{t \alpha + t^2 (\alpha^2 + \beta^2)}{(1 + t \alpha)^2 + t^2 \beta^2}   & \frac{t\beta}{(1 + t \alpha)^2 + t^2 \beta^2} & 0 \\
\frac{-t\beta}{(1 + t \alpha)^2 + t^2 \beta^2} & \frac{t\alpha + t^2 (\alpha^2 + \beta^2)}{(1 + t \alpha)^2 + t^2 \beta^2} & 0 \\ 0 & 0 & 0\end{array} \right) P^{-1} = P (e_1 \otimes e_1 + e_2 \otimes e_2) P^{-1}
= I  - P e_3 \otimes P^{-T} e_3. \label{cd}
\eeq
Let $a = P e_3$ and $n = P^{-T} e_3$, so that $a \cdot n = 1$.  Note that any $a$ and $n$ satisfying
$a \cdot n = 1$ are possible by choosing (the invertible) $P = a \otimes e_3 + n_1^{\perp} \otimes e_2 + n_2^{\perp} \otimes e_2$, where  $n_1^{\perp} , n_2^{\perp}$ are orthonormal and perpendicular to $n$.  The required basis
can be taken as the orthonormal basis $\hat{e}_1, \hat{e}_2, n/|n|$ where $n \cdot \hat{e}_2 = 
a \cdot \hat{e}_2= 0$.  In 
this basis $a = (-K/\nu, 0, 1/\nu)$ and $n = (0,0,\nu)$, $\nu \ne 0, K \in \mathbb R$, which
gives (\ref{T1E2}).

Consider now the second matrix in (\ref{cases}).  If $ \lambda_1 \lambda_2 \lambda_3>0$ then
we get the right hand expression in (\ref{T1E1}).  If exactly one of $\lambda_1, \lambda_2, \lambda_3$
vanishes, say $\lambda_3 = 0$, we get the expression on the right of (\ref{cd}) multiplying $1/t$,
and we recover case (\ref{T1E2}). If exactly two of $\lambda_1, \lambda_2, \lambda_3$
vanish, say $\lambda_1 = \lambda_2 = 0$, then
\beq
\lim_{t \to \infty} P \left( \begin{array}{ccc}0 & 0& 0 \\
0 & 0 & 0 \\ 0 & 0 &  \frac{t \lambda_3}{1 + t \lambda_3}  \end{array} \right) P^{-1} = P e_3 \otimes e_3 P^{-1}
= P e_3 \otimes P^{-T} e_3. \label{pd}
\eeq
As above, we write $a = P e_3, n = P^{-T}e_3$, so $a \cdot n = 1$, and, as above, choose an orthonormal basis
in which $a = (0, K \nu, \nu), n = (0, 0, 1/\nu)$, $K \in \mathbb R, \nu \ne 0$.  This gives (\ref{T1E3}).

Consider now the third matrix in (\ref{cases}).  If $\xi = \eta = 0$, we get immediately (\ref{T1E5}).
If $\xi = 0$ but $\eta \ne 0$, we have
\beqs
P (I + t J)^{-1} J P^{-1} &=& 
P \left( \begin{array}{ccc} 0  & 1 & 0\\
0 & 0 & 0 \\ 0 & 0 & 0  \end{array} \right)P^{-1}
+ \frac{1}{t} P \left( \begin{array}{ccc} 0  & 0 & 0  \\
0 & 0 & 0 \\ 0 & 0 & \frac{t \eta }{1 + t \eta }  \end{array} \right)P^{-1} \nonumber \\
&=& P e_1 \otimes P^{-T} e_2  + \frac{1}{t} P e_3 \otimes P^{-T} e_3 + O\bigg( \frac{1}{t^2} \bigg) .
\label{new}
\eeqs
Let $a = P e_1, n = P^{-T}e_2, b = P e_3, m = P^{-T}e_3$.   These are restricted by the 
necessary conditions 
\beq
a \cdot n = 0,\ b \cdot m = 1,\ a \cdot m = 0,\ b \cdot n = 0,\ m \nparallel n,\ a \nparallel b. \label{ab}
\eeq
Choose the orthonormal basis $\hat{e}_1 =  a/|a|,  \hat{e}_2 = n/|n|, \hat{e}_3 = \hat{e}_1 \times \hat{e}_2$.
By scaling, we can assume without loss of generality that $a \cdot \hat{e}_1 = 1$ and $m \cdot \hat{e}_3 = 1$.
In this basis $a = (1,0,0), n = (0, K_2,0), b = (K_1, 0, 1), m = (0, K_3,1)$, $K_2 \ne 0$.
The conditions (\ref{ab}) are necessary and sufficient that the first two terms on the right hand side of (\ref{new}) are 
$a \otimes n + (1/t) b \otimes m$, as can be verified by the choice $P =  \hat{e}_1 \otimes {e}_1 + 
(1/K_2)(\hat{e}_2 - K_3 \hat{e}_3) \otimes {e}_2 + (K_1 \hat{e}_1 + \hat{e}_3) \otimes {e}_3$,
which is invertible.  The basis $\hat{e}_1, \hat{e}_2, \hat{e}_3$ is the required basis, and the result
is given in (\ref{T1E6}).

Still considering the third matrix in (\ref{cases}), assume $\xi \ne  0$ and $\eta \ne 0$.  We have
\beq
\lim_{t \to \infty} P \left( \begin{array}{ccc} \frac{t \xi}{1 + t \xi } & \frac{t}{(1 + t \xi )^2} & 0 \\
0 & \frac{t \xi}{1 + t \xi } & 0 \\ 0 & 0 & \frac{t \eta}{1 + t \eta }  \end{array} \right)
= I ,
\eeq
and we recover case (\ref{T1E1}).

In cases  (\ref{T1E1})-(\ref{T1E3}) and  (\ref{T1E6}) the error is clearly bounded by $const/t^2$.

Finally, consider the last matrix in (\ref{cases}).  If $\lambda \ne 0$ we recover case (\ref{T1E1}).
 If $\lambda = 0$, we have
\beqs
P (I + t J)^{-1} J P^{-1} &=& 
P \left( \begin{array}{ccc} 0  & 1 & -t \\
0 & 0 & 1 \\ 0 & 0 & 0  \end{array} \right)P^{-1}\nonumber \\
&=& P e_1 \otimes P^{-T} e_2  +( P e_2 - t  P e_1 ) \otimes P^{-T} e_3 .
\eeqs
Let $a = P e_1, n = P^{-T} e_2, b = P e_2, m = P^{-T} e_3$.   We have the necessary conditions
\beq
a \cdot n = 0,\ b \cdot m = 0,\ a \cdot m = 0,\ b \cdot n = 1,\ m \nparallel n,\ a \nparallel b. \label{ab1}
\eeq
Choose the basis $\hat{e}_1 = a/|a|, \hat{e}_3 = n/|n|, \hat{e}_2 = \hat{e}_3 \times \hat{e}_1$.
By scaling, assume that $n \cdot \hat{e}_2 = 1$.   In this basis $a = (K_3,0,0), 
n = (0,1,K_4), b = (K_5, 1, 0), m = (0,0, K_1)$, $K_1 K_3 \ne 0$.  The corresponding 
$P = K_3 \hat{e}_1 \otimes e_1 + (K_5 \hat{e}_1 + \hat{e}_2) \otimes e_2 +(1/K_1) (\hat{e}_3- K_4 \hat{e}_2) \otimes e_3$ is 
invertible.  Putting $K_2 = K_3 K_4 - K_1 K_5 $, we get case (\ref{T1E4}).

\begin{remark}
It is possible to obtain a more extensive classification of the homoenergetic flows
if $\det (I + t^* A) =0$ at some $t^*>0$. In that case $L\left(
t\right)  $ blows-up at $t^*$ and the behavior of $L(t)$
can then be read off from (\ref{dets}), (\ref{cases})
in the proof of 
Theorem \ref{ClassHomEne}.  Nevertheless, in
this paper we will restrict our analysis only to the cases in which $L\left(
t\right)  $ is globally defined in time.
\end{remark}

\section{General properties of homoenergetic flows}

\label{ss:genprop}

\bigskip

\subsection{Well posedness theory for homoenergetic flows of Maxwell molecules}

\bigskip

We first prove using standard arguments for the Boltzmann equation that the
homoenergetic flows with the form (\ref{B1_0}), (\ref{B4_0}), (\ref{B3_0})
exist for a large class of initial data $g_{0}\left(  w\right)  .$ This
question has been considered in \cite{CercArchive}, \cite{Cerc2000}. However,
the approach used in those papers is based on the $L^{1}$ theory for Boltzmann
equations (cf. \cite{Cerc69}), and it will be more convenient for the type of arguments used in
this paper to consider homoenergetic flows in the class of Radon measures. On
the other hand, the analysis in \cite{CercArchive}, \cite{Cerc2000} is
restricted to the case of simple shear (cf., (\ref{T1E5})) and we will study
more general classes of homoenergetic flows. Moreover, in some cases we will
need to consider equations with additional terms which are due to rescalings
of the solutions. For this reason we formulate here a well posedness theorem
for a family of Boltzmann equations with the degree of generality that we will
require. The class of equations that we will consider is the following one
\begin{align}
\partial_{t}G-\partial_{w}\cdot\left(  \left[  Q(t)w\right]  G\right)   &
=\mathbb{C}G\left(  w\right) \label{eq:boltzgen}\\
\mathbb{C}G\left(  w\right)   &  =\int_{\mathbb{R}^{3}}dw_{\ast}\int_{S^{2}%
}d\omega B\left(  n\cdot\omega,\left\vert w-w_{\ast}\right\vert \right)
\left[  G^{\prime}G_{\ast}^{\prime}-G_{\ast}G\right]
\ ,\label{eq:collboltzgen}\\
G\left(  0,w\right)   &  =G_{0}\left(  w\right),  \ \label{eq:invalBol}%
\end{align}
where
\begin{equation}
Q(\cdot)\in C^{1}\left(  \left[  0,\infty\right)  ;M_{3\times3}\left(
\mathbb{R}\right)  \right)  \ \ ,\ \ \left\Vert Q(t)\right\Vert \leq
c_{1}+c_{2}t,\text{ with }c_{1}>0,\ c_{2}>0. \label{T3E2}%
\end{equation}
with the norm $\left\Vert \cdot\right\Vert $ defined in (\ref{T3E1}). We will
assume also in the following that the function
\begin{equation}
\Lambda\left(  w,w_{\ast}\right)  =\int_{S^{2}}B\left(  n\cdot\omega
,\left\vert w-w_{\ast}\right\vert \right)  d\omega\ \ ,\ \ n=\frac{\left(
w-w_{\ast}\right)  }{\left\vert w-w_{\ast}\right\vert } \label{T3E2c}%
\end{equation}
satisfies
\begin{equation}
\Lambda\text{ is continuous and}\ 0\leq\Lambda\left(  w,w_{\ast}\right)  \leq
c_{3}\ \ \text{with }c_{3}>0. \label{T3E2a}%
\end{equation}

We will prove well posedness results for the collision kernel
associated to Maxwell molecules (or more generally Maxwell
pseudomolecules in the notation of \cite{TM}). In this case, the kernel $B$ is
homogeneous of degree zero in $\left\vert w-w_{\ast}\right\vert $ 
(e.g., the homogeneity parameter $\gamma = 0$) and we then
have $B\left(  n\cdot\omega,\left\vert w-w_{\ast}\right\vert \right)
=B\left(  n\cdot\omega\right)  .$ This restriction to Maxwell molecules is the reason we assume the stringent boundedness condition (\ref{T3E2a}).

We now introduce some definitions and notation. We denote by $\mathcal{M}%
_{+}\left(  \mathbb{R}_{c}^{3}\right)  $ the set of Radon measures in
$\mathbb{R}_{c}^{3}$. We denote as $\mathbb{R}_{c}^{3}$ the compactification
of $\mathbb{R}^{3}$ by means of a single point $\infty.$ This is  a
technical issue that we need in order to have convenient compactness
properties for some subsets of $\mathcal{M}_{+}\left(  \mathbb{R}_{c}%
^{3}\right)  .$ The space $C\left(  \left[  0,\infty\right)  :\mathcal{M}%
_{+}\left(  \mathbb{R}_{c}^{3}\right)  \right)  $ is defined endowing
$\mathcal{M}_{+}\left(  \mathbb{R}_{c}^{3}\right)  $ with the measure norm
\begin{equation}
\left\Vert \mu \right\Vert _{M}=\sup_{\varphi\in C\left(  \mathbb{R}_{c}%
^{3}\right)  :\left\Vert \varphi\right\Vert _{\infty}=1}\left\vert \mu\left(
\varphi\right)  \right\vert =\int_{\mathbb{R}_{c}^{3}}\left\vert
\mu\right\vert \left(  dw\right)  . \label{T3E2b}%
\end{equation}
We remark that this definition implies that the total measure of
$\mathbb{R}_{c}^{3}$ is finite if $\mu\in\mathcal{M}_{+}\left(  \mathbb{R}%
_{c}^{3}\right)  .$ Moreover $\varphi\in C\left(  \mathbb{R}_{c}^{3}\right)  $
implies that the limit value $\varphi\left(  \infty\right)  $ exists.

Given $G\in C\left(  \left[  0,\infty\right]  :\mathcal{M}_{+}\left(
\mathbb{R}_{c}^{3}\right)  \right)  $ we define%
\begin{equation}
\mathbb{A}\left[  G\right]  \left(  t,w\right)  =\int_{\mathbb{R}^{3}}%
dw_{\ast}\int_{S^{2}}d\omega B\left(  n\cdot\omega,\left\vert w-w_{\ast
}\right\vert \right)  G_{\ast}\left(  t,\cdot\right)  . \label{T3E4}%
\end{equation}
Given $h_{0}\in\mathcal{M}_{+}\left(  \mathbb{R}_{c}^{3}\right)  $ we will
denote as $S\left(  t;t_{0}\right)  ,\ t\geq t_{0}\geq0,$ the operator
$S_{G}\left(  t;t_{0}\right)  :\mathcal{M}_{+}\left(  \mathbb{R}_{c}%
^{3}\right)  \rightarrow\mathcal{M}_{+}\left(  \mathbb{R}_{c}^{3}\right)  $
defined by means of
\begin{align}
\partial_{t}h-\partial_{w}\cdot\left(  \left[  Q(t)w\right]  h\right)   &
=-\mathbb{A}\left[  G\right]  \left(  t,w\right)  h\ \ ,\ \ h\left(
t_{0},\cdot\right)  =h_{0}\label{T3E3a}\\
h\left(  t,w\right)   &  =S_{G}\left(  t;t_{0}\right)  h_{0}. \label{T3E3}%
\end{align}
The operator $S_{G}\left(  t;t_{0}\right)  $ is well defined, since
(\ref{T3E3a}) can be solved explicitly using the method of characteristics
taking into account (\ref{T3E2}). The solution is given below in (\ref{mocsoln}). A relevant point is that $\mathbb{A}\left[
G\right]  \left(  w\right)  \geq0$ and as a consequence no divergences arise
from large values of $\left\vert w\right\vert $. We will use the following
concept of solutions of (\ref{eq:boltzgen})-(\ref{eq:invalBol}).

\begin{definition}
\label{mildSol}We will say that $G\in C\left(  \left[  0,\infty\right]
:\mathcal{M}_{+}\left(  \mathbb{R}_{c}^{3}\right)  \right)  $ is a mild
solution of (\ref{eq:boltzgen})-(\ref{eq:invalBol}) with initial value
$G\left(  0,\cdot\right)  =G_{0}\in\mathcal{M}_{+}\left(  \mathbb{R}_{c}%
^{3}\right)  $ if $G$ satisfies the following integral equation:%
\begin{equation}
G\left(  t,w\right)  =S_{G}\left(  t;0\right)  G_{0}\left(  w\right)
+\int_{0}^{t}S_{G}\left(  t;s\right)  \mathbb{C}^{(+)}G\left(  s,w\right)  ds,
\label{T3E5}%
\end{equation}
where the operator $S_{G}\left(  t;s\right)  $ is as in (\ref{T3E3}) and
\begin{align}
\mathbb{C}^{\left(  +\right)  }G\left(  w\right)   &  =\int_{\mathbb{R}^{3}%
}dw_{\ast}\int_{S^{2}}d\omega B\left(  n\cdot\omega,\left\vert w-w_{\ast
}\right\vert \right)  G^{\prime}G_{\ast}^{\prime},\label{T3E5a}\\
\mathbb{C}^{\left(  -\right)  }G\left(  w\right)   &  =G\int_{\mathbb{R}^{3}%
}dw_{\ast}\int_{S^{2}}d\omega B\left(  n\cdot\omega,\left\vert w-w_{\ast
}\right\vert \right)  G_{\ast}=\mathbb{A}\left[  G\right]  G. \label{T3E5b}%
\end{align}

\end{definition}

We emphasize that (\ref{T3E5}) must be understood as an identity in the sense
of measure, i.e., acting over an arbitrary test function $\varphi\in C\left(
\mathbb{R}_{c}^{3}\right)  .$ Note also that all the operators appearing in
(\ref{T3E5}) are well defined for $G\in C\left(  \left[  0,\infty\right]
:\mathcal{M}_{+}\left(  \mathbb{R}_{c}^{3}\right)  \right)  $ and that
$S_{G}\left(  t;s\right)  $ is a bounded operator from $\mathcal{M}_{+}\left(
\mathbb{R}_{c}^{3}\right)  $ to $\mathcal{M}_{+}\left(  \mathbb{R}_{c}%
^{3}\right)  $ for $0\leq s\leq t\leq T<\infty.$

\bigskip

\begin{theorem}
\label{WellPos}Suppose that $G_{0}\in\mathcal{M}_{+}\left(  \mathbb{R}%
^{3}\right)  $ satisfies
\[
\int_{\mathbb{R}^{3}}G_{0}\left(  dw\right)  <\infty,
\]
Then, there exists a unique mild solution $G\in C\left(  \left[
0,\infty\right)  :\mathcal{M}_{+}\left(  \mathbb{R}_{c}^{3}\right)  \right)  $
in the sense of Definition \ref{mildSol} to the initial value problem
(\ref{eq:boltzgen})-(\ref{eq:invalBol}) with $Q$ and $B$ satisfying
(\ref{T3E2}), (\ref{T3E2a}). Moreover, the problem (\ref{eq:boltzgen}%
)-(\ref{eq:invalBol}) is satisfied in the sense of measures.
\end{theorem}

\bigskip

\begin{remark}
Notice that since $B$ is continuous and it satisfies (\ref{T3E2a}) we have
that $\mathbb{C}^{\left(  +\right)  }G\left(  w\right)  $ and $\mathbb{C}%
^{\left(  -\right)  }G\left(  w\right)  $ in (\ref{T3E5a}), (\ref{T3E5b})
define measures in $\mathcal{M}_{+}\left(  \mathbb{R}_{c}^{3}\right)  $ and it
makes sense to say that the equation (\ref{eq:boltzgen})-(\ref{eq:invalBol})
is satisfied in the sense of measures. The term $-\partial_{w}\cdot\left(
\left[  Q(t)w\right]  G\right)  $ is understood integrating by parts and
passing the derivative to the test function $\varphi.$ The only difference
between solutions in the sense of measures and the {\it weak solutions} defined in
 Definition \ref{WeakSol} below is that in this second case we write the
collision kernel in a symmetrized form which will be convenient in forthcoming computations.
\end{remark}

We introduce now the concept of weak solution of (\ref{eq:boltzgen}%
)-(\ref{eq:invalBol}) which will be also needed later.

\bigskip

\begin{definition}
\label{WeakSol}We will say that $G\in C\left(  \left[  0,\infty\right)
:\mathcal{M}_{+}\left(  \mathbb{R}_{c}^{3}\right)  \right)  $ is a weak
solution of (\ref{eq:boltzgen})-(\ref{eq:invalBol}) with initial value
$G\left(  0,\cdot\right)  =G_{0}\in\mathcal{M}_{+}\left(  \mathbb{R}_{c}%
^{3}\right)  $ if for any $T\in\left(  0,\infty\right)  $ and any test
function $\varphi\in C\left(  \left[  0,T\right)  :C^{1}\left(  \mathbb{R}%
_{c}^{3}\right)  \right)  $ the following identity holds
 \begin{align}\label{T3E6}
&  \int_{\mathbb{R}^{3}}\varphi\left(  T,w\right)  G\left(  T,dw\right)
-\int_{\mathbb{R}^{3}}\varphi\left(  0,w\right)  G_{0}\left(  dw\right)
\\&  
=-\int_{0}^{T}dt\int_{\mathbb{R}^{3}}\int\partial_{t}\varphi G\left(
t,dw\right)  -\int_{0}^{T}dt\int_{\mathbb{R}^{3}}\left[  Q(t)w\cdot
\partial_{w}\varphi\right]  G\left(  t,dw\right)  \nonumber\\
& \quad +\frac{1}{2}\int_{0}^{T}dt\int_{\mathbb{R}^{3}}\int_{\mathbb{R}^{3}}%
\int_{S^{2}}d\omega G\left(  t,dw\right)  G\left(  t,dw_{\ast}\right)
B\left(  n\cdot\omega,\left\vert w-w_{\ast}\right\vert \right) 
\Big[\varphi\left(  t,w^{\prime}\right)  +\varphi\left(  t,w_{\ast}^{\prime
}\right)\nonumber \vspace{5mm}\\&  \qquad -\varphi\left(  t,w\right)  -\varphi\left(  t,w_{\ast}\right)
\Big]. \nonumber
\end{align}
\end{definition}

We will use repeatedly the following norms:
\begin{equation}
\left\Vert G\right\Vert _{1,s}=\int_{\mathbb{R}^{3}}\left(  1+\left\vert
w\right\vert ^{s}\right)  G\left(  dw\right)  \ \ \text{for\ }G\in
\mathcal{M}_{+}\left(  \mathbb{R}_{c}^{3}\right)  \text{\ },\ \ s>0
\label{NormS}%
\end{equation}

\begin{theorem}
\label{WeakSolTh}Suppose that $G\in C\left(  \left[  0,\infty\right)
:\mathcal{M}_{+}\left(  \mathbb{R}_{c}^{3}\right)  \right)  $ is a mild
solution of (\ref{eq:boltzgen})-(\ref{eq:invalBol}) with $Q$ and $B$ satisfying
(\ref{T3E2}), (\ref{T3E2a}) and initial value
$G\left(  0,\cdot\right)  =G_{0}\in\mathcal{M}_{+}\left(  \mathbb{R}_{c}%
^{3}\right)  .$ Then $G$ is also a weak solution of (\ref{eq:boltzgen}%
)-(\ref{eq:invalBol}) in the sense of Definition \ref{WellPos}.
Suppose that in addition $G_{0}$ satisfies
\begin{equation}
\left\Vert G_{0}\right\Vert _{1,s}<\infty\label{T4E2}%
\end{equation}
for some $s>0.$ Then the mild solution of (\ref{eq:boltzgen}%
)-(\ref{eq:invalBol}) satisfies
\begin{equation}
\sup_{0\leq t\leq T}\left\Vert G\left(  t,\cdot\right)  \right\Vert
_{1,s}<C\left(  T,\left\Vert G_{0}\right\Vert _{M}\right)  \left\Vert
G_{0}\right\Vert _{1,s}<\infty\label{T4E3}%
\end{equation}
for any $T\in\left(  0,\infty\right)  .$ Moreover, if $s>2$ the identity
(\ref{T3E6}) is satisfied for any test function $\varphi\in C\left(  \left[
0,T\right]  :C^{1}\left(  \mathbb{R}_{c}^{3}\right)  \right)  $ such that
$\left\vert \varphi\left(  w\right)  \right\vert +\left\vert w\right\vert
\left\vert \nabla_{v}\varphi\left(  w\right)  \right\vert \leq C_{0}\left(
1+\left\vert w\right\vert ^{2}\right)  $.
\end{theorem}

\begin{proofof}
[Proof of Theorem \ref{WellPos}]Given $T\in\left(  0,\infty\right)  ,$ we
define an operator $\mathcal{T}_{T}:C\left(  \left[  0,T\right]
:\mathcal{M}_{+}\left(  \mathbb{R}_{c}^{3}\right)  \right)  \rightarrow
C\left(  \left[  0,T\right]  :\mathcal{M}_{+}\left(  \mathbb{R}_{c}%
^{3}\right)  \right)  $ by means of
\[
\mathcal{T}_{T}\left[  G\right]  \left(  t,w\right)  =S_{G}\left(  t;0\right)
G_{0}\left(  w\right)  +\int_{0}^{t}S_{G}\left(  t;s\right)  \mathbb{C}%
^{(+)}G\left(  s,w\right)  ds,\ \ 0\leq t\leq T
\]
where $S_{G}\left(  t;s\right)  $ is as in (\ref{T3E3}). Then, due to
(\ref{T3E5}) the proof of the Theorem reduces to proving existence and
uniqueness of solutions for the fixed point problem
\[
G=\mathcal{T}_{T}\left[  G\right]  .
\]

We prove that the operator $\mathcal{T}_{T}$ is contractive if $T>0$ is
sufficiently small. To this end we prove the following estimates:
\begin{align}
\left\Vert \mathbb{C}^{(+)}G\right\Vert _{M}  &  \leq C\left\Vert G\right\Vert
_{M}^{2}\label{T3E7}\\
\left\Vert \mathbb{C}^{(+)}G_{1}-\mathbb{C}^{(+)}G_{2}\right\Vert _{M}  &
\leq C\left(  \left\Vert G_{1}\right\Vert _{M}+\left\Vert G_{2}\right\Vert
_{M}\right)  \left\Vert G_{1}-G_{2}\right\Vert _{M} \label{T3E8}%
\end{align}
where the norm $\left\Vert \cdot\right\Vert _{M}$ is as in (\ref{T3E2b}) and
we have used (\ref{T3E2a}) as well as the fact that the mapping $\left(
w,w_{\ast}\right)  \rightarrow\left(  w^{\prime},w_{\ast}^{\prime}\right)  $
is bijective and that the symplectic identity $dwdw_{\ast}=dw^{\prime}dw_{\ast
}^{\prime}$ holds (cf., (\ref{CM1}), (\ref{CM2})). We define $\mathcal{A}%
\subset C\left(  \left[  0,T\right]  :\mathcal{M}_{+}\left(  \mathbb{R}%
_{c}^{3}\right)  \right)  $ by%
\[
\mathcal{A}=\left\{  G\in C\left(  \left[  0,T\right]  :\mathcal{M}_{+}\left(
\mathbb{R}_{c}^{3}\right)  \right)  :\left\Vert G\left(  \cdot,t\right)
\right\Vert _{M}\leq2\left\Vert G_{0}\right\Vert _{M}\text{ \ for }t\in\left[
0,T\right]  \right\}  .
\]

On the other hand we have the following estimates:%
\begin{equation}
\left\Vert S_{G}\left(  t;s\right)  \right\Vert \leq1\ \ \text{for }0\leq
s\leq t<\infty, \label{T3E9}%
\end{equation}
where we denote the norm on the space
$\mathcal{L}\left(  \mathcal{M}_{+}\left(  \mathbb{R}_{c}^{3}\right)
,\mathcal{M}_{+}\left(  \mathbb{R}_{c}^{3}\right)  \right) $ by $\left\Vert \cdot\right\Vert $. Moreover we
have%
\begin{equation}
\left\Vert S_{G_{1}}\left(  t;s\right)  -S_{G_{2}}\left(  t;s\right)
\right\Vert \leq CT\left\Vert G_{1}-G_{2}\right\Vert _{T},\ \text{ }0\leq
s\leq t\leq T. \label{T3E9a}%
\end{equation}
Here we used the norm $\left\Vert \cdot\right\Vert _{T}$ given by:%
\[
\left\Vert G\right\Vert _{T}=\sup_{t\in\left[  0,T\right]  }\left\Vert
G\left(  \cdot,t\right)  \right\Vert _{M}.
\]

The estimate (\ref{T3E9}) follows by integrating (\ref{T3E3a}) and using 
$\mathbb{A}\left[  G\right]  \left(  t,w\right)  \geq0.$ On the other hand,
(\ref{T3E9a}) can be proved using the representation formula for $S_{G}\left(
t;s\right)  $:%
\begin{equation}
S_{G}\left(  t;s\right)  h_{0}\left(  w\right)  =\exp\left(  -\int_{s}%
^{t}\mathbb{A}\left[  G\right]  \left(  \xi,U\left(  t,\xi\right)  w\right)
d\xi\right)  \exp\left(  \int_{s}^{t}\mathop{\rm tr}\left(  Q\left(
\xi\right)  \right)  d\xi\right)  h_{0}\left(  U\left(  t;s\right)  w\right) \label{mocsoln}
\end{equation}
where:%
\[
\frac{\partial\left[  U\left(  s;t\right)  w\right]  }{\partial t}%
=-Q(t)U\left(  s;t\right)  w\ \ ,\ \ U\left(  s;s\right)  w=w\in\mathbb{R}^{3}
.
\]
Then, (\ref{T3E9a}) follows using that $\left\vert e^{-x_{1}}-e^{-x_{2}%
}\right\vert \leq\left\vert x_{1}-x_{2}\right\vert $, and $\sup_{w\in
\mathbb{R}^{3}}\left\vert \mathbb{A}\left[  G_{1}\right]  \left(  w\right)
-\mathbb{A}\left[  G_{2}\right]  \left(  w\right)  \right\vert \leq
C\left\Vert G_{1}-G_{2}\right\Vert _{M}$ yields the estimate
\begin{align*}
&  \left\Vert \left[  S_{G_{1}}\left(  t;s\right)  -S_{G_{2}}\left(
t;s\right)  \right]  h_{0}\right\Vert _{M}\\
&  \leq CT\left\Vert G_{1}-G_{2}\right\Vert _{T}\int_{\mathbb{R}^{3}}%
h_{0}\left(  U\left(  t,s\right)  w\right)  dw\\
&  \leq CT\left\Vert G_{1}-G_{2}\right\Vert _{T}\int_{\mathbb{R}^{3}}%
h_{0}\left(  w\right)  \mathop{\rm Jac}\left(  U\left(  s,t\right)  \right)
dw\\
&  \leq\tilde{C}T\left\Vert G_{1}-G_{2}\right\Vert _{T}\left\Vert
h_{0}\right\Vert _{M} .
\end{align*}

Combining (\ref{T3E7}), (\ref{T3E8}), (\ref{T3E9}), (\ref{T3E9a}) we obtain
that
\begin{equation}
\left\Vert \mathcal{T}_{T}\left[  G\right]  \right\Vert _{T}\leq\left\Vert
G_{0}\right\Vert _{M}+4C\left\Vert G_{0}\right\Vert _{M}^{2}T \label{T4E1}%
\end{equation}
for any $G\in\mathcal{A}$. We have also
\[
\left\Vert \mathcal{T}_{T}\left[  G_{1}\right]  -\mathcal{T}_{T}\left[
G_{2}\right]  \right\Vert _{T}\leq C\left\Vert G_{0}\right\Vert _{M}%
T\left\Vert G_{1}-G_{2}\right\Vert _{T} .
\]

Therefore, the operator $\mathcal{T}_{T}\left[  G\right]  $ is contractive in the
space $\mathcal{A}$ with a metric given by the norm $\left\Vert \cdot
\right\Vert _{T}$ if $T$ is sufficiently small. This implies the existence of
a mild solution in the time interval $\left[  0,T\right]  $ for $T$
sufficiently small. Notice that the fact that $G$ is nonnegative follows
immediately due to the choice of the space of functions $\mathcal{A}$.

Applying the differential operator $\partial_{t}-\partial_{w}\cdot\left[
Q(t)w\right]  $ to (\ref{T3E5}) we obtain, using (\ref{T3E3a}), (\ref{T3E3})
that the following identity holds in the sense of measures (i.e. the whole
expression is understood using a test function $\varphi
=\varphi\left(  w\right)  $):
\begin{equation}
\partial_{t}G\left(  t,w\right)  -\partial_{w}\cdot\left[  Q(t)wG\left(
t,w\right)  \right]  =\mathbb{C}^{(+)}G\left(  t,w\right)  -\mathbb{A}\left[
G\right]  \left(  t,w\right)  G\left(  t,w\right)  \label{T4E6}%
\end{equation}
whence $G$ satisfies (\ref{eq:boltzgen})-(\ref{eq:invalBol}) in the sense of
measures. Integrating (\ref{T4E6}) we obtain
\begin{equation}
\left\Vert G\left(  \cdot,t\right)  \right\Vert _{M}=\left\Vert G_{0}%
\right\Vert _{M}\text{ \ for }0\leq t\leq T . \label{T4E6a}%
\end{equation}

Therefore, using a similar argument we can extend the solution to an interval
$\left[  T,T+\delta\right]  $ and iterating we then obtain a global solution
defined for $0\leq t<\infty.$ Notice that the constants $C$ above depend on
the time $T$ where we start the iteration argument due to the fact that
$\left\Vert Q\left(  t\right)  \right\Vert $ can increase as $t\rightarrow
\infty,$ but this norm is bounded for any finite interval $0\leq t\leq T$ and
therefore we can prove global existence.
\end{proofof}

\smallskip
\begin{proofof}
[Proof of Theorem \ref{WeakSolTh}]Multiplying (\ref{T4E6}) by a test function
$\varphi\left(  t,w\right)  $ integrating by parts and using a standard
symmetrization argument on the right-hand side of (\ref{T4E6}) (cf.
\cite{CIP}, \cite{TM}) and integrating in $t\in\left[  0,T\right]  $ we obtain
(\ref{T3E6}).

We now prove that under the assumption (\ref{T4E2}) the solution obtained in
Theorem \ref{WellPos} satisfies (\ref{T4E3}). Using the symplectic formula
$dw^{\prime}dw_{\ast}^{\prime}=dwdw_{\ast}$ (cf. (\ref{CM1}), (\ref{CM2})) and
(\ref{T3E5a}) we obtain
\begin{equation}
\left\Vert \mathbb{C}^{(+)}G\right\Vert _{1,s}\leq C\left\Vert G\right\Vert
_{1,s}\left\Vert G\right\Vert _{M}\ \ ,\ \ G\in\mathcal{M}_{+}\left(
\mathbb{R}_{c}^{3}\right)  . \label{T4E4}%
\end{equation}
where $s>0.$ On the other hand we claim that
\begin{equation}
\left\Vert S_{G}\left(  t;s\right)  G\right\Vert _{1,s}\leq C_{s}\left(
T\right)  \left\Vert G\right\Vert _{1,s}\ \ ,\ \ 0\leq s\leq t\leq
T<\infty\ ,\ \ G\in\mathcal{M}_{+}\left(  \mathbb{R}_{c}^{3}\right)  .
\label{T4E5}%
\end{equation}

This estimate follows by multiplying (\ref{T3E3a}) by the test functions $1$ and
$\left\vert w\right\vert ^{s},$ using that $\mathbb{A}\left[  G\right]
\left(  t,w\right)  \geq0$, integrating over $\mathbb{R}^{3}$ and using a
Gronwall-type argument. Then, estimating $\left\Vert G\left(  t,\cdot\right)
\right\Vert _{1,s}$ in (\ref{T3E5}) and using (\ref{T4E4}), (\ref{T4E5}) as
well as the mass conservation property (\ref{T4E6a}) we obtain
\[
\left\Vert G\left(  t,\cdot\right)  \right\Vert _{1,s}\leq C_{s}\left(
T\right)  \left\Vert G_{0}\right\Vert _{1,s}+C_{s}\left(  T\right)  \left\Vert
G_{0}\right\Vert _{M}\int_{0}^{t}\left\Vert G\left(  s,\cdot\right)
\right\Vert _{1,s}ds\ \ ,\ \ 0\leq t\leq T
\]
whence (\ref{T4E3}) follows using Gronwall's Lemma.

The fact that the identity (\ref{T3E6}) in Definition \ref{WeakSol} is
satisfied for test functions $\varphi$ bounded by a quadratic function as
stated in the statement of Theorem \ref{WeakSolTh} then follows by approximating
the test function $\varphi$ by a sequence of test functions $\varphi_{n}\in
C\left(  \mathbb{R}_{c}^{3}\right)  $ and using the fact that (\ref{T4E3})
implies that the contribution of the integrals due to the sets with
$\left\vert w\right\vert \geq R$ tends to zero as $R\rightarrow\infty. $
\end{proofof}

\begin{remark}
\label{AppSol}Suppose that $G_{0}$ satisfies $\left\Vert G_{0}\right\Vert
_{1,s}<\infty,$ and let $G$ be the corresponding solution of
(\ref{eq:boltzgen})-(\ref{eq:invalBol}) obtained in Theorem \ref{WeakSolTh}.
We can then obtain a sequence $\left\{  G_{m}\right\}  _{m\in\mathbb{N}}$ of
solutions of (\ref{eq:boltzgen}), (\ref{eq:collboltzgen}) with initial data
$G_{m,0}$ satisfying $\left\Vert G_{m,0}\right\Vert _{1,\bar{s}}<\infty$ for
some $\bar{s}>s$ and such that $\sup_{t\in\left[  0,T\right]  }\left\Vert
G_{m}\left(  t,\cdot\right)  -G\left(  t,\cdot\right)  \right\Vert
_{1,s}\rightarrow0$ as $m\rightarrow\infty.$ Indeed, we define $G_{m,0}%
=G_{0}\chi_{\left\{  \left\vert w\right\vert \leq m\right\}  }.$ Then
$\left\Vert G_{m,0}\right\Vert _{1,\bar{s}}<\infty$ and by dominated
convergence $\left\Vert G_{m,0}-G_{0}\right\Vert _{1,s}\rightarrow0$ as
$m\rightarrow\infty.$ Using (\ref{T3E5}) with initial data $G_{m,0}$ and
$G_{0},$ taking the difference of the resulting equations and arguing as in
the proof of Theorem \ref{WeakSolTh} we obtain
\[
\left\Vert G_{m}\left(  t,\cdot\right)  -G\left(  t,\cdot\right)  \right\Vert
_{1,s}\leq\left\Vert G_{m,0}-G_{0}\right\Vert _{1,s}+C\int_{0}^{t}\left\Vert
G_{m}\left(  \bar{t},\cdot\right)  -G\left(  \bar{t},\cdot\right)  \right\Vert
_{1,s}d\bar{t}%
\]
whence the stated convergence follows using Gronwall.
\end{remark}

\begin{remark}
Well posedness Theorems analogous to Theorems \ref{WellPos} and
\ref{WeakSolTh} for more general collision kernels $B$ (in particular for
kernels with homogeneity $\gamma$ different from zero) can be proved adapting
the theory of homogeneous Boltzmann equations as described in \cite{CIP}. We
restricted to kernels satisfying (\ref{T3E2a}) since the theory is simpler and
are the only ones needed in the following.
\end{remark}

\bigskip

\subsection{Moment equations for Maxwell molecules}

A crucial fact that we use repeatedly in this paper is the fact that for
Maxwell molecules the tensor of second moments $M_{j,k}%
=\int_{\mathbb{R}^{3}}w_{j}w_{k}G\left(  t,dw\right)  $ satisfies a linear
system of equations if $G$ is a mild solution of (\ref{eq:boltzgen}%
)-(\ref{eq:invalBol}). 
In order to compute the evolution equations for
$M_{j,k}$ we will use (\ref{T3E6}) with the test functions $\varphi=w_{j}%
w_{k}.$ The resulting right-hand side can then be computed using suitable
tensorial properties of the Boltzmann equation acting over quadratic functions
which shall be collected in the following.

We will assume in the rest of this subsection that $B=B\left(  n\cdot
\omega,\left\vert w-w_{\ast}\right\vert \right)  =B\left(  n\cdot
\omega\right)  $ in \eqref{eq:boltzgen} is homogeneous of order zero in
$\left\vert w-w_{\ast}\right\vert .$ We will denote by $W=W\left(  u,v\right)
$ a bilinear form:%
\begin{equation}
W:\mathbb{R}^{3}\times\mathbb{R}^{3}\rightarrow\mathbb{R}. \label{T4E7}%
\end{equation}
In order to simplify the notation we will write the quadratic form associated
to this bilinear form by $W\left(  v\right)  $ instead of $W\left(
v,v\right)  .$ We first prove the following lemma which allows us to transform
dependence on two vectors to dependence on just one vector.

\begin{lemma}
\label{BilRep}Suppose that $B=B\left(  n\cdot\omega,\left\vert w-w_{\ast
}\right\vert \right)  =B\left(  n\cdot\omega\right)  $ in \eqref{eq:boltzgen}
is homogeneous of order zero in $\left\vert w-w_{\ast}\right\vert $ and $W$ is
any bilinear form as in (\ref{T4E7}). Then
\begin{align}
\mathcal{Q}_{W}\left(  w,w_{\ast}\right)   &  =\frac{1}{2}\int_{S^{2}}d\omega
B\left(  n\cdot\omega\right)  \left[  W\left(  w^{\prime}\right)  +W\left(
w_{\ast}^{\prime}\right)  -W\left(  w\right)  -W\left(  w_{\ast}\right)
\right]  =\mathcal{\tilde{Q}}_{W}\left(  w-w_{\ast}\right) \label{S5E6}\\
&  :=-\int_{S^{2}}d\omega B\left(  \frac{\omega\cdot\left(  w-w_{\ast}\right)
}{\left\vert w-w_{\ast}\right\vert }\right)  \left[  W\left(  P_{\omega}%
^{\bot}\left(  w-w_{\ast}\right)  ,P_{\omega}\left(  w-w_{\ast}\right)
\right)  \right]  \ \label{S5E7}%
\end{align}
where for each $\omega\in S^{2}$ we denote as $P_{\omega}$ and $P_{\omega
}^{\bot}$ the orthogonal projections in the subspaces
$\mathop{\rm span}\left\{  \omega\right\}  $ and $\mathop{\rm span}\left\{
\omega\right\}  ^{\bot} $\ respectively, i.e.,%
\begin{equation}
P_{\omega}v=\left(  v\cdot\omega\right)  \omega\ ,\ \ P_{\omega}^{\bot
}v=v-\left(  v\cdot\omega\right)  \omega\ \ \text{for }v\in\mathbb{R}^{3}.
\label{S5E8}%
\end{equation}

\end{lemma}

\begin{proof}
Using the collision rule (\ref{CM1}), (\ref{CM2}) we can write $w^{\prime
}=w+\eta,\ w_{\ast}^{\prime}=w_{\ast}-\eta$ with $\eta=\left(  \left(
w_{\ast}-w\right)  \cdot\omega\right)  \omega=P_{\omega}\left(  w_{\ast
}-w\right)  .$ Then
\begin{align*}
&  W\left(  w^{\prime}\right)  +W\left(  w_{\ast}^{\prime}\right)  -W\left(
w\right)  -W\left(  w_{\ast}\right) \\
&  =W\left(  w+\eta,w+\eta\right)  +W\left(  w_{\ast}-\eta,w_{\ast}%
-\eta\right)  -W\left(  w,w\right)  -W\left(  w_{\ast},w_{\ast}\right) \\
&  =2\left[  W\left(  w,\eta\right)  -W\left(  w_{\ast},\eta\right)  +W\left(
\eta,\eta\right)  \right]  =2\left[  -W\left(  w_{\ast}-w,\eta\right)
+W\left(  \left(  \left(  w_{\ast}-w\right)  \cdot\omega\right)  \omega
,\eta\right)  \right] \\
&  =-2W\left(  P_{\omega}^{\bot}\left(  w_{\ast}-w\right)  ,P_{\omega}\left(
w_{\ast}-w\right)  \right)
\end{align*}
whence the Lemma follows.
\end{proof}

To quantify the moment equations for Maxwell molecules, we
introduce the following object:

\begin{equation}
Z\left(  v\right)  =\int_{S^{2}}d\omega B\left(  \frac{\omega\cdot
v}{\left\vert v\right\vert }\right)  \left[  P_{\omega}^{\bot}v\otimes
P_{\omega}v\right]  ,\ \ v\in\mathbb{R}^{3} \label{T4E8}%
\end{equation}
where we will understand $a\otimes b$ as a bilinear functional acting on
$\mathbb{R}^{3}\times\mathbb{R}^{3}$ in the following manner. Given two
vectors $w_{1},w_{2}\in\mathbb{R}^{3}$ we define
\begin{equation}
\left(  a\otimes b\right)  \left(  w_{1},w_{2}\right)  =\left(  a\cdot
w_{1}\right)  \left(  b\cdot w_{2}\right)  . \label{T4E8a}%
\end{equation}
Then $Z\left(  v\right)  $ as defined in (\ref{T4E8}) is a bilinear functional in $\mathbb{R}^{3}\times\mathbb{R}^{3}.$  
We have the following result.

\begin{lemma}
\label{Tensor}Suppose that $U\in SO\left(  3\right)  .$ Then, for any
$v\in\mathbb{R}^{3}$ the following identity holds:
\begin{equation}
Z\left(  Uv\right)  =\int_{S^{2}}d\omega B\left(  \frac{\omega\cdot
v}{\left\vert v\right\vert }\right)  \left[  \left(  UP_{\omega}^{\bot
}v\right)  \otimes\left(  UP_{\omega}v\right)  \right] 
 = U Z(v) U^T
 \label{S5E4}%
\end{equation}
where $Z$ is as in (\ref{T4E8}).
\end{lemma}

\begin{proof}
Using the definition of $Z$ we have
\[
Z\left(  Uv\right)  =\int_{S^{2}}d\omega B\left(  \frac{\omega\cdot
Uv}{\left\vert Uv\right\vert }\right)  \left(  P_{\omega}^{\bot}Uv\otimes
P_{\omega}Uv\right) .
\]
We now change variables as $\omega=U\hat{\omega},$ with $\hat{\omega}\in
S^{2}.$ Then:%
\[
Z\left(  Uv\right)  =\int_{S^{2}}d\hat{\omega}B\left(  \frac{U\hat{\omega
}\cdot Uv}{\left\vert Uv\right\vert }\right)  \left(  P_{U\hat{\omega}}^{\bot
}Uv\otimes P_{U\hat{\omega}}Uv\right) .
\]

An elementary geometrical argument, using the fact that orthogonal
transformations commute with the projection operators yields
\[
P_{U\hat{\omega}}Uv=UP_{\hat{\omega}}v\ ,\ P_{U\hat{\omega}}^{\bot}%
Uv=UP_{\hat{\omega}}^{\bot}v
\]
whence, using also $U\hat{\omega}\cdot Uv=\hat{\omega}\cdot v$,
\[
Z\left(  Uv\right)  =\int_{S^{2}}d\hat{\omega}B\left(  \frac{\hat{\omega}\cdot
v}{\left\vert v\right\vert }\right)  \left[  U\left(  P_{\hat{\omega}}^{\bot
}v\right)  \otimes U\left(  P_{\hat{\omega}}v\right)  \right]
\]
and the result follows.
\end{proof}

Lemma \ref{Tensor} implies that $Z$ defined by means of (\ref{T4E8}) is a
second order tensor under orthogonal transformations. We now compute a
suitable tensorial expression for $Z\left(  v\right) $ in a coordinate system
where this computation is particularly simple.

\begin{proposition}
\label{TensComp}The tensor $Z$ defined by means of (\ref{T4E8}) is given by
\begin{equation}
Z\left(  v\right)  =b\left[  v\otimes v-\frac{\left\vert v\right\vert ^{2}}%
{3}I\right]  \label{S5E9}%
\end{equation}
where:%
\begin{equation}
b=3\pi\int_{-1}^{1}B\left(  x\right)  x^{2}\left(  1-x^{2}\right)  dx>0.
\label{S6E1}%
\end{equation}
Moreover, suppose that we define
\begin{equation}
T_{j,k}=\frac{1}{2}\int_{S^{2}}d\omega B\left(  n\cdot\omega\right)  \left[
W_{j,k}\left(  w^{\prime}\right)  +W_{j,k}\left(  w_{\ast}^{\prime}\right)
-W_{j,k}\left(  w\right)  -W_{j,k}\left(  w_{\ast}\right)  \right]
\ \label{S5E1a}%
\end{equation}
where $W_{j,k}$ are the quadratic functions $W_{j,k}\left(  w,w\right)
=w_{j}w_{k}.$ Then
\begin{equation}
T_{j,k}=-b\left[  \left(  w-w_{\ast}\right)  _{j}\left(  w-w_{\ast}\right)
_{k}-\frac{\left\vert w-w_{\ast}\right\vert ^{2}}{3}\delta_{j,k}\right]
\ \ ,\ \ j,k=1,2,3 . \label{S6E2}%
\end{equation}

\end{proposition}

\begin{proof}
Suppose that $v\neq0,$ since otherwise $Z\left(  v\right)  =0.$ We then
compute $Z\left(  v\right)  $ in a very particular system of spherical
coordinates. More precisely, we take the direction of the North Pole in the
direction of $v$ and we denote as $\theta$ the angle of any vector with
respect to this direction and $\phi$ the azymuthal angle in a plane orthogonal
to $v$ with respect to any arbitrary direction in this plane. We construct an
orthonormal basis of $\mathbb{R}^{3}$ by means of $e_{1}=\frac{v}{\left\vert
v\right\vert }$ and choosing as $e_{2},\ e_{3}$ two orthonormal vectors
contained in the plane orthogonal to $e_{1}.$ Using this coordinate system we
can parametrize the sphere $S^{2}$ as
\begin{equation}
\omega=\left(
\begin{array}
[c]{c}%
\cos\left(  \theta\right) \\
\sin\left(  \theta\right)  \cos\left(  \phi\right) \\
\sin\left(  \theta\right)  \sin\left(  \phi\right)
\end{array}
\right)  ,\ \ \theta\in\left[  0,\pi\right]  ,\ \phi\in\left[  0,2\pi\right)
. \label{S5E5}%
\end{equation}
Then
\[
P_{\omega}v=\cos\left(  \theta\right)  \omega=\left(
\begin{array}
[c]{c}%
\cos^{2}\left(  \theta\right) \\
\sin\left(  \theta\right)  \cos\left(  \theta\right)  \cos\left(  \phi\right)
\\
\sin\left(  \theta\right)  \cos\left(  \theta\right)  \sin\left(  \phi\right)
\end{array}
\right)
\]
and using (\ref{S5E5}) and (\ref{S5E8}) we get
\[
P_{\omega}^{\bot}v=\left(
\begin{array}
[c]{c}%
\sin^{2}\left(  \theta\right) \\
-\sin\left(  \theta\right)  \cos\left(  \theta\right)  \cos\left(  \phi\right)
\\
-\sin\left(  \theta\right)  \cos\left(  \theta\right)  \sin\left(
\phi\right)
\end{array}
\right) .
\]
Then $P_{\omega}^{\bot}v\otimes P_{\omega}v$ can be represented by the matrix
\[
Y\left(  \theta,\phi\right)  =\left(
\begin{array}
[c]{ccc}%
\cos^{2}\theta\sin^{2}\theta & \cos\theta\cos\phi\sin^{3}\theta & \cos
\theta\sin^{3}\theta\sin\phi\\
-\cos^{3}\theta\cos\phi\sin\theta & -\cos^{2}\theta\cos^{2}\phi\sin^{2}\theta
& -\cos^{2}\theta\cos\phi\sin^{2}\theta\sin\phi\\
-\cos^{3}\theta\sin\theta\sin\phi & -\cos^{2}\theta\cos\phi\sin^{2}\theta
\sin\phi & -\cos^{2}\theta\sin^{2}\theta\sin^{2}\phi
\end{array}
\right) .
\]
Therefore the computation of $Z\left(  v\right)  $ in (\ref{T4E8}) reduces to
the computation of
\[
\int_{0}^{\pi}\sin\left(  \theta\right)  B\left(  \cos\left(  \theta\right)
\right)  \left[  \int_{0}^{2\pi}Y\left(  \theta,\phi\right)  d\phi\right]
d\theta.
\]

In the integration in the $\phi$ variable all the elements outside the
diagonal give zero. Using that in the coordinate system under consideration
$v=\left(  1,0,0\right)  ^{T}$ we then obtain
\[
Z\left(  v\right)  =b\left[  v\otimes v-\frac{\left\vert v\right\vert ^{2}}%
{3}I\right]
\]
where:%
\[
b=3\pi\int_{0}^{\pi}B\left(  \cos\left(  \theta\right)  \right)  \cos
^{2}\theta\sin^{3}\theta d\theta=3\pi\int_{-1}^{1}B\left(  x\right)
x^{2}\left(  1-x^{2}\right)  dx .
\]

Since both sides of (\ref{S5E9}) transform according to the law for the
transformation of second order tensors under orthogonal transformations (cf.
Lemma \ref{Tensor}), it then follows that this formula is valid in any
coordinate system.

Using the definition of $Z\left(  v\right)  $ and Lemma \ref{BilRep} we can
write the functions $\mathcal{Q}_{W_{j,k}}\left(  w,w_{\ast}\right)  $ with
$W_{j,k}\left(  w,w\right)  =w_{j}w_{k}$ as
\begin{equation}
T_{j,k}=\mathcal{Q}_{W_{j,k}}\left(  w,w_{\ast}\right)  =\mathcal{\tilde{Q}%
}_{W_{j,k}}\left(  w-w_{\ast}\right)  =-Z\left(  w-w_{\ast}\right)  \left(
e_{j},e_{k}\right)  \ \ ,\ \ j,k=1,2,3 \label{T4E8b}%
\end{equation}
where $\left\{  e_{1},e_{2},e_{3}\right\}  $ is the canonical orthonormal
basis of $\mathbb{R}^{3}.$ Finally (\ref{S6E2}) follows using (\ref{T4E8b}) as
well as the fact that $\left(  v\otimes v\right)  \left(  e_{j},e_{k}\right)
=v_{j}v_{k}.$
\end{proof}

We now compute the evolution equation for the moments $\left(  M_{j,k}\right)
_{j,k=1,2,3}.$

\begin{proposition}
\label{EqMoments}Suppose that $G_{0}\in\mathcal{M}_{+}\left(  \mathbb{R}%
_{c}^{3}\right)  $ satisfies and $\int_{\mathbb{R}^{3}}\left(  1+\left\vert
w\right\vert ^{s}\right)  G_{0}\left(  dw\right)  <\infty$ for some $s>2$ and
that $\mathbb{C}G$ is as in (\ref{eq:collboltzgen}) with $B$ satisfying the
assumptions in Theorem \ref{th:ssprof}. Let us assume also that
\begin{equation}
\int_{\mathbb{R}^{3}}G_{0}\left(  dw\right)  =1\ \ \text{and\ \ }%
\int_{\mathbb{R}^{3}}w\,G_{0}\left(  dw\right)  =0\ . \label{S6E3}%
\end{equation}
Suppose that $G\in C\left(  \left[  0,\infty\right]  :\mathcal{M}_{+}\left(
\mathbb{R}_{c}^{3}\right)  \right)  $ is the unique mild solution of
(\ref{eq:boltzgen})-(\ref{eq:invalBol}) in Theorem \ref{WellPos}. Then the
tensor $M=\left(  M_{j,k}\right)  _{j,k=1,2,3}$ defined by means of
$M_{j,k}=\int_{\mathbb{R}^{3}}w_{j}w_{k}G\left(  t,dw\right)  $ is defined for
$t\geq0$ and it satisfies the system of ODEs
\begin{equation}
\frac{dM_{j,k}}{dt}+Q_{j,\ell}\left(  t\right)  M_{k,\ell}+Q_{k,\ell}\left(
t\right)  M_{j,\ell}=-2b\left(  M_{j,k}-m\delta_{j,k}\right)
\ ,\ \ j,k=1,2,3,\ \ \ M_{j,k}=M_{k,j} \label{S4E8}%
\end{equation}
where $b$ is as in (\ref{S6E1}).
\end{proposition}

\begin{remark}
In (\ref{S4E8}) we use the convention that the repeated indexes are summed. We
will use the same convention in the rest of the paper.
\end{remark}

\begin{proof}
Due to Theorem \ref{WeakSolTh} with $s>2$ the tensor $M$ is well defined and
$G$ is a weak solution of (\ref{eq:boltzgen})-(\ref{eq:invalBol}) in the sense
of Definition \ref{WeakSol}. Moreover, Theorem \ref{WeakSolTh} implies also
that we can take as test functions in (\ref{T3E6}) $\varphi=1$ and
$\varphi=w_{j}$
\begin{equation}
\int_{\mathbb{R}^{3}}G\left(  t,dw\right)  =1\ \ \text{and\ \ }\int
_{\mathbb{R}^{3}}w\,G\left(  t,dw\right)  =0 . \label{S6E4}%
\end{equation}

Moreover, taking in (\ref{T3E6}) the test functions $\varphi=W_{j,k}%
=w_{j}w_{k}$ we obtain
\begin{equation}
\frac{dM_{j,k}}{dt}+\int_{\mathbb{R}^{3}}\left[  Q(t)w\cdot\partial_{w}\left(
w_{j}w_{k}\right)  \right]  G\left(  t,dw\right)  =K_{j,k} \label{S5E0}%
\end{equation}
where we define $K_{j,k}$ as
\begin{equation}
K_{j,k}=\int_{\mathbb{R}^{3}}\int_{\mathbb{R}^{3}}T_{j,k}G\left(  t,dw\right)
G\left(  t,dw_{\ast}\right)  \label{S5E1}%
\end{equation}
with $T_{j,k}$ as in (\ref{S5E1a}) and where we have used that $B\left(
n\cdot\omega,\left\vert w-w_{\ast}\right\vert \right)  =B\left(  n\cdot
\omega\right)  $ due to the fact that $B$ is homogeneous of order zero. Using
(\ref{S6E2}) we obtain
\[
K_{j,k}=-b\int_{\mathbb{R}^{3}}\int_{\mathbb{R}^{3}}\left[  \left(  w-w_{\ast
}\right)  _{j}\left(  w-w_{\ast}\right)  _{k}-\frac{\left\vert w-w_{\ast
}\right\vert ^{2}}{3}\delta_{j,k}\right]  G\left(  t,dw\right)  G\left(
t,dw_{\ast}\right) .
\]

Expanding the products and using (\ref{S6E4}) as well as the symmetry
properties of the integrals, we obtain
\begin{equation}
K_{j,k}=-2b\left(  M_{j,k}-m\delta_{j,k}\right)  \ ,\ m=\frac{1}%
{3}\mathop{\rm tr}\left(  M\right)  \label{S5E2}%
\end{equation}
where $b$ is as in (\ref{S6E1}). Notice that the trace of $K=\left(
K_{j,k}\right)  _{j,k=1,2,3}$ is zero, something that might be seen directly
from (\ref{S5E1}) using that $\left\vert w^{\prime}\right\vert ^{2}+\left\vert
w_{\ast}^{\prime}\right\vert ^{2}=\left\vert w\right\vert ^{2}+\left\vert
w_{\ast}\right\vert ^{2}.$

Computing the integral on the left hand side of (\ref{S5E0}) and using
(\ref{S5E2}) we obtain (\ref{S4E8}).
\end{proof}

\subsection{Self-similar profiles for Maxwell molecules \label{SelfSim}}

We prove in this subsection a general theorem on the
existence of  
self-similar homoenergetic solutions for
several  choices of the matrix $A$ in (\ref{B7_0}).   Two special cases
are considered in detail in the sequel: 1) {\it simple shear} (cf. (\ref{T1E5}) and Section \ref{ss:SelfSimSimpSh}),
and 2) {\it planar shear} (cf. (\ref{T1E3}) and Section \ref{planarshear}).  
Formally, we begin from (\ref{D1_0}) and make self-similar ansatzes.  In the case of simple
shear we make the ansatz
\begin{equation}
g\left(  w,t\right)  =e^{-3\beta t}G\left(  \xi\right)  \ \ ,\ \ \xi=\frac
{w}{e^{\beta t}} \label{0S8E8}%
\end{equation} and reduce (\ref{D1_0})  to
\begin{equation}
-\beta\partial_{\xi}\left(  \xi G\right)  -K\partial_{\xi_{1}}\left(  \xi
_{2}G\right)  =\mathbb{C}\left[  G\right]  . \label{0W1}%
\end{equation}
In the case of planar shear we first make the change of variables
\[
g\left(  t,w\right)  =\frac{1}{t}\bar{g}\left(  \tau,w\right)  \ \ ,\ \ \tau
=\log\left(  t\right),
\]
which reduces (\ref{D1_0}) to 
\begin{equation}
\partial_{\tau}\bar{g}-Kw_{3}\partial_{w_{2}}\bar{g}-\partial_{w_{3}}\left(
w_{3}\bar{g}\right)  =\mathbb{C}\bar{g}\left(  w\right). \label{0T4E9}%
\end{equation}
Then we further assume that
\begin{equation}
\bar{g}\left(  \tau,w\right)  =e^{-3\beta\tau}G\left(  \xi\right)
\ \ ,\ \ \xi=\frac{w}{e^{\beta\tau}}. \label{0T5E2}%
\end{equation}
Inserting (\ref{0T5E2}) into (\ref{0T4E9}) we get
\begin{equation}
-\beta\partial_{\xi}\left(  \xi\cdot G\right)  -K\partial_{\xi_{2}}\left(
\xi_{3}G\right)  -\partial_{\xi_{3}}\left(  \xi_{3}G\right)  =\mathbb{C}%
G\left(  w\right) . \label{0T5E3}%
\end{equation}
Detailed descriptions of these reductions are found in Sections \ref{ss:SelfSimSimpSh} and
\ref{planarshear}.

\smallskip
A general equation that includes both of these cases is
\begin{equation}
-\alpha\partial_{w}\cdot\left(  wG\right)  -\partial_{w}\cdot\left(
Lw\,G\right)  =\mathbb{C}G\left(  w\right)  \ \ ,\ \ w\in\mathbb{R}^{3}
\label{S4E5}%
\end{equation}
where $L\in M_{3\times3}({\mathbb{R}})$, $\alpha\in\mathbb{R}$, $\mathbb{C}%
\left[  \cdot\right]  $ is the collision operator in (\ref{eq:collboltzgen})
with the function $\Lambda\left(  w,w_{\ast}\right)  $ defined in
(\ref{T3E2c}) satisfying (\ref{T3E2a}). We will assume also that the function
$B\left(  n\cdot\omega,\left\vert w-w_{\ast}\right\vert \right)  $ is
homogeneous of order zero on the variable $\left\vert w-w_{\ast}\right\vert .$
In that case $\Lambda\left(  w,w_{\ast}\right)  $ is just a real number, which
will be denoted as $\Lambda_{0}\in\mathbb{R}_{+}.$ The effect of the collision
kernel $\mathbb{C}$ is nontrivial if $\Lambda_{0}>0.$

The main result that we will prove in this subsection is the following:

\begin{theorem}
\label{th:ssprof}Suppose that $B$ in (\ref{eq:collboltzgen}) is homogeneous of
order zero and that $b$ in (\ref{S6E1}) is strictly positive. There exists a sufficiently small 
$k_{0}>0$ such that, for any $\zeta\geq0$ and any $L\in
M_{3\times3}({\mathbb{R}})$ satisfying $\Vert L\Vert\leq k_{0}b$, there exists
$\alpha\in\mathbb{R}$ and $G\in\mathcal{M}_{+}\left(  \mathbb{R}_{c}%
^{3}\right)  $ that solve (\ref{S4E5}) in the sense of measures and
satisfy
\begin{equation}
\int_{\mathbb{R}^{3}}G\left(  dw\right)  =1\ ,\ \ \int_{\mathbb{R}^{3}}%
w_{j}G\left(  dw\right)  =0,\ \int_{\mathbb{R}^{3}}\left\vert w\right\vert
^{2}G\left(  dw\right)  =\zeta\ . \label{S4E6}%
\end{equation}

\end{theorem}

\begin{remark}
\label{ZeroEner}Notice that if $\zeta=0$ the only solution of (\ref{S4E5})
satisfying (\ref{S4E6}) is $G=\delta_{w=0}.$\ The main content of Theorem
\ref{th:ssprof} is the existence of solutions of (\ref{S4E5}) with arbitrary
values of the velocity dispersion, something that can be thought also as
arbitrary values of the temperature.
\end{remark}

\begin{remark}
Theorem \ref{th:ssprof} is a perturbative result, because we assume the
smallness condition $\left\Vert L\right\Vert \leq k_{0}b.$ This will allow us
to prove the existence of self-similar solutions for several of the fluxes
described in Section \ref{ss:classeqsol}. 
However, this smallness condition is probably
not really needed, at least in the case of simple shear in (\ref{T1E5}). We derive in Theorem \ref{ThSimpShearArbK}
a sufficient condition for the existence of self-similar solutions in the simple shear case for arbitrary values of the shear parameter. This condition could be checked numerically for each choice of the kernel $B$.
The derivation of this condition requires a more careful examination of the
interplay between the hyperbolic term and the collision term in (\ref{S4E5}).
\end{remark}

The main idea in the proof of Theorem \ref{th:ssprof} is to prove the
existence of nontrivial steady states for the solutions of the evolution
equation
\begin{equation}
G_{t}-\alpha\partial_{w}\cdot\left(  wG\right)  -\partial_{w}\cdot\left(
Lw\,G\right)  =\mathbb{C}G\left(  t,w\right)  \ ,\ t>0,\ w\in\mathbb{R}^{3}.
\label{S4E7}%
\end{equation}
The equation (\ref{S4E7}) is a particular case of the equation
(\ref{eq:boltzgen}) where we take $Q\left(  t\right)  =L+\alpha I$. In this
case the  equations (\ref{S4E8}) for the second moments become
\begin{equation}
\frac{dM_{j,k}}{dt}+2\alpha M_{j,k}+L_{j,\ell}M_{k,\ell}+L_{k,\ell}M_{j,\ell
}=-2b\left(  M_{j,k}-m\delta_{j,k}\right)  \ ,\ \ j,k=1,2,3,\ \ M_{j,k}%
=M_{k,j} . \label{S5E3}%
\end{equation}
The equations (\ref{S5E3}) comprise a linear system of equations with constant
coefficients. Therefore they have solutions of the form $M_{j,k}=\Gamma
_{j,k}e^{2b\lambda t}.$ On the other hand we can formulate an equivalent
problem, namely to determine the values of $\alpha$ for which there is a
stationary solution of (\ref{S5E3}) with the form $M_{j,k}=\Gamma_{j,k}.$ Such
values of $\alpha$ solve the eigenvalue problem
\begin{equation}
\frac{\alpha}{b}\Gamma_{j,k}+\frac{1}{2b}\left(  L_{j,\ell}\Gamma_{k,\ell
}+L_{k,\ell}\Gamma_{j,\ell}\right)  =-\left(  \Gamma_{j,k}-\Gamma\delta
_{j,k}\right)  \ ,\ \ j,\ k=1,2,3,\ \ \Gamma_{j,k}=\Gamma_{k,j} \label{S6E5}%
\end{equation}
with%
\begin{equation}
\Gamma=\frac{1}{3}\left(  \Gamma_{1,1}+\Gamma_{2,2}+\Gamma_{3,3}\right) 
\label{S6E6}%
\end{equation}
and $b>0$ given by (\ref{S6E1}).

We prove the following result.  

\begin{lemma}
\label{EigenPb} There exists a sufficiently small  $k_{0}>0$ such that, for any
$L\in M_{3\times3}({\mathbb{R}})$ satisfying  $\left\Vert L\right\Vert \leq
k_{0}b$,\ there exists $\alpha\in\mathbb{R}$ and a real symmetric,  positive-definite matrix
$\left(  \Gamma_{j,k}\right)  _{j,k=1,2,3}$ such that
(\ref{S6E5}), (\ref{S6E6}) hold. Moreover, $\alpha$ 
can be chosen to be 
the complex number
with largest real part for which (\ref{S6E5}), (\ref{S6E6}) holds for a
nonzero $\left(  \Gamma_{j,k}\right)  _{j,k=1,2,3}.$ 
This particular choice of $\alpha$,
denoted $\bar{\alpha}$, satisfies
\begin{equation}
\left\vert \bar{\alpha}\right\vert \leq Ck_{0}b \label{S6E6a}%
\end{equation}
for some numerical constant $C>0$. 
\end{lemma}

\begin{proof}
Suppose first that $L=0.$ In that case the eigenvalue problem (\ref{S6E5}),
(\ref{S6E6}) can be solved explicitly. The problem is solved in a
six-dimensional space due to the symmetry condition $\Gamma_{j,k}=\Gamma
_{k,j}$. In this case there are two eigenvalues, namely $\alpha=-b$ and
$\alpha=0.$ In the case of the eigenvalue $\alpha=-b$ there is a
five-dimensional subspace of eigenvectors given by $\left\{  \Gamma=0\right\}
.$ On the other hand, in the case of the eigenvalue $\alpha=0$ we obtain the
one-dimensional subspace of eigenvectors $\Gamma_{j,k}=K\delta_{j,k}$ where
$K\in\mathbb{R}$. Then if $\left\Vert L\right\Vert \leq k_{0}b$ the result
follows from standard continuity results for the eigenvalues. The
corresponding matrix $\Gamma_{j,k}$ is a perturbation of the identity and then
it is positive definite. The fact that the eigenvalue with the largest real
part is real follows from the fact that the problem (\ref{S6E5}), (\ref{S6E6})
has real coefficients and therefore the eigenvalues, if they have a nonzero
imaginary part appear in pairs of complex conjugate numbers. However, there is
only one eigenvalue close to $\alpha=0$ since the degeneracy of the eigenvalue
$\alpha=0$ if $L=0$ is one. The estimate (\ref{S6E6a}) follows from standard
differentiability properties for the simple eigenvalues of matricial
eigenvalue problems (cf. \cite{K76}). 
\end{proof}

\bigskip

\begin{remark}
\label{dimSubsp}We notice that the dimension of the space of eigenvectors of
the eigenvalue problem (\ref{S6E5}), (\ref{S6E6}) is not necessarily six,
because the problem is not self-adjoint. Actually, we will see in Subsection
\ref{ss:SelfSimSimpSh} that if the matrix $L$ is chosen as in the simple shear
case (cf. (\ref{T1E5})) the subspace of eigenvectors is five-dimensional.
\end{remark}

As indicated in the Lemma we will use the notation $\bar{\alpha}$ to denote
the eigenvalue of the problem (\ref{S6E5}), (\ref{S6E6}) with the largest real
part obtained in Lemma \ref{EigenPb} and we will denote as $\bar{N}_{j,k}$ the
corresponding eigenvector. Then
\begin{align}
\frac{\bar{\alpha}}{b}\bar{N}_{j,k}+\frac{1}{2b}\left(  L_{j,\ell}\bar
{N}_{k,\ell}+L_{k,\ell}\bar{N}_{j,\ell}\right)   &  =-\left(  \bar{N}%
_{j,k}-\bar{N}\delta_{j,k}\right)  \ ,\ \ \bar{N}_{j,k}=\bar{N}_{k,j}%
\ \ ,\ \ j,\ k=1,2,3\nonumber\\
\bar{N}  &  =\frac{1}{3}\left(  \bar{N}_{1,1}+\bar{N}_{2,2}+\bar{N}%
_{3,3}\right)  \label{S7E3}%
\end{align}
where in order to have uniqueness we normalize $\bar{N}_{j,k}$ as
\begin{equation}
\sum_{j,k}\left(  \bar{N}_{j,k}\right)  ^{2}=1 . \label{S7E3a}%
\end{equation}

Notice that $\bar{\alpha}$ is bounded by $Ck_0b.$   

The following result is standard in Kinetic Theory (cf. \cite{CIP},
\cite{V02}). We just write here a version of the result convenient for the
arguments made later.

\begin{proposition}
[Povzner Estimates]\label{Povzner}Let $s>2$. There exists a continuous
function $\kappa_{s}:\left[  0,1\right]  \rightarrow\mathbb{R}$ such that
$\kappa_{s}\left(  y\right)  >0$ if $y\in\left[  0,1\right)  ,\ \kappa
_{s}\left(  0\right)  =0$, and a constant $C_{s}>0$ such that, for any
$w,w_{\ast}\in\mathbb{R}^{3}$ the following inequality holds
\begin{equation}
\left\vert w^{\prime}\right\vert ^{s}+\left\vert w_{\ast}^{\prime}\right\vert
^{s}-\left\vert w\right\vert ^{s}-\left\vert w_{\ast}\right\vert ^{s}%
\leq-\kappa_{s}\left(  \left\vert n\cdot\omega\right\vert \right)  \left(
\left\vert w\right\vert ^{s}+\left\vert w_{\ast}\right\vert ^{s}\right)
+C_{s}\left[  \left\vert w\right\vert ^{s-1}\left\vert w_{\ast}\right\vert
+\left\vert w_{\ast}\right\vert ^{s-1}\left\vert w\right\vert \right]  .
\label{S6E7}%
\end{equation}
 where $n=\frac{\left(
w-w_{\ast}\right)  }{\left\vert w-w_{\ast}\right\vert } $. 

\end{proposition}

\begin{proof}
Suppose first that $\frac{1}{2}\left\vert w\right\vert \leq\left\vert w_{\ast
}\right\vert \leq2\left\vert w\right\vert .$ Then, using the collision rule
(\ref{CM1}), (\ref{CM2}) we obtain, using that both norms are comparable
\begin{align*}
\left\vert w^{\prime}\right\vert ^{s}+\left\vert w_{\ast}^{\prime}\right\vert
^{s}-\left\vert w\right\vert ^{s}-\left\vert w_{\ast}\right\vert ^{s}  &
\leq\bar{C}_{s}\left(  \left\vert w\right\vert ^{s}+\left\vert w_{\ast
}\right\vert ^{s}\right)  -\left\vert w\right\vert ^{s}-\left\vert w_{\ast
}\right\vert ^{s}\\
&  \leq-\left(  \left\vert w\right\vert ^{s}+\left\vert w_{\ast}\right\vert
^{s}\right)  +C_{s}\left[  \left\vert w\right\vert ^{s-1}\left\vert w_{\ast
}\right\vert +\left\vert w_{\ast}\right\vert ^{s-1}\left\vert w\right\vert
\right]
\end{align*}
whence (\ref{S6E7}) follows. Let us assume then without loss of generality
that $\left\vert w\right\vert \leq\frac{1}{2}\left\vert w_{\ast}\right\vert ,$
since the symmetric case $\left\vert w_{\ast}\right\vert \leq\frac{1}%
{2}\left\vert w\right\vert $ can be studied analogously. We have several
possibilities. If $\left\vert w^{\prime}\right\vert \leq\frac{\left\vert
w\right\vert }{2}$ and $\left\vert w_{\ast}^{\prime}\right\vert \leq
\frac{\left\vert w\right\vert }{2}$ we obtain
\[
\left\vert w^{\prime}\right\vert ^{s}+\left\vert w_{\ast}^{\prime}\right\vert
^{s}-\left\vert w\right\vert ^{s}-\left\vert w_{\ast}\right\vert ^{s}\leq
C_{s}\left\vert w\right\vert ^{s}-\left\vert w\right\vert ^{s}-\left\vert
w_{\ast}\right\vert ^{s}%
\]
and (\ref{S6E7}) also follows. If $\max\left\{  \left\vert w^{\prime
}\right\vert ,\left\vert w_{\ast}^{\prime}\right\vert \right\}  >\frac
{\left\vert w\right\vert }{2}$ we argue as follows. Suppose that both
$\left\vert w^{\prime}\right\vert ,\left\vert w_{\ast}^{\prime}\right\vert $
are larger than $\left\vert w\right\vert .$ Then, using the triangular
inequality as well as $\left(  1+x\right)  ^{s}\leq1+C_{s}x,$ for any
$x\in\left[  0,1\right]  $ we obtain
\[
\left\vert w^{\prime}\right\vert ^{s}\leq\left\vert w^{\prime}-w\right\vert
^{s}+C_{s}\left\vert w\right\vert \left\vert w^{\prime}-w\right\vert
^{s-1}\ \ ,\ \ \left\vert w_{\ast}^{\prime}\right\vert ^{s}\leq\left\vert
w_{\ast}^{\prime}-w\right\vert ^{s}+C_{s}\left\vert w\right\vert \left\vert
w_{\ast}^{\prime}-w\right\vert ^{s-1}.
\]

On the other hand
\[
\left\vert w_{\ast}\right\vert ^{s}\geq\left\vert w_{\ast}-w\right\vert
^{s}-C_{s}\left\vert w\right\vert \left\vert w_{\ast}-w\right\vert ^{s-1}%
\]
whence, using that $\left\vert w^{\prime}-w\right\vert \leq C\left\vert
w_{\ast}\right\vert $ and $\left\vert w_{\ast}^{\prime}-w\right\vert \leq
C\left\vert w_{\ast}\right\vert $, we get
\begin{equation}
\left\vert w^{\prime}\right\vert ^{s}+\left\vert w_{\ast}^{\prime}\right\vert
^{s}-\left\vert w\right\vert ^{s}-\left\vert w_{\ast}\right\vert ^{s}%
\leq\left\vert w^{\prime}-w\right\vert ^{s}+\left\vert w_{\ast}^{\prime
}-w\right\vert ^{s}-\left\vert w_{\ast}-w\right\vert ^{s}+C_{s}\left\vert
w_{\ast}\right\vert ^{s-1}\left\vert w\right\vert . \label{S6E8}%
\end{equation}

Using (\ref{CM1}), (\ref{CM2}) we obtain $w^{\prime}-w=P_{\omega}\left(
w_{\ast}-w\right)  $ and $w_{\ast}^{\prime}-w=P_{\omega}^{\bot}\left(
w_{\ast}-w\right)  $ (cf. (\ref{S5E8})). Representing $\left(  w_{\ast
}-w\right)  ,\ \left(  w^{\prime}-w\right)  $ and $\left(  w_{\ast}^{\prime
}-w\right)  $ in a spherical coordinate system for which the North Pole is
$\omega$ we obtain
\begin{equation}
\left\vert w^{\prime}-w\right\vert ^{s}+\left\vert w_{\ast}^{\prime
}-w\right\vert ^{s}-\left\vert w_{\ast}-w\right\vert ^{s}=-\kappa_{s}\left(
n\cdot\omega\right)  \left\vert w_{\ast}-w\right\vert ^{s}\left[  \left\vert
\cos\left(  \theta\right)  \right\vert ^{s}+\left\vert \sin\left(
\theta\right)  \right\vert ^{s}-1\right]  \ \label{S6E9}%
\end{equation}
where
\[
\kappa_{s}\left(  n\cdot\omega\right)  =1-\left\vert \cos\left(
\theta\right)  \right\vert ^{s}-\left\vert \sin\left(  \theta\right)
\right\vert ^{s}
\]
and $n=\frac{\left(
w-w_{\ast}\right)  }{\left\vert w-w_{\ast}\right\vert } $.

The function $\kappa_{s}$ has all the properties stated in the Proposition if
$s>2.$ Combining (\ref{S6E8}) and (\ref{S6E9}) we obtain
\begin{align*}
\left\vert w^{\prime}\right\vert ^{s}+\left\vert w_{\ast}^{\prime}\right\vert
^{s}-\left\vert w\right\vert ^{s}-\left\vert w_{\ast}\right\vert ^{s}  &
\leq-\kappa_{s}\left(  n\cdot\omega\right)  \left\vert w_{\ast}-w\right\vert
^{s}+C_{s}\left\vert w_{\ast}\right\vert ^{s-1}\left\vert w\right\vert \\
&  \leq-\kappa_{s}\left(  n\cdot\omega\right)  \left\vert w_{\ast}\right\vert
^{s}+C_{s}\left\vert w\right\vert ^{s}+C_{s}\left\vert w_{\ast}\right\vert
^{s-1}\left\vert w\right\vert \\
&  \leq-\kappa_{s}\left(  n\cdot\omega\right)  \left\vert w_{\ast}\right\vert
^{s}+2C_{s}\left\vert w_{\ast}\right\vert ^{s-1}\left\vert w\right\vert
\end{align*}
whence (\ref{S6E7}) follows in this case since $\left\vert w\right\vert
\leq\frac{1}{2}\left\vert w_{\ast}\right\vert $. If $\left\vert w^{\prime
}\right\vert \leq\frac{\left\vert w\right\vert }{2}<\left\vert w_{\ast
}^{\prime}\right\vert $ we obtain with similar arguments
\begin{align*}
\left\vert w^{\prime}\right\vert ^{s}+\left\vert w_{\ast}^{\prime}\right\vert
^{s}-\left\vert w\right\vert ^{s}-\left\vert w_{\ast}\right\vert ^{s}  &
\leq\left\vert w_{\ast}^{\prime}-w\right\vert ^{s}-\left\vert w_{\ast
}-w\right\vert ^{s}+C_{s}\left\vert w_{\ast}\right\vert ^{s-1}\left\vert
w\right\vert \\
&  \leq\left\vert w^{\prime}-w\right\vert ^{s}+\left\vert w_{\ast}^{\prime
}-w\right\vert ^{s}-\left\vert w_{\ast}-w\right\vert ^{s}+C_{s}\left\vert
w_{\ast}\right\vert ^{s-1}\left\vert w\right\vert
\end{align*}
and we can the obtain (\ref{S6E7}) in the same way. The case $\left\vert
w_{\ast}^{\prime}\right\vert \leq\frac{\left\vert w\right\vert }{2}<\left\vert
w^{\prime}\right\vert $ is similar.
\end{proof}

We now prove that the nonlinear evolution defined by means of Theorems
\ref{WellPos}, \ref{WeakSolTh} is continuous in time in the weak topology of measures.

\begin{lemma}
\label{Semig}Suppose that $G_{0}\in\mathcal{M}_{+}\left(  \mathbb{R}_{c}%
^{3}\right)  $ satisfies
\begin{equation}
\int_{\mathbb{R}^{3}}G_{0}\left(  dw\right)  =1,\ \int_{\mathbb{R}^{3}%
}\left\vert w\right\vert ^{s}G_{0}\left(  dw\right)  <\infty\label{S7E1}%
\end{equation}
for some $s>2.$ We denote as $\mathcal{S}_{\alpha}\left(  t\right)
G_{0}=G\left(  t,\cdot\right)  $ the unique mild solution of (\ref{S4E7})
given by Theorem \ref{WellPos}. Then the family of operators $\mathcal{S}%
_{\alpha}\left(  t\right)  $ define an evolution semigroup.\ The mapping
$\mathcal{S}_{\alpha}:\left[  0,\infty\right)  \times\mathcal{M}_{+}\left(
\mathbb{R}_{c}^{3}\right)  \rightarrow\mathcal{M}_{+}\left(  \mathbb{R}%
_{c}^{3}\right)  $ is uniformly continuous in the weak topology of
$\mathcal{M}_{+}\left(  \mathbb{R}_{c}^{3}\right)  $ on any set of the form
$\left[  0,T\right]  \times\mathcal{M}_{s,+}\left(  \mathbb{R}_{c}^{3}\right)
,$ where $T\in\left(  0,\infty\right)  $ and $\mathcal{M}_{s,+}\left(
\mathbb{R}_{c}^{3}\right)  $ is the subset of measures $G_{0}\in
\mathcal{M}_{+}\left(  \mathbb{R}_{c}^{3}\right)  $ satisfying (\ref{S7E1}).
\end{lemma}

\begin{proof}
The semigroup property is just a consequence of the results in Theorems
\ref{WellPos}, \ref{WeakSolTh}. In order to prove the weak continuity of the
operator in $t$ we prove first that the functions $t\rightarrow\int
_{\mathbb{R}_{c}^{3}}\varphi\left(  w\right)  G\left(  t,dw\right)  $ are
continuous for any test function $\varphi\in C\left(  \mathbb{R}_{c}%
^{3}\right)  .$ To prove this we notice that for any function $\varphi\in
C^{1}\left(  \mathbb{R}^{3}\right)  $ such that $\varphi\left(  w\right)  $ is
constant for $\left\vert w\right\vert \geq R$ with $R$ large we have
\begin{equation}
\left\vert \int_{\mathbb{R}^{3}}\varphi\left(  w\right)  G\left(
dw,t_{1}\right)  -\int_{\mathbb{R}^{3}}\varphi\left(  w\right)  G\left(
dw,t_{2}\right)  \right\vert \leq C\left\vert t_{1}-t_{2}\right\vert .
\label{S2E7a}%
\end{equation}
with $C$ depending on the derivatives of $\varphi,$ but independent on $G_{0}$
if $\int_{\mathbb{R}^{3}}G_{0}\left(  dw\right)  =1$. This is a
consequence of the weak formulation identity (\ref{T3E6}). Since
\[
\left\Vert G\left(  t,\cdot\right)  \right\Vert _{1,\sigma}\leq C\left(
T\right)  ,\ 0\leq t\leq T
\]
we obtain that the contributions to the integrals due to the region $\left\{
\left\vert w\right\vert \geq R\right\}  $ can be made arbitrarily small if $R$
is large. Then the stated weak continuity in time follows using the
density of the chosen test functions in $C\left(  \mathbb{R}_{c}^{3}\right)
.$

It only remains to prove that for any $T>0$ the mapping $\mathcal{S}_{\alpha
}\left(  T\right)  :\mathcal{M}_{s,+}\left(  \mathbb{R}_{c}^{3}\right)
\rightarrow\mathcal{M}_{s,+}\left(  \mathbb{R}_{c}^{3}\right)  $ is continuous
in the weak topology. To this end, we first notice that the function
$\Lambda\left(  w,w_{\ast}\right)  $ defined in (\ref{T3E2c}) is a constant
$\Lambda_{0}$ in the case of Maxwell molecules. We now use the following
metric to characterize the weak topology: 
\begin{equation}
\mathcal{W}_{1}\left(  G,H\right)  =\sup_{\left\vert \partial_{w}%
\varphi\right\vert \leq1}\int\varphi\left(  w\right)  \left(  G-H\right)
\left(  dw\right)  . \label{S8E1}%
\end{equation}
This metric is referred to as the $1$-Wasserstein distance (see for instance
\cite{V03}).

Suppose that $G$ is a weak solution of (\ref{S4E7}). Using the weak
formulation (\ref{T3E6}) we obtain
\begin{align*}
&  \partial_{t}\left(  \int_{\mathbb{R}^{3}}\varphi\left(  t,w\right)
G\left(  t,dw\right)  \right) \\
&  =\int_{\mathbb{R}^{3}}\partial_{t}\varphi\left(  t,w\right)  G\left(
t,dw\right)  -\int_{\mathbb{R}^{3}}\left[  \left(  \alpha w+Lw\right)
\cdot\partial_{w}\varphi\right]  G\left(  t,dw\right) \\
&  +\frac{1}{2}\int_{\mathbb{R}^{3}}dw\int_{\mathbb{R}^{3}}dw_{\ast}%
\int_{S^{2}}d\omega G\left(  t,dw\right)  G\left(  t,dw_{\ast}\right)
B\left(  n\cdot\omega\right)  \left[  \varphi^{\prime}+\varphi_{\ast}^{\prime
}-\varphi-\varphi_{\ast}\right] .
\end{align*}

We now consider a test function with the form $\varphi\left(  t,w\right)
=\psi\left(  U\left(  t\right)  w\right)  ,\ $where $U\left(  t\right)
=e^{\left(  \alpha+L\right)  t}.$ Then the identity above becomes
\begin{align}
\label{S7E8}\partial_{t}\left(  \int_{\mathbb{R}^{3}}\psi\left(  U\left(
t\right)  w\right)  G\left(  t,dw\right)  \right)  = & \frac{1}{2}%
\int_{\mathbb{R}^{3}}dw\int_{\mathbb{R}^{3}}dw_{\ast}\int_{S^{2}}d\omega
G\left(  t,dw\right)  G\left(  t,dw_{\ast}\right)  B\left(  n\cdot
\omega\right) \\
&  \times\left[  \psi\left(  U\left(  t\right)  w^{\prime}\right)
+\psi\left(  U\left(  t\right)  w_{\ast}^{\prime}\right)  -\psi\left(
U\left(  t\right)  w\right)  -\psi\left(  U\left(  t\right)  w_{\ast}\right)
\right]  .\nonumber
\end{align}

Suppose that we have two solutions of (\ref{S4E7}), $G_{1},G_{2}$ with initial
data $G_{1,0},G_{2,0}\in\mathcal{M}_{+}\left(  \mathbb{R}_{c}^{3}\right)  $
respectively. Integrating in time (\ref{S7E8}), writing the resulting equation
for both solutions $G_{1},G_{2}$ and taking the difference we obtain
\begin{align}
&  \int_{\mathbb{R}^{3}}\psi\left(  U\left(  t\right)  w\right)  \left(
G_{1}-G_{2}\right)  \left(  t,dw\right) \nonumber\\
&  = \int_{\mathbb{R}^{3}}\psi\left(  w\right)  \left(  G_{1,0}-G_{2,0}%
\right)  \left(  dw\right) \nonumber\\
&  \quad+\frac{1}{2}\int_{0}^{t}ds\int_{\mathbb{R}^{3}}dw\int_{\mathbb{R}^{3}%
}dw_{\ast}\int_{S^{2}}d\omega G_{1}\left(  s,dw\right)  \left(  G_{1}%
-G_{2}\right)  \left(  s,dw_{\ast}\right)  B\left(  n\cdot\omega\right)
\nonumber\\
&  \quad\times\left[  \psi\left(  U\left(  s\right)  w^{\prime}\right)
+\psi\left(  U\left(  s\right)  w_{\ast}^{\prime}\right)  -\psi\left(
U\left(  s\right)  w\right)  -\psi\left(  U\left(  s\right)  w_{\ast}\right)
\right] \nonumber\\
&  \quad+\frac{1}{2}\int_{0}^{t}ds\int_{\mathbb{R}^{3}}dw\int_{\mathbb{R}^{3}%
}dw_{\ast}\int_{S^{2}}d\omega\left(  G_{1}-G_{2}\right)  \left(  s,dw\right)
G_{2}\left(  s,dw_{\ast}\right)  B\left(  n\cdot\omega\right) \nonumber\\
&  \quad\times\left[  \psi\left(  U\left(  s\right)  w^{\prime}\right)
+\psi\left(  U\left(  s\right)  w_{\ast}^{\prime}\right)  -\psi\left(
U\left(  s\right)  w\right)  -\psi\left(  U\left(  s\right)  w_{\ast}\right)
\right]  . \label{S7E9}%
\end{align}

We now take $\psi\left(  w\right)  =\varphi\left(  U\left(  -t\right)
w\right)  $ for a test function $\varphi\left(  \xi\right)  $ satisfying
$\left\vert \partial_{\xi}\varphi\right\vert \leq1.$ Then, using the chain
rule as well as the fact $\left\vert U\left(  t\right)  \right\vert
+\left\vert U\left(  -t\right)  \right\vert \leq C_{T}$ for $0\leq t\leq T$
and the collision rule (\ref{CM1}), (\ref{CM2}) we obtain
\begin{align*}
\left\vert \partial_{w}\psi\left(  w\right)  \right\vert  &  =\left\vert
\partial_{w}\varphi\left(  U\left(  -t\right)  w\right)  \right\vert \leq
C_{T},\\
\left\vert \partial_{w}\psi\left(  U\left(  s\right)  w^{\prime}\right)
\right\vert  &  =\left\vert \partial_{w}\varphi\left(  U\left(  -t\right)
U\left(  s\right)  w^{\prime}\right)  \right\vert \leq C_{T},\\
\left\vert \partial_{w}\psi\left(  U\left(  s\right)  w_{\ast}^{\prime
}\right)  \right\vert  &  =\left\vert \partial_{w}\varphi\left(  U\left(
-t\right)  U\left(  s\right)  w_{\ast}^{\prime}\right)  \right\vert \leq
C_{T},\\
\left\vert \partial_{w}\psi\left(  U\left(  s\right)  w_{\ast}\right)
\right\vert  &  =\left\vert \partial_{w}\varphi\left(  U\left(  -t\right)
U\left(  s\right)  w_{\ast}\right)  \right\vert \leq C_{T},\\
\left\vert \partial_{w}\psi\left(  U\left(  s\right)  w\right)  \right\vert
&  =\left\vert \partial_{w}\varphi\left(  U\left(  -t\right)  U\left(
s\right)  w\right)  \right\vert \leq C_{T}.
\end{align*}

Moreover, the second and third estimates as well as our assumptions (cf.
\eqref{T3E2c}, \eqref{T3E2a}) in $B$ imply
\begin{equation}
\left\vert \partial_{w}\left(  \int_{S^{2}}B\left(  n\cdot\omega\right)
\psi\left(  U\left(  s\right)  w^{\prime}\right)  d\omega\right)  \right\vert
+\left\vert \partial_{w}\left(  \int_{S^{2}}B\left(  n\cdot\omega\right)
\psi\left(  U\left(  s\right)  w_{\ast}^{\prime}\right)  d\omega\right)
\right\vert \leq C_{T} \label{S8E2}%
\end{equation}

To check this estimate we argue as follows. We estimate the second term, since
the first one is similar. We introduce a rotation matrix $\mathcal{R}\left(
n\right)  \in SO\left(  3\right)  $ which transforms one of the coordinate
axes, say $e_{1}$ into the vector $n$. We then change variables by means of
$\omega=\mathcal{R}\left(  n\right)  \hat{\omega},$ whence
\begin{align*}
&  \int_{S^{2}}B\left(  n\cdot\omega\right)  \psi\left(  U\left(  s\right)
w_{\ast}^{\prime}\right)  d\omega\\
&  =\int_{S^{2}}B\left(  \mathcal{R}\left(  n\right)  e_{1}\cdot
\mathcal{R}\left(  n\right)  \hat{\omega}\right)  \psi\left(  U\left(
s\right)  w_{\ast}^{\prime}\right)  d\hat{\omega}\\
&  =\int_{S^{2}}B\left(  e_{1}\cdot\hat{\omega}\right)  \psi\left(  U\left(
s\right)  w_{\ast}^{\prime}\right)  d\hat{\omega}%
\end{align*}
where $w_{\ast}^{\prime}=w_{\ast}-\left(  \left(  w_{\ast}-w\right)
\cdot\omega\right)  \omega=w_{\ast}-\left(  \left(  w_{\ast}-w\right)
\cdot\mathcal{R}\left(  n\right)  \hat{\omega}\right)  \mathcal{R}\left(
n\right)  \hat{\omega}.$ We then need to compute $\partial_{w}\psi\left(
U\left(  s\right)  w_{\ast}^{\prime}\right)  $ with $w_{\ast}^{\prime}$ given
by this formula. In particular this requires to estimate $\partial_{w}\left[
\left(  \left(  w_{\ast}-w\right)  \cdot\mathcal{R}\left(  n\right)
\hat{\omega}\right)  \mathcal{R}\left(  n\right)  \hat{\omega}\right]  $ where
$n=\frac{\left(  w-w_{\ast}\right)  }{\left\vert w-w_{\ast}\right\vert }.$
Therefore, using that $\left\vert \partial_{w}n\right\vert \leq\frac
{C}{\left\vert w-w_{\ast}\right\vert }$ we obtain
\[
\left\vert \partial_{w}\left[  \left(  \left(  w_{\ast}-w\right)
\cdot\mathcal{R}\left(  n\right)  \hat{\omega}\right)  \mathcal{R}\left(
n\right)  \hat{\omega}\right]  \right\vert \leq C .
\]
Thus%
\[
\left\vert \partial_{w}\psi\left(  U\left(  s\right)  w_{\ast}^{\prime
}\right)  \right\vert \leq C_{T}\ \ ,\ \ 0\leq s\leq T
\]
and this implies (\ref{S8E2}).

Taking now the supremum in (\ref{S7E9}) over all the functions $\varphi$
satisfying $\left\vert \partial_{\xi}\varphi\right\vert \leq1$ and using the
definition of the $1$-Wasserstein distance in (\ref{S8E1}) we obtain

\[
\mathcal{W}_{1}\left(  G_{1},G_{2}\right)  \left(  t\right)  \leq
C_{T}\mathcal{W}_{1}\left(  G_{1,0},G_{2,0}\right)  +C_{T}\int_{0}%
^{t}\mathcal{W}_{1}\left(  G_{1},G_{2}\right)  \left(  s\right)  ds
\]
and using Gronwall's Lemma we obtain
\[
\mathcal{W}_{1}\left(  G_{1},G_{2}\right)  \left(  t\right)  \leq
C_{T}\mathcal{W}_{1}\left(  G_{1,0},G_{2,0}\right)
\]
and this implies the continuity of $\mathcal{S}_{\alpha}\left(  t\right)  $ in
the weak topology.
\end{proof}

\begin{proposition}
\label{InvReg}Let $2<s<3.$ Suppose that $\int_{\mathbb{R}^{3}}G_{0}\left(
dw\right)  =1$ $\int_{\mathbb{R}^{3}}\left\vert w\right\vert ^{s}G_{0}\left(
dw\right)  <\infty$ and that the following identities hold:
\begin{equation}
\int_{\mathbb{R}^{3}}w_{j}G_{0}\left(  dw\right)  =0\ \ ,\ \ j\in\left\{
1,2,3\right\}  \ \ ,\ \ \int_{\mathbb{R}^{3}}w_{j}w_{k}G_{0}\left(  dw\right)
=K\bar{N}_{j,k}\ ,\ \ j,k\in\left\{  1,2,3\right\}  \label{S7E4}%
\end{equation}
where $\left(  \bar{N}_{j,k}\right)  $ is as in (\ref{S7E3}), (\ref{S7E3a})
and $K\geq0.$ Then
\begin{equation}%
\begin{split}
& \int_{\mathbb{R}^{3}}\mathcal{S}_{\bar{\alpha}}\left(  t\right)
G_{0}\left(  dw\right)  =1\ ,\ \ \int_{\mathbb{R}^{3}}w_{j}\mathcal{S}%
_{\bar{\alpha}}\left(  t\right)  G_{0}\left(  dw\right)  =0\ \ ,\\
&  \, \int_{\mathbb{R}^{3}}w_{j}w_{k}\mathcal{S}_{\bar{\alpha}}\left(
t\right)  \left(  G_{0}\right)  \left(  dw\right)  =K\bar{N}_{j,k}%
\ \ \text{for any }t\geq0\label{S2E6a}%
\end{split}
\end{equation}
where $\bar{\alpha}$ is as in (\ref{S7E3}). Moreover, there exists $k_{0}>0$
sufficiently small, which depends on $B,$ such that if $\Vert L\Vert\leq
k_{0}b$ there exists a constant $C_{\ast}=C_{\ast}\left(  K\right)  >0$ such
that if we assume that
\begin{equation}
\int_{\mathbb{R}^{3}}\left\vert w\right\vert ^{s}G_{0}\left(  dw\right)  \leq
C_{\ast} \label{S7E5}%
\end{equation}
then, for any $t\geq0$%
\begin{equation}
\int_{\mathbb{R}^{3}}\left\vert w\right\vert ^{s}\mathcal{S}\left(  t\right)
\left(  G_{0}\right)  \left(  dw\right)  \leq C_{\ast}\ . \label{S7E6}%
\end{equation}

\end{proposition}

\begin{remark}
Due to Lemma \ref{EigenPb} the matrix $\left(  \bar{N}_{j,k}\right)
_{j,k=1,2,3}$ is positive definite. Then, it might be readily seen, using a
coordinate system in which $\left(  \bar{N}_{j,k}\right)  _{j,k=1,2,3}$ is
diagonal that there exists measures $G_{0}$ satisfying (\ref{S7E4}) and
(\ref{S7E5}).
\end{remark}

\begin{remark}
\label{ModCont}It will be seen in the proof that the constant $k_{0}$ depends
on a function that characterizes the absolute continuity of the integrals of
$B$. More precisely, we define the function
\[
\Omega\left(  \delta\right)  :=\sup_{\left\vert A\right\vert \leq\delta}%
\int_{A}B\left(  n\cdot\omega\right)  d\omega\ \ ,\ \ \delta>0\ \ ,\ \ n\in
S^{2}%
\]
where the supremum is taken over all the Borel sets $A$ such that $A\subset
S^{2}$ and $\left\vert A\right\vert $ is its measure in $S^{2}.$ Notice that
the function $\Omega$ is independent of $n$ due to its invariance under
rotations. Our assumptions on $B$ (cf. (\ref{T3E2c}), (\ref{T3E2a})) imply,
due to the absolute continuity \ property of the $L^{1}$ functions, that
$\lim_{\delta\rightarrow0}\Omega\left(  \delta\right)  =0.$ The constant
$k_{0}$ in Proposition \ref{InvReg} depends only on the function $\Omega.$
Notice that if we had assumed that $B$ contains Dirac masses, we would not
have $\lim_{\delta\rightarrow0}\Omega\left(  \delta\right)  =0$ and it will be
seen in the proof of (\ref{S7E6}) below would fail.
\end{remark}

\begin{proof}
Due to Proposition \ref{EqMoments} the moments $M_{j,k}$ satisfy (\ref{S5E3}).
Then, choosing $\alpha=\bar{\alpha}$ as well as (\ref{S7E3}) we obtain the
second group of identities in (\ref{S2E6a}). The conservation of mass and
linear momentum in (\ref{S2E6a}) follows as in the proof of Proposition
\ref{EqMoments}.

It only remains to prove (\ref{S7E6}) assuming (\ref{S7E5}) with $C_{\ast}$
sufficiently large. To this end we approximate $G_{0}$ by the sequence
$G_{0,m}$ described in the Remark \ref{AppSol}. Given that $\left\Vert
G_{0,m}\right\Vert _{1,\bar{s}}<\infty,$ with $\bar{s}>s$ we can use in the
corresponding version of (\ref{T3E6}) the test function $\varphi\left(
w\right)  =\left\vert w\right\vert ^{s}$ with $2<s<3$ we obtain that the
function $M_{s}^{\left(  m\right)  }\left(  t\right)  =\int_{\mathbb{R}^{3}%
}\left\vert w\right\vert ^{s}G_{m}\left(  t,dw\right)  $ satisfies
\begin{align}
&  \partial_{t}M_{s}^{\left(  m\right)  }\left(  t\right)  =-s\bar{\alpha
}M_{s}^{\left(  m\right)  }\left(  t\right)  -s\int_{\mathbb{R}^{3}}\left\vert
w\right\vert ^{s-2}w(w\cdot Lw)G_{m}\left(  dw,t\right) \nonumber\\
&  +\frac{1}{2}\int_{\mathbb{R}^{3}}\int_{\mathbb{R}^{3}}\int_{S^{2}}d\omega
G_{m}\left(  dw,t\right)  G_{m}\left(  dw_{\ast},t\right)  B\left(
n\cdot\omega\right)  \left[  \left\vert w^{\prime}\right\vert ^{s}+\left\vert
w_{\ast}^{\prime}\right\vert ^{s}-\left\vert w\right\vert ^{s}-\left\vert
w_{\ast}\right\vert ^{s}\right] .\nonumber
\end{align}

We then estimate $\int_{\mathbb{R}^{3}}\left\vert w\right\vert ^{s-2}w(w\cdot
Lw)G_{m}\left(  dw,t\right)  $ by $k_{0}bM_{s}^{\left(  m\right)  }\left(
t\right)  .$ It then follows using (\ref{S6E6a}) as well as the Povzner
estimates (cf., (\ref{S6E7})) that
\begin{align*}
\partial_{t}M_{s}^{\left(  m\right)  }\left(  t\right)  \leq &  \, C\,
k_{0}bM_{s}^{\left(  m\right)  }\left(  t\right) \\
&  +\frac{1}{2}\int_{\mathbb{R}^{3}}\int_{\mathbb{R}^{3}}\int_{S^{2}}d\omega
G_{m}\left(  t,dw\right)  G_{m}\left(  t,dw_{\ast}\right)  B\left(
n\cdot\omega\right) \\
&  \quad\times\left[  -\kappa_{s}\left(  \left\vert n\cdot\omega\right\vert
\right)  \left(  \left\vert w\right\vert ^{s}+\left\vert w_{\ast}\right\vert
^{s}\right)  +C_{s}\left[  \left\vert w\right\vert ^{s-1}\left\vert w_{\ast
}\right\vert +\left\vert w_{\ast}\right\vert ^{s-1}\left\vert w\right\vert
\right]  \right]
\end{align*}
where $C$ is just a numerical constant. The function $\kappa_{s}\left(
y\right)  $ is continuous and it vanishes only for $y=0.$ Since $B$ is also
continuous we can prove that
\[
\int_{S^{2}}B\left(  n\cdot\omega\right)  \kappa_{s}\left(  \left\vert
n\cdot\omega\right\vert \right)  d\omega\geq\mu b
\]
for some $\mu>0$ which depends only on the modulus of continuity of $B$. Then
\begin{align*}
&  \partial_{t}M_{s}^{\left(  m\right)  }\left(  t\right)  \leq Ck_{0}%
bM_{s}^{\left(  m\right)  }\left(  t\right)  -\frac{\mu b}{2}\int
_{\mathbb{R}^{3}}\int_{\mathbb{R}^{3}}\left(  \left\vert w\right\vert
^{s}+\left\vert w_{\ast}\right\vert ^{s}\right)  G_{m}\left(  t,dw\right)
G_{m}\left(  t,dw_{\ast}\right) \\
&  +C_{s}\int_{\mathbb{R}^{3}}\int_{\mathbb{R}^{3}}\left[  \left\vert
w\right\vert ^{s-1}\left\vert w_{\ast}\right\vert +\left\vert w_{\ast
}\right\vert ^{s-1}\left\vert w\right\vert \right]  G_{m}\left(  t,dw\right)
G_{m}\left(  t,dw_{\ast}\right) \\
&  =\left(  Ck_{0}-\mu\right)  bM_{s}^{\left(  m\right)  }\left(  t\right)
+C_{s}\int_{\mathbb{R}^{3}}\int_{\mathbb{R}^{3}}\left[  \left\vert
w\right\vert ^{s-1}\left\vert w_{\ast}\right\vert +\left\vert w_{\ast
}\right\vert ^{s-1}\left\vert w\right\vert \right]  G_{m}\left(  t,dw\right)
G_{m}\left(  t,dw_{\ast}\right) .
\end{align*}

The estimates (\ref{S2E6a}) imply that $\int_{\mathbb{R}^{3}}\left\vert
w\right\vert ^{2}G_{m}\left(  t,dw\right)  \leq CK.$ Then, since $s<3,$%
\[
\partial_{t}M_{s}^{\left(  m\right)  }\left(  t\right)  \leq\left(  Ck_{0}%
-\mu\right)  bM_{s}^{\left(  m\right)  }\left(  t\right)  +C_{s}K .
\]

Here $C$ is just a numerical constant. Then, it follows that, choosing
$k_{0}\leq\frac{\mu}{2C}$, we have 
 $M_{s}^{\left(
m\right)  }\leq C_{\ast}=2CC_{s}K.$  Taking the limit $m\rightarrow\infty$ we
obtain $M_{s}^{\left(  m\right)  }\rightarrow M_{s}=\int_{\mathbb{R}^{3}%
}\left\vert w\right\vert ^{s}G\left(  t,dw\right)  \leq C_{\ast}$ and the
result follows.
\end{proof}

With this Proposition is rather easy to prove now the existence of the desired
self-similar solution, as stated in the Theorem below which is the main result
of this section, using Schauder fixed point Theorem. A similar idea has been
also used with adaptations in \cite{EMR}, \cite{EV}, \cite{GPV}, \cite{KV}, \cite{NV},
\cite{NTV}.

\smallskip

\begin{proofof}
[Proof of Theorem \ref{th:ssprof}]Suppose that $\zeta$ in (\ref{S4E6}) is
strictly positive, since for $\zeta=0$ we have $G=\delta\left(  w\right)  $
(see Remark \ref{ZeroEner}). We define the subset $\mathcal{U}$ of
$\mathcal{M}_{+}\left(  \mathbb{R}_{c}^{3}\right)  $ such that
\begin{equation}
\int_{\mathbb{R}^{3}}G\left(  dw\right)  =1\ ,\ \ \int_{\mathbb{R}^{3}}%
w_{j}G\left(  dw\right)  =0,\ \int_{\mathbb{R}^{3}}w_{j}w_{k}G\left(
dw\right)  =K\bar{N}_{j,k}\ \label{S2E9}%
\end{equation}
holds, as well as the inequality $\int_{\mathbb{R}_{c}^{3}}\left\vert
w\right\vert ^{s}G\left(  dw\right)  \leq C_{\ast}=C_{\ast}\left(
\zeta\right)  .$ We choose $K$ in (\ref{S2E9}) in order to have
\[
K\sum_{j=1}^{3}\bar{N}_{j,j}=\zeta.
\]
The set $\mathcal{U}$ is convex and closed in the $\ast-$weak topology of
measures. Moreover $\mathcal{U}$ is compact in this topology. We consider the
semigroup $\mathcal{S}_{\bar{\alpha}}\left(  t\right)  $ defined in Lemma
\ref{Semig}. For any $h>0$ (arbitrarily small) we have that the operator
$\mathcal{S}_{\bar{\alpha}}\left(  h\right)  $ transforms $\mathcal{U}$ in
itself. Given that $\mathcal{S}_{\bar{\alpha}}\left(  h\right)  $ is compact,
we can apply Schauder theorem to prove the existence of $G_{\ast}^{\left(
h\right)  }\in\mathcal{U}$ such that $\mathcal{S}_{\bar{\alpha}}\left(
h\right)  G_{\ast}^{\left(  h\right)  }=G_{\ast}^{\left(  h\right)  }.$
Moreover, since $\mathcal{S}_{\bar{\alpha}}\left(  h\right)  $ defines a
semigroup we have $\mathcal{S}_{\bar{\alpha}}\left(  mh\right)  G_{\ast
}^{\left(  h\right)  }=G_{\ast}^{\left(  h\right)  }$ for any integer $m.$ We
then take a subsequence $\left\{  h_{k}\right\}  $ such that $h_{k}%
\rightarrow0$ and the corresponding sequence of fixed points $\left\{
G_{\ast}^{\left(  h_{k}\right)  }\right\}  .$ This sequence is compact in
$\mathcal{U}$ and, taking a subsequence if needed (but denoted still as
$\left\{  h_{k}\right\}  $), we obtain that it converges to some $G_{\ast}.$
Given any $t>0$ we can obtain integers $m_{k}$ such that $m_{k}h_{k}%
\rightarrow t.$ We have $\mathcal{S}_{\bar{\alpha}}\left(  m_{k}h_{k}\right)
G_{\ast}^{\left(  h_{k}\right)  }=G_{\ast}^{\left(  h_{k}\right)  }\rightarrow
G_{\ast}$ and on the other hand
\[
\mathcal{S}_{\bar{\alpha}}\left(  m_{k}h_{k}\right)  G_{\ast}^{\left(
h_{k}\right)  }=\left[  \mathcal{S}_{\bar{\alpha}}\left(  m_{k}h_{k}\right)
-\mathcal{S}_{\bar{\alpha}}\left(  t\right)  \right]  G_{\ast}^{\left(
h_{k}\right)  }+\mathcal{S}_{\bar{\alpha}}\left(  t\right)  G_{\ast}^{\left(
h_{k}\right)  }.
\]

The last term converges to $\mathcal{S}_{\bar{\alpha}}\left(  t\right)
G_{\ast}$ using the weak continuity of the semigroup $\mathcal{S}_{\bar
{\alpha}}\left(  t\right)  $ (cf. Lemma \ref{Semig}). On the other hand we
have that $\left[  \mathcal{S}_{\bar{\alpha}}\left(  m_{k}h_{k}\right)
-\mathcal{S}_{\bar{\alpha}}\left(  t\right)  \right]  G_{\ast}^{\left(
h_{k}\right)  }\rightarrow0$ as $k\rightarrow\infty$ in the weak topology due
to the uniformicity of the estimate (\ref{S2E7a}).
Then $\mathcal{S}_{\bar{\alpha}}\left(  t\right)  G_{\ast}=G_{\ast}$ for any
$t>0.$ Then $G_{\ast}$ is a stationary point for the semigroup. Notice that we
can pass to the limit in (\ref{S2E9}).
\end{proofof}

\bigskip

\subsection{Behavior of the density and internal energy for homoenergetic solutions}

\bigskip

In the next section we will apply the tools developed in the previous
subsections to the different homoenergetic flows described in Section
\ref{ss:classeqsol}. By the reader's convenience we recall that the equation
describing homoenergetic flows is:
\begin{equation}
\partial_{t}g-L\left(  t\right)  w\cdot\partial_{w}g=\mathbb{C}g\left(
w\right)  \ . \label{S8E9}%
\end{equation}

We also recall (cf. (\ref{S8E7})) that the kernel $B$ in (\ref{A0_0}) is
homogeneous with homogeneity $\gamma.$ We want construct solutions of
(\ref{D1_0}) with the different choices of $L\left(  t\right)  $ in Theorem
\ref{ClassHomEne}. The solutions in which we are interested have some suitable
scaling properties, and two quantities which play a crucial role determining
how are these rescalings are the density $\rho\left(  t\right)  $ and the
internal energy $\varepsilon\left(  t\right)  .$ These are given by (cf.
(\ref{S8E3})):%
\begin{equation}
\rho\left(  t\right)  =\int_{\mathbb{R}^{3}}g\left(  t,dw\right)
\ \ ,\ \ \ \varepsilon\left(  t\right)  =\int_{\mathbb{R}^{3}}\left\vert
w\right\vert ^{2}g\left(  t,dw\right)  \label{S8E4}%
\end{equation}
which will be assumed to be finite for each given $t$ in all the solutions
considered in this paper. Integrating (\ref{D1_0}) and using the conservation
of mass property of the collision kernel, we obtain:%
\begin{equation}
\partial_{t}\rho\left(  t\right)  +\Tr \left(  L\left(  t\right)  \right)
\rho\left(  t\right)  =0 \label{S8E5}%
\end{equation}
whence:%
\begin{equation}
\rho\left(  t\right)  =\rho\left(  0\right)  \exp\left(  -\int_{0}%
^{t}\Tr \left(  L\left(  s\right)  \right)  ds\right)  . \label{S8E6}%
\end{equation}

Nevertheless it is not possible to derive a similarly simple equation for the
internal energy $\varepsilon\left(  t\right)  ,$ because the term $-L\left(
t\right)  w\cdot\partial_{w}g$ on the left-hand side of (\ref{D1_0}) yields in
general terms which cannot be written neither in terms of $\rho\left(
t\right)  ,\ \varepsilon\left(  t\right)  .$ Actually these terms have an
interesting physical meaning, because they produce heating or cooling of the
system and therefore they contribute to the change of $\varepsilon\left(
t\right)  .$ To obtain the precise form of these terms we need to study the
detailed form of the solutions of (\ref{D1_0}). The rate of growth or decay of
$\varepsilon\left(  t\right)  $ would then typically appear as an eigenvalue
of the corresponding PDE problem.

\bigskip

\section{Applications: Self-similar solutions of homoenergetic
flows \label{SelfSimEx}}

The self-similar solutions which we construct in this paper are characterized
by a balance between the terms $-L\left(  t\right)  w\cdot\partial_{w}g$ and
$\mathbb{C}g\left(  w\right)  $ in (\ref{S8E9}). Such a balance is only
possible for specific choices of the homogeneity of the kernel $\gamma.$
Actually in all the cases in which we prove the existence of self-similar
solutions in this paper we have $\gamma=0,$ i.e. Maxwell molecules.

\subsection{Simple shear \label{ss:SelfSimSimpSh}}

In this case, combining (\ref{D1_0}) and (\ref{T1E5}) we obtain:%
\begin{equation}
\partial_{t}g-Kw_{2}\partial_{w_{1}}g=\mathbb{C}\left[  g\right]  . \label{W0}%
\end{equation}

Notice that in this case (\ref{S8E6}) reduces to%
\begin{equation}
\rho\left(  t\right)  =\rho\left(  0\right)  =1 \label{W0a}%
\end{equation}
where, without loss of generality, we can use the normalization $\rho\left(
0\right)  =1$ rescaling the time unit. Using (\ref{A0_0}), the definition of
$\rho\left(  t\right)  $ in (\ref{S8E4}) and (\ref{W0a}) we obtain that the
physical dimensions of the three terms in (\ref{W0}) are:%
\begin{equation}
\frac{\left[  g\right]  }{\left[  t\right]  },\ \left[  g\right]  ,\ \left[
w\right]  ^{\gamma}\left[  g\right]  . \label{W0b}%
\end{equation}

Notice that if $\left\vert w\right\vert $ changes in time, we can have a
balance of second and third terms in (\ref{W0b}) only if $\gamma=0,$ i.e. for
Maxwell molecules. On the other hand (\ref{W0b}) indicates that we cannot
obtain a balance between the first two terms of this equation with power law
behaviors for $\left[  w\right]  ,$ and the only way to obtain such a balance
will be assuming that $\left[  w\right]  $ scales like an exponential of $t.$
In the case of $\gamma=0$ we consider solutions with the following scaling%

\begin{equation}
g\left(  w,t\right)  =e^{-3\beta t}G\left(  \xi\right)  \ \ ,\ \ \xi=\frac
{w}{e^{\beta t}} \label{S8E8}%
\end{equation}
where $\beta\in\mathbb{R},$ which characterizes the behavior of the internal
energy is an eigenvalue to be determined. The factor $e^{-3\beta t}$ has been
chosen in order to have the density conservation condition (\ref{W0a}).

Plugging (\ref{S8E8}) into (\ref{W0}) we obtain:%
\begin{equation}
-\beta\partial_{\xi}\left(  \xi G\right)  -K\partial_{\xi_{1}}\left(  \xi
_{2}G\right)  =\mathbb{C}\left[  G\right]  . \label{W1}%
\end{equation}
where (\ref{W0a}) implies the normalization condition:%
\begin{equation}
\int_{\mathbb{R}^{3}}G\left(  \xi\right)  d\xi=1 . \label{W2}%
\end{equation}

Notice that given that the homogeneity of the kernel is $\gamma=0,$ it is not
possible to eliminate the constant $K$ in (\ref{W1}) by means of a scaling
argument which preserves the normalization (\ref{W2}). The equation \eqref{W1}
is a particular case of (\ref{S4E5}). We can then apply Theorem
\ref{th:ssprof} which yields immediately the following result.

\bigskip

\begin{theorem}
\ \label{ThSimpShear} Suppose that $B$ in (\ref{eq:collboltzgen}) is
homogeneous of order zero and that $b$ in (\ref{S6E1}) is strictly positive.
There exists $k_{0}>0$ small such that for any $\zeta\geq0$ and $K\in
{\mathbb{R}}$ such that $\frac{K}{b}\leq k_{0}$ there exists $\beta
\in\mathbb{R}$ and $G\in\mathcal{M}_{+}\left(  \mathbb{R}_{c}^{3}\right)  $
which solves (\ref{W1}) in the sense of measures and satisfies the
normalization condition (\ref{W2}) as well as:%
\begin{equation}
\int_{\mathbb{R}^{3}}w_{j}G\left(  dw\right)  =0,\ \int_{\mathbb{R}^{3}%
}\left\vert w\right\vert ^{2}G\left(  dw\right)  =\zeta\ . \label{W3}%
\end{equation}

\end{theorem}

\begin{remark}
We observe that in the Theorem above the assumption $\int_{\mathbb{R}^{3}%
}w_{j}G\left(  dw\right)  =0$ is not restrictive. Indeed, if this assumption
is not satisfied, we can compute the evolution equation for the first order
moments and we get $\partial_{t} \left( \int_{\mathbb{R}^{3}}wG\left(
dw\right)  \right) +L(t)\int_{\mathbb{R}^{3}}wG\left(  dw\right) =0$.
Furthermore, in the case of simple shear considere here, we have
$L(t)\int_{\mathbb{R}^{3}}wG\left(  dw\right) =\big(K\int_{\mathbb{R}^{3}%
}w_{2}\,G\left(  dw\right) ,0,0 \big)^{T}.$ Therefore, $\left(  \int
_{\mathbb{R}^{3}}wG\left(  dw\right)  \right) (t)=\exp\left( - \int_{0}^{t}
L(s)ds\right)  \left(  \int_{\mathbb{R}^{3}}wG\left(  dw\right)  \right)
_{\vert_{t=0}}.$ We now set $\left(  \int_{\mathbb{R}^{3}}wG\left(  dw\right)
\right) _{\vert_{t=0}}=
\left( \gamma_{1},\gamma_{2},\gamma_{3}\right) ^{T} \neq\left( 0,0,0\right)
^{T}$ and introduce the propagated solution $\bar{G}$ such that $G(w_{1}%
,w_{2},w_{3},t)=\bar{G}(w_{1}+\gamma_{1},w_{2}+\gamma_{2}+\gamma_{1}
t,w_{3}+\gamma_{3},t).$ It is then straightforward to show that $\bar{G}$
satisfies \eqref{W0}.
\end{remark}

Therefore, solutions of (\ref{W0}) with the form (\ref{S8E8}) exist, at least
if the shear parameter $K$ is sufficiently small compared with the parameter
$b$ which measures the strength of the collision term. Actually we can give a
physical meaning to the condition $\frac{K}{b}$ in terms of a nondimensional
parameter. The parameter $K$ is, up to a multiplicative constant, the inverse
of the time scale $\tau_{shear}$ in which the effect of the shear deformes a
sphere into a ellipsoid for which the largest semiaxes has double length than
the shortest one. On the other hand $b$ is the inverse of the average time
between collisions $\tau_{coll}.$ Then $\frac{K}{b}=\frac{\tau_{coll}}%
{\tau_{shear}}$ and therefore the smallness condition in Theorem
\ref{ThSimpShear} just means:%
\begin{equation}
\frac{K}{b}=\frac{\tau_{coll}}{\tau_{shear}}\; \text{ small} . \label{W4}%
\end{equation}

We remark that the value of $\beta$ can be computed explicitly. Indeed, we
have seen in subsection that the eigenvalue $\beta$ in Theorem \ref{th:ssprof}
is the solution $\alpha$ of the eigenvalue problem (\ref{S6E5}), (\ref{S6E6})
with the largest real part. In the particular case of the equation (\ref{W1})
the problem (\ref{S6E5}), (\ref{S6E6}) with the normalization condition
(\ref{S7E3a}) takes the form:%
\begin{align}
\left(  \frac{\alpha}{b}+1\right)  \Gamma_{1,1}+\frac{K}{b}\Gamma_{1,2}  &
=\Gamma\ \ ,\ \ \Gamma=\frac{1}{3}\left(  \Gamma_{1,1}+\Gamma_{2,2}%
+\Gamma_{3,3}\right) \nonumber\\
\left(  \frac{\alpha}{b}+1\right)  \Gamma_{1,2}+\frac{K}{2b}\Gamma_{2,2}  &
=0\ ,\ \left(  \frac{\alpha}{b}+1\right)  \Gamma_{1,3}+\frac{K}{2b}%
\Gamma_{2,3}=0\nonumber\\
\left(  \frac{\alpha}{b}+1\right)  \Gamma_{2,2}  &  =\Gamma\ ,\ \left(
\frac{\alpha}{b}+1\right)  \Gamma_{2,3}=0\ ,\ \left(  \frac{\alpha}%
{b}+1\right)  \Gamma_{3,3}=\Gamma\label{W5}%
\end{align}
with:%
\begin{equation}
\Gamma_{j,k}=\Gamma_{k,j}\ \ ,\ \ j,\ k=1,2,3 . \label{W5a}%
\end{equation}

The eigenvalue problem (\ref{W5}) (or more precisely an equivalent formulation
of it) has been studied in detail in \cite{TM}, Chapter XIV. We summarize some
relevant information about the solutions of (\ref{S6E5}), (\ref{S6E6}) which
will be used later.

\begin{proposition}
\label{EigSimShear}The eigenvalues of the problem are $\alpha\in\left\{
-b,b\left(  \lambda_{1}-1\right)  ,b\left(  \lambda_{2}-1\right)  ,b\left(
\lambda_{3}-1\right)  \right\}  $ where we denote as $\lambda_{j},\ j=1,2,3$
the roots of
\begin{equation}
\lambda^{3}=\lambda^{2}+\frac{K^{2}}{6b^{2}}. \label{W6}%
\end{equation}

The equation (\ref{W6}) has for any $K\neq0$ a real root $\lambda_{1}>1$ and
two complex conjugates roots $\lambda_{2},\lambda_{3}$ with $\operatorname{Im}%
\left(  \lambda_{2}\right)  =-\operatorname{Im}\left(  \lambda_{3}\right)  >0$
and $\operatorname{Re}\left(  \lambda_{2}\right)  =\operatorname{Re}\left(
\lambda_{3}\right)  <0.$

The subspace of eigenvectors associated to the eigenvalue $\alpha=-b$ is the
two-dimensional (complex) subspace $\left\{  N_{1,1}=\mu_{1},\ N_{3,3}%
=-\mu_{1},\ N_{1,3}=\mu_{2},\ N_{1,2}=N_{2,2}=N_{2,3}=0,\ \mu_{1},\mu_{2}%
\in\mathbb{C}\right\}  .$

We have the following asymptotic formulas for $\lambda_{1}:$%
\begin{equation}
\lambda_{1}\sim1+\frac{K^{2}}{6b^{2}}+...\text{ as }K\rightarrow
0\ \ ,\ \ \lambda_{1}\sim\frac{K^{\frac{2}{3}}}{\left(  6b^{2}\right)
^{\frac{1}{3}}}\ \text{as\ }K\rightarrow\infty. \label{W7}%
\end{equation}

\end{proposition}

\begin{remark}
We assume that the vector spaces are complex, given that some of the
eigenvalues are complex. Notice that the subspace spanned by all the
eigenvectors has dimension five, in spite of the fact that the underlying
space is six-dimensional (see Remark \ref{dimSubsp}).
\end{remark}

\begin{proof}
The claim about the set of eigenvalues follows using the change of variables
$\lambda=\frac{\alpha}{b}+1$ and distinguishing the cases $\lambda=0$ and
$\lambda\neq0.$ In the first case we obtain $N_{1,2}=N_{2,2}=N_{2,3}%
=N_{1,1}+N_{3,3}=0$ and this yields the structure of eigenvectors in the case
$\alpha=-b.$

It is immediate to check,just plotting the function $\left(  \lambda
^{3}-\lambda^{2}\right)  $ that there is a unique real solution $\lambda_{1} $
of (\ref{W6}) which satisfy $\lambda_{1}>1.$ Then $\operatorname{Re}\left(
\lambda_{2}\right)  =\operatorname{Re}\left(  \lambda_{3}\right)  $ and using
that $\lambda_{1}+\lambda_{2}+\lambda_{3}=1$ we obtain $\lambda_{1}%
+2\operatorname{Re}\left(  \lambda_{2}\right)  =1$ whence $\operatorname{Re}%
\left(  \lambda_{2}\right)  =\frac{1}{2}\left(  1-\lambda_{1}\right)  <0.$

The asymptotic formulas (\ref{W7}) follow from elementary arguments.
\end{proof}

Notice that Proposition \ref{EigSimShear} implies that the largest eigenvalue
of the problem (\ref{W5}) is $\bar{\alpha}=b\left(  \lambda_{1}-1\right)  >0.$
Since $\beta=\bar{\alpha}$ we then obtain that $\beta>0$ in (\ref{S8E8}). This
implies that the average of $\left\vert w\right\vert ^{2}$ increases as $t$
increases, something which might be expected, since the effect of the shear in
the gas yields an increase of the internal energy of the system.

\bigskip

Theorem \ref{ThSimpShear} requires a strong smallness condition on $\frac
{K}{b}$ (cf. (\ref{W4})). Actually this smallness assumption can be removed,
but this requires to derive a more sophisticated version of Povzner estimates
which takes into account the effect of the shear. This is the next point which
we consider.

\bigskip

\subsubsection{Sufficient condition to have self-similar solutions for
arbitrary shear parameters}

\label{sss:sspsimpleshear}

\bigskip

We now formulate a sufficient condition for the existence of self-similar
solutions in the case of simple shear for arbitrary values of the shear
parameter $K.$ The stated condition depends on the collision kernel $B.$ In
order to formulate this condition we introduce the following quadratic form:%
\begin{equation}
W_{0}\left(  \xi,\eta\right)  =\sum_{j\leq k}u_{j,k}\left(  \xi\otimes
\eta\right)  \left(  e_{j},e_{k}\right)  \label{S9E3b}%
\end{equation}
:where the quadratic forms $\left(  \xi\otimes\eta\right)  \left(  e_{j}%
,e_{k}\right)  $ are as in (\ref{T4E8}), (\ref{T4E8a}) (cf. also the quadratic
forms $W_{j,k}$ in Proposition \ref{TensComp}) and $u_{j,k}\in\mathbb{R}$ are
given by:%
\begin{equation}
u_{1,1}=1\ ,\ u_{2,2}=\left(  3\lambda_{1}-2\right)  \ ,\ u_{3,3}%
=1\ ,\ u_{1,2}=-\frac{K}{b\lambda_{1}}\ ,\ \ u_{1,3}=u_{2,3}=0 \label{S9E5}%
\end{equation}
where $\lambda_{1}$ is as in Proposition \ref{EigSimShear}. The quadratic form
$W_{0}$ is positive definite. This follows writing $W_{0}$ in matrix form as:%
\[
\bar{W}_{0}=\left(
\begin{array}
[c]{ccc}%
1 & -\frac{K}{2b\lambda_{1}} & 0\\
-\frac{K}{2b\lambda_{1}} & \left(  3\lambda_{1}-2\right)  & 0\\
0 & 0 & 1
\end{array}
\right) .
\]

Since $\lambda_{1}>0,$ in order to check that $W_{0}$ is positive definite we
only need to check that the determinant of $\bar{W}_{0}$ is positive. We have:%
\begin{equation}
\det\left(  \bar{W}_{0}\right)  =\left(  3\lambda_{1}-2\right)  -\frac{K^{2}%
}{4b^{2}}\frac{1}{\left(  \lambda_{1}\right)  ^{2}}\ . \label{V1E1}%
\end{equation}
Then, using (\ref{W6}).
\[
\det\left(  \bar{W}_{0}\right)  =\frac{3}{\left(  \lambda_{1}\right)  ^{2}%
}\left[  \left(  \lambda_{1}\right)  ^{3}-\frac{2}{3}\left(  \lambda
_{1}\right)  ^{2}-\frac{K^{2}}{12b^{2}}\right]  =\frac{3}{\left(  \lambda
_{1}\right)  ^{2}}\left[  \frac{\left(  \lambda_{1}\right)  ^{2}}{3}%
+\frac{K^{2}}{12b^{2}}\right]  >0.
\]
Therefore:%
\begin{equation}
\sup_{\left\vert \xi\right\vert =1}W_{0}\left(  \xi\right)  \geq c_{1}>0.
\label{S9E6}%
\end{equation}

In order to formulate a stability criterium which would yield the existence of
self-similar solutions for a given value of the shear we define the following
function:%
\begin{align}
\mathcal{H}\left(  \xi;K\right)   &  =\int_{S^{2}}d\omega B\left(  \omega
\cdot\xi\right)  \left(  W_{0}\left(  P_{\omega}^{\bot}\xi\right)  \log\left(
\frac{W_{0}\left(  P_{\omega}^{\bot}\xi\right)  }{W_{0}\left(  \xi\right)
}\right)  +W_{0}\left(  P_{\omega}\xi\right)  \log\left(  \frac{W_{0}\left(
P_{\omega}\xi\right)  }{W_{0}\left(  \xi\right)  }\right)  \right)
-\label{V1E2}\\
&  \quad-\int_{S^{2}}d\omega B\left(  \omega\cdot\xi\right)  \left(
W_{0}\left(  P_{\omega}^{\bot}\xi\right)  +W_{0}\left(  P_{\omega}\xi\right)
-W_{0}\left(  \xi\right)  \right)  \ \ ,\ \ \xi\in\mathbb{R}^{3}%
\diagdown\left\{  0\right\} \nonumber
\end{align}
and, for any quadratic form $W$ we define:
\begin{equation}
\mathcal{W}\left(  \xi;W;K\right)  =\int_{S^{2}}deB\left(  e\cdot\xi\right)
\left(  W\left(  P_{e}^{\bot}\xi\right)  +W\left(  P_{e}\xi\right)  -W\left(
\xi\right)  \right)  -K\xi_{2}\partial_{\xi_{1}}W\left(  \xi\right)  -2\beta
W\left(  \xi\right)  . \label{V1E3}%
\end{equation}

We then have the following result.

\begin{theorem}
\ \label{ThSimpShearArbK} Suppose that $B$ in (\ref{eq:collboltzgen}) is
homogeneous of order zero (i.e., $\gamma=0$) and that $b$ in (\ref{S6E1}) is strictly positive.  Let $\mathcal{H}$ and $\mathcal{W}$ be as in \eqref{V1E2} and \eqref{V1E3} respectively. 
Suppose that $K\in\mathbb{R}$ satisfies the following property:%
\begin{equation}
\inf_{W}\left[  \min_{\left\vert \xi\right\vert =1}\left[  \mathcal{W}\left(
\xi;W;K\right)  +\mathcal{H}\left(  \xi;K\right)  \right]  \right]  <0
\label{V1E4}%
\end{equation}
where the infimum is taken over all the quadratic forms $W.$ Then, for any
$\zeta\geq0$ there exists $\beta\in\mathbb{R}$ and $G\in\mathcal{M}_{+}\left(
\mathbb{R}_{c}^{3}\right)  $ which solves (\ref{W1}) in the sense of measures
and satisfies the normalization condition (\ref{W2}) as well as (\ref{W3}).
\end{theorem}

\bigskip

Theorem \ref{ThSimpShearArbK} allows to obtain quantitative estimates about
the value of the shear parameter $K$ for which self-similar solutions exist.
Notice that in particular, Theorem \ref{ThSimpShearArbK} implies Theorem
\ref{ThSimpShear}. Indeed, if $K=0$ we obtain that $W_{0}$ is just the
identity. Then the last integral in (\ref{V1E1}) vanishes due to Pithagoras
Theorem and we have $\log\left(  \frac{W_{0}\left(  P_{\omega}^{\bot}%
\xi\right)  }{W_{0}\left(  \xi\right)  }\right)  <0$,\ $\log\left(
\frac{W_{0}\left(  P_{\omega}\xi\right)  }{W_{0}\left(  \xi\right)  }\right)
<0$ if $\left\vert \xi\right\vert =1.$ Then $\mathcal{H}\left(  \xi;0\right)
<0.$ Then the inequality (\ref{V1E4}) holds for $\left\vert K\right\vert $
sufficiently small whence Theorem \ref{ThSimpShear} follows.\bigskip

The main difference between the proof of the Theorems \ref{th:ssprof} and
\ref{ThSimpShearArbK} is the fact that instead of using the classical Povzner
estimates (cf. Proposition \ref{Povzner}) in order to control the dynamics of
large particles, we will use a modified version which takes into account not
only the collisions between particles, but also the effect of the shear term
$-K\partial_{\xi_{1}}\left(  \xi_{2}G\right)  .$ The result is the following.\ 

\bigskip

\begin{proposition}
\label{PovnerShear} Suppose that $B$ in (\ref{eq:collboltzgen}) is homogeneous
of order zero and that $b$ in (\ref{S6E1}) is strictly positive. Suppose that
for a given $K\in\mathbb{R}$ the condition (\ref{V1E4}) holds. Then there
exists a function $\varphi:\mathbb{R}^{3}\rightarrow\mathbb{\ R}$,
$\varphi=\varphi\left(  \xi\right)  $ homogeneous in $\xi$ with homogeneity
$s\in\left(  2,3\right)  $ (depending on $K$) and positive constants
$\kappa,\ C,$ depending also on $K,$ such that, for any measure $\xi,\xi
_{\ast}\in\mathbb{R}^{3}$ satisfying (\ref{W2}) and $\int_{\mathbb{R}^{3}}\xi
G\left(  d\xi\right)  =0$ the following inequality holds:%
\begin{equation}
\mathcal{U}\left[  \varphi\right]  \left(  \xi,\xi_{\ast}\right)  \leq
-\kappa\varphi\left(  \xi\right)  +C\left\vert \xi\right\vert ^{\frac{s}{2}%
}\left\vert \xi_{\ast}\right\vert ^{\frac{s}{2}}\ \ \text{for\ }\left\vert
\xi_{\ast}\right\vert \leq\left\vert \xi\right\vert \label{S9E1}%
\end{equation}
where:
\begin{equation}
\mathcal{U}\left[  \varphi\right]  \left(  \xi,\xi_{\ast}\right)
:=\int_{S^{2}}d\omega B\left(  \omega\cdot\left(  \xi-\xi_{\ast}\right)
\right)  \left(  \varphi^{\prime}+\varphi_{\ast}^{\prime}-\varphi
-\varphi_{\ast}\right)  -K\left[  \xi_{2}\partial_{\xi_{1}}\varphi\right]
\left(  \xi\right)  -\beta\xi\cdot\partial_{\xi}\varphi\left(  \xi\right)
\label{S9E2}%
\end{equation}
and where $\beta=b\left(  \lambda_{1}-1\right)  $ with $\lambda_{1}$ as in
Proposition \ref{EigSimShear}.

Moreover, there exists $c_{0}>0$ such that $\varphi\left(  \xi\right)  \geq
c_{0}\left\vert \xi\right\vert ^{s}$ for any $\xi\in\mathbb{R}^{3}.$
\end{proposition}

\begin{proof}
If $K=0$ we have $\beta=b\left(  \lambda_{1}-1\right)  =0$\ and the result
just follows from the classical Povzner estimates (cf. Proposition
\ref{Povzner}). Therefore we will assume that $K\neq0$ whence $\lambda_{1}>1.$
We will prove Proposition \ref{PovnerShear} in two steps.

\textbf{Step 1:} We first prove that the positive definite quadratic form
$W_{0}$ defined in (\ref{S9E3b}), (\ref{S9E5}) satisfies:%
\begin{equation}
\Phi\left(  \xi;W_{0}\right)  =0\ \ \text{for\ all }\xi\in\mathbb{R}^{3}
\label{S9E3}%
\end{equation}
where:
\begin{equation}
\Phi\left(  \xi;W_{0}\right)  =-2\int_{S^{2}}d\omega B\left(  \frac
{\omega\cdot\xi}{\left\vert \xi\right\vert }\right)  \left[  W_{0}\left(
P_{\omega}^{\bot}\xi,P_{\omega}\xi\right)  \right]  -K\left[  \xi_{2}%
\partial_{\xi_{1}}W_{0}\right]  \left(  \xi\right)  -2\beta W_{0}\left(
\xi\right)  . \label{S9E3a}%
\end{equation}

We look for $W_{0}$ in the form. Using Proposition \ref{TensComp} and the
definition (\ref{S9E3b}), (\ref{S9E5}) we obtain:%
\begin{align*}
\Phi\left(  \xi;W_{0}\right)   &  =-2\int_{S^{2}}d\omega B\left(  \frac
{\omega\cdot\xi}{\left\vert \xi\right\vert }\right)  \left[  W_{0}\left(
P_{\omega}^{\bot}\xi,P_{\omega}\xi\right)  \right]  -K\left[  \xi_{2}%
\partial_{\xi_{1}}W_{0}\right]  \left(  \xi\right)  -2\beta W_{0}\left(
\xi\right) \\
&  =-2\sum_{j\leq k}u_{j,k}\int_{S^{2}}d\omega B\left(  \frac{\omega\cdot\xi
}{\left\vert \xi\right\vert }\right)  \left(  P_{\omega}^{\bot}\xi\otimes
P_{\omega}\xi\right)  \left(  e_{j},e_{k}\right) \\
&  \quad-K\sum_{j\leq k}u_{j,k}\left[  \xi_{2}\xi_{k}\delta_{1,j}+\xi_{2}%
\xi_{j}\delta_{1,k}\right]  -2\beta\sum_{j\leq k}u_{j,k}\xi_{j}\xi_{k}\\
&  =-2b\sum_{j\leq k}u_{j,k}\left[  \xi_{j}\xi_{k}-\frac{\left\vert
\xi\right\vert ^{2}}{3}\delta_{j,k}\right]  -K\sum_{k}u_{1,k}\xi_{2}\xi
_{k}-Ku_{1,1}\xi_{1}\xi_{2}-2\beta\sum_{j\leq k}u_{j,k}\xi_{j}\xi_{k}.
\end{align*}
Then:
\begin{align}
\Phi\left(  \xi;W_{0}\right)   &  =-\left(  2b+2\beta\right)  \sum_{j\leq
k}u_{j,k}\xi_{j}\xi_{k}+\frac{2b\left\vert \xi\right\vert ^{2}}{3}\sum
_{j}u_{j,j}-K\sum_{k}u_{1,k}\xi_{k}\xi_{2}-Ku_{1,1}\xi_{1}\xi_{2}\nonumber\\
&  =-\left(  2b+2\beta\right)  \sum_{j\leq k}u_{j,k}W_{j,k}+\frac{2b}%
{3}\left[  \sum_{\ell}u_{\ell,\ell}\right]  \sum_{j\leq k}\delta_{j,k}%
W_{j,k}\nonumber \\&\quad -K\sum_{j\leq k}u_{1,k}\delta_{j,2}W_{j,k}-2Ku_{1,1}W_{1,2}.
\label{S9E4a}%
\end{align}

Since the quadratic forms $W_{j,k}$ are linearly independent in the space of
quadratic forms we obtain that $\Phi\left(  \xi;W_{0}\right)  =0$ for any
$\xi\in\mathbb{R}^{3}$ if:%
\begin{equation}
-\left(  \frac{\beta}{b}+1\right)  u_{j,k}+\frac{\delta_{j,k}}{3}\sum_{\ell
}u_{\ell,\ell}-\frac{K}{2b}u_{1,k}\delta_{j,2}-\frac{K}{b}\delta_{1,j}%
\delta_{2,k}u_{1,1}=0\ ,\ j\leq k . \label{S9E4}%
\end{equation}

The eigenvalue problem is the adjoint problem (\ref{W5}), (\ref{W5a}) assuming
that in space of quadratic forms we take the scalar product $\sum_{j}\sum
_{k}\Gamma_{j,k}\tilde{\Gamma}_{j,k}=\left\langle \Gamma,\tilde{\Gamma
}\right\rangle .$ 
Using that $\beta=b\left(  \lambda_{1}-1\right)  $ (cf. Proposition
\ref{EigSimShear}) we then readily obtain that $u_{j,k}$ in (\ref{S9E5}) yield
a nontrivial solution of (\ref{S9E4}). 
Therefore, the quadratic form given by (\ref{S9E3b}) with the coefficients
$u_{j,k}$ in (\ref{S9E5}) satisfies (\ref{S9E3}).

\bigskip\textbf{Step 2. } Suppose now that the stability condition
(\ref{V1E4}) holds. Then, there exists a quadratic form $W_{1}$ such that:%
\begin{equation}
\min_{\left\vert \xi\right\vert =1}\left[  \mathcal{W}\left(  \xi
;W_{1};K\right)  +\mathcal{H}\left(  \xi;K\right)  \right]  \leq-c_{0}<0
\label{V1E5a}%
\end{equation}
for any $\xi\in\mathbb{R}^{3}$ satisfying $\left\vert \xi\right\vert =1.$ We
define a function $\varphi$ homogeneous with homogeneity $s=2\left(
1+\varepsilon\right)  >2$ in the form:%
\begin{equation}
\varphi\left(  \xi\right)  =\left[  W_{0}\left(  \xi\right)  +\varepsilon
W_{1}\left(  \xi\right)  \right]  ^{1+\varepsilon} . \label{V1E5}%
\end{equation}

Given a function $\varphi:\mathbb{R}^{3}\rightarrow\mathbb{R}$ we define:%
\begin{equation}
\mathcal{V}\left(  \xi;\varphi\right)  =\int_{S^{2}}deB\left(  e\cdot
\xi\right)  \left(  \varphi\left(  P_{e}^{\bot}\xi\right)  +\varphi\left(
P_{e}\xi\right)  -\varphi\left(  \xi\right)  \right)  -K\xi_{2}\partial
_{\xi_{1}}\varphi\left(  \xi\right)  -\beta\xi\cdot\nabla\varphi\left(
\xi\right)  \label{V1E6}%
\end{equation}
where $\beta=b\left(  \lambda_{1}-1\right)  .$ If $\varphi$ is homogeneous
with homogeneity $s=2\left(  1+\varepsilon\right)  $ we can rewrite
(\ref{V1E6}) as
\[
\mathcal{V}\left(  \xi;\varphi\right)  =\int_{S^{2}}deB\left(  e\cdot
\xi\right)  \left(  \varphi\left(  P_{e}^{\bot}\xi\right)  +\varphi\left(
P_{e}\xi\right)  -\varphi\left(  \xi\right)  \right)  -K\xi_{2}\partial
_{\xi_{1}}\varphi\left(  \xi\right)  -2\left(  1+\varepsilon\right)
\beta\varphi\left(  \xi\right) .
\]

Using then (\ref{S9E6}), (\ref{V1E3}), (\ref{V1E5}) and Taylor's Theorem we
obtain the following approximation of $\mathcal{V}\left(  \xi;\varphi\right)
:$
\begin{align*}
\mathcal{V}\left(  \xi;\varphi\right)  =  &  \mathcal{W}\left(  \xi
;W_{0};K\right)  +\varepsilon\int_{S^{2}}deB\left(  e\cdot\xi\right)
\Big( W_{0}\left(  P_{e}^{\bot}\xi\right)  \log\left(  W_{0}\left(
P_{e}^{\bot}\xi\right)  \right) \\
&  +W_{0}\left(  P_{e}\xi\right)  \log\left(  W_{0}\left(  P_{e}\xi\right)
\right)  -W_{0}\left(  \xi\right)  \log\left(  W_{0}\left(  \xi\right)
\right)  \Big)\\
&  -K\xi_{2}\varepsilon\left[  1+\log\left(  W_{0}\left(  \xi\right)  \right)
\right]  \partial_{\xi_{1}}W_{0}\left(  \xi\right)  -2\beta\varepsilon
W_{0}\left(  \xi\right)  \log\left(  W_{0}\left(  \xi\right)  \right)
-2\beta\varepsilon W_{0}\left(  \xi\right)  +\\
&  +\varepsilon\mathcal{W}\left(  \xi;W_{1};K\right)  +O\left(  \varepsilon
^{2}\right)
\end{align*}
as $\varepsilon\rightarrow0.$

We now use the fact that $\mathcal{W}\left(  \xi;W_{0};K\right)  =0$, also to
rewrite $K\xi_{2}\partial_{\xi_{1}}W_{0}\left(  \xi\right)  +2\beta
W_{0}\left(  \xi\right)  $ as an integral term we then obtain, using
(\ref{V1E2}):%
\begin{equation}
\frac{\mathcal{V}\left(  \xi;\varphi\right)  }{\varepsilon}=\mathcal{H}\left(
\xi;K\right)  +\mathcal{W}\left(  \xi;W_{1};K\right)  +O\left(  \varepsilon
\right)  \text{ as }\varepsilon\rightarrow0 \label{V1E7}%
\end{equation}
uniformly in $\left\vert \xi\right\vert =1.$

Using (\ref{V1E5a}) we obtain:%
\begin{equation}
\mathcal{V}\left(  \xi;\varphi\right)  \leq-\frac{c_{0}}{2}<0\ \ \text{in\ \ }%
\left\vert \xi\right\vert =1 \label{V1E8}%
\end{equation}
if $\varepsilon>0$ is sufficiently small. Moreover, using (\ref{S9E6}) we
obtain also that, for $\left\vert \xi\right\vert =1$ and $\varepsilon$
sufficiently small we have:%
\begin{equation}
\varphi\left(  \xi\right)  \geq\frac{c_{1}}{2}>0\ \ \text{if \ }\left\vert
\xi\right\vert =1. \label{V1E9}%
\end{equation}

We can then prove (\ref{S9E1}). The right-hand side of (\ref{S9E2}) is
homogeneous in $\xi,\xi_{\ast}.$ We can then assume without loss of generality
that $\left\vert \xi\right\vert =1.$ Suppose first that $\left\vert \xi_{\ast
}\right\vert \leq\delta$ for some $\delta>0$ sufficiently small to be
determined. Then, using also (\ref{V1E6}):%
\[
\mathcal{U}\left[  \varphi\right]  \left(  \xi,\xi_{\ast}\right)
=\mathcal{V}\left(  \xi;\varphi\right)  +\mathcal{R}\left(  \xi,\xi_{\ast
};\varphi\right)
\]
where:%
\begin{align*}
\mathcal{R}\left(  \xi,\xi_{\ast};\varphi\right)  = &  \int_{S^{2}}d\omega
B\left(  \omega\cdot\left(  \xi-\xi_{\ast}\right)  \right)  \left(
\varphi^{\prime}+\varphi_{\ast}^{\prime}-\varphi-\varphi_{\ast}\right) \\
&  -\int_{S^{2}}deB\left(  e\cdot\xi\right)  \left(  \varphi\left(
P_{e}^{\bot}\xi\right)  +\varphi\left(  P_{e}\xi\right)  -\varphi\left(
\xi\right)  \right) .
\end{align*}

Using the collision rule (\ref{CM1}), (\ref{CM2}) as well as the continuity of
the function $\mathcal{R}\left(  \xi,\xi_{\ast};\varphi\right)  $ in
$\xi_{\ast}$ it then follows that $\mathcal{R}\left(  \xi,\xi_{\ast}%
;\varphi\right)  $ can be made arbitrarily small if $\left\vert \xi\right\vert
=1$ and $\delta$ is small enough. Then, using (\ref{V1E8}) we obtain:%
\[
\mathcal{U}\left[  \varphi\right]  \left(  \xi,\xi_{\ast}\right)  \leq
-\frac{c_{0}}{2}\ \ \text{if\ \ }\left\vert \xi\right\vert =1,\ \ \left\vert
\xi_{\ast}\right\vert \leq\delta
\]
and using the homogeneity of $\mathcal{U}\left[  \varphi\right]  $ as well as
(\ref{V1E9}) we then obtain:%
\[
\mathcal{U}\left[  \varphi\right]  \left(  \xi,\xi_{\ast}\right)  \leq
-\kappa\varphi\left(  \xi\right)  \ \ \text{if\ \ }\xi\neq0,\ \ \left\vert
\xi_{\ast}\right\vert \leq\delta\left\vert \xi\right\vert .
\]

On the other hand, if $\delta\left\vert \xi\right\vert <\left\vert \xi_{\ast
}\right\vert \leq\left\vert \xi\right\vert $ we just use that $\mathcal{U}%
\left[  \varphi\right]  \left(  \xi,\xi_{\ast}\right)  $ can be estimated as
$C\left[  \left\vert \xi\right\vert ^{s}+\left\vert \xi_{\ast}\right\vert
^{s}\right]  \leq C\left\vert \xi\right\vert ^{\frac{s}{2}}\left\vert
\xi_{\ast}\right\vert ^{\frac{s}{2}}.$ Therefore (\ref{S9E1}) follows.
\end{proof}

We can now prove Theorem \ref{ThSimpShearArbK}. \bigskip

\begin{proofof}
[Proof of Theorem \ref{ThSimpShearArbK}]We now argue as in the Proof of
Theorem \ref{th:ssprof}. The only difference in the argument arises in the
Proof of Proposition \ref{InvReg} where the inequalities (\ref{S7E5}),
(\ref{S7E6}) must be replaced by:%
\begin{equation}
\int_{\mathbb{R}^{3}}\varphi\left(  w\right)  G_{0}\left(  dw\right)  \leq
C_{\ast}\ \label{V2E1}%
\end{equation}
and%
\begin{equation}
\int_{\mathbb{R}^{3}}\varphi\left(  w\right)  \mathcal{S}\left(  t\right)
G_{0}\left(  dw\right)  \leq C_{\ast}\ \label{V2E2}%
\end{equation}
respectively. In order to prove that (\ref{V2E1}) implies (\ref{V2E2}) we
compute the derivative of the function $M_{s}\left(  t\right)  =\int
_{\mathbb{R}^{3}}\varphi\left(  w\right)  G\left(  t,dw\right)  .$ Then, using
(\ref{T3E6}) we obtain:%
\begin{align*}
\partial_{t}M_{s}\left(  t\right)  =  &  \int_{\mathbb{R}^{3}}G\left(
t,dw\right)  \int_{\mathbb{R}^{3}}G\left(  t,dw_{\ast}\right)  \int_{S^{2}%
}d\omega\Big[ \, \frac{1}{2}B\left(  n\cdot\omega\right)  \left[
\varphi\left(  w^{\prime}\right)  +\varphi\left(  w_{\ast}^{\prime}\right)
-\varphi\left(  w\right)  -\varphi\left(  w_{\ast}\right)  \right] \\
&  -K\left[  \xi_{2}\partial_{\xi_{1}}\varphi\right]  \left(  w\right)
-\beta\xi\cdot\partial_{\xi}\varphi\left(  w\right)  \Big].
\end{align*}

We decompose the first integral in the regions $\left\vert w\right\vert
>\left\vert w_{\ast}\right\vert $ and $\left\vert w\right\vert <\left\vert
w_{\ast}\right\vert $ respectively. Then the first integral, containing the
term $\frac{1}{2}B\left(  n\cdot\omega\right)  $ can be rewritten as an
integral in the region $\left\vert w\right\vert >\left\vert w_{\ast
}\right\vert $ exchanging the variables $w\longleftrightarrow w_{\ast}.$ We
then obtain:%
\begin{align*}
\partial_{t}M_{s}\left(  t\right)  =  &  \int_{\mathbb{R}^{3}\setminus\left\{
\left\vert w\right\vert >\left\vert w_{\ast}\right\vert \right\}  }%
\int_{\mathbb{R}^{3}}G\left(  t,dw\right)  G\left(  t,dw_{\ast}\right)
\int_{S^{2}}d\omega\Big[ B\left(  n\cdot\omega\right)  \left[  \varphi\left(
w^{\prime}\right)  +\varphi\left(  w_{\ast}^{\prime}\right)  -\varphi\left(
w\right)  -\varphi\left(  w_{\ast}\right)  \right] \\
&  -K\left[  \xi_{2}\partial_{\xi_{1}}\varphi\right]  \left(  w\right)
-\beta\xi\cdot\partial_{\xi}\varphi\left(  w\right)  \Big]\\
&  -\int_{\mathbb{R}^{3}\ \left\{  \left\vert w\right\vert <\left\vert
w_{\ast}\right\vert \right\}  }\int_{\mathbb{R}^{3}}G\left(  t,dw\right)
G\left(  t,dw_{\ast}\right)  \left[  K\left[  \xi_{2}\partial_{\xi_{1}}%
\varphi\right]  \left(  w\right)  +\beta\xi\cdot\partial_{\xi}\varphi\left(
w\right)  \right]  .
\end{align*}

The first integral on the right-hand side can be estimated using Proposition
\ref{PovnerShear}. On the other hand, in the second integral we use that
$\left\vert K\left[  \xi_{2}\partial_{\xi_{1}}\varphi\right]  \left(
w\right)  +\beta\xi\cdot\partial_{\xi}\varphi\left(  w\right)  \right\vert
\leq C\left\vert w\right\vert ^{s}\leq C\left\vert w\right\vert ^{\frac{s}{2}%
}\left\vert w_{\ast}\right\vert ^{\frac{s}{2}}$ since we are integrating in
the region $\left\{  \left\vert w\right\vert <\left\vert w_{\ast}\right\vert
\right\}  .$ We then obtain the estimate:\
\begin{align*}
\partial_{t}M_{s}\left(  t\right)   &  \leq-\kappa\int_{\mathbb{R}%
^{3}\ \left\{  \left\vert w\right\vert >\left\vert w_{\ast}\right\vert
\right\}  }\int_{\mathbb{R}^{3}}G\left(  t,dw\right)  G\left(  t,dw_{\ast
}\right)  \varphi\left(  w\right)  +C\int_{\mathbb{R}^{3}}\int_{\mathbb{R}%
^{3}}G\left(  t,dw\right)  G\left(  t,dw_{\ast}\right)  \left\vert
w\right\vert ^{\frac{s}{2}}\left\vert w_{\ast}\right\vert ^{\frac{s}{2}}\\
&  =-\kappa\int_{\mathbb{R}^{3}}\int_{\mathbb{R}^{3}}G\left(  t,dw\right)
G\left(  t,dw_{\ast}\right)  \varphi\left(  w\right)  +\int_{\mathbb{R}%
^{3}\ \left\{  \left\vert w\right\vert <\left\vert w_{\ast}\right\vert
\right\}  }\int_{\mathbb{R}^{3}}G\left(  t,dw\right)  G\left(  t,dw_{\ast
}\right)  \varphi\left(  w\right)  +\\
&  +C\int_{\mathbb{R}^{3}}\int_{\mathbb{R}^{3}}G\left(  t,dw\right)  G\left(
t,dw_{\ast}\right)  \left\vert w\right\vert ^{\frac{s}{2}}\left\vert w_{\ast
}\right\vert ^{\frac{s}{2}}\\
&  \leq-\kappa M_{s}\left(  t\right)  +C\int_{\mathbb{R}^{3}}\int
_{\mathbb{R}^{3}}G\left(  t,dw\right)  G\left(  t,dw_{\ast}\right)  \left\vert
w\right\vert ^{\frac{s}{2}}\left\vert w_{\ast}\right\vert ^{\frac{s}{2}}.
\end{align*}

Since $\frac{s}{2}<2$ we can estimate the last integral in terms of the
particle density and the energy. Then
\[
\partial_{t}M_{s}\left(  t\right)  \leq-\kappa M_{s}\left(  t\right)
+C_{0}\left(  \zeta\right)
\]
with $\kappa>0.$ It then follows that the set $\left\{  M_{s}\left(  t\right)
\leq C_{\ast}\right\}  $ with $C_{\ast}$ sufficiently large, is invariant. The
rest of the proof can then be made along the lines of the proof of Theorem
\ref{th:ssprof}. Actually the argument above must be made using an
approximating sequence $G_{m}$ as in the proof of Theorem \ref{th:ssprof}. On
the other hand, we have implicitly assumed that the measure of the set
$\left\{  \left\vert w\right\vert =\left\vert w_{\ast}\right\vert \right\}  $
is zero. If this is not the case we must split the mass in this line in equal
portions in the regions $\left\{  \left\vert w\right\vert >\left\vert w_{\ast
}\right\vert \right\}  $ and $\left\{  \left\vert w\right\vert <\left\vert
w_{\ast}\right\vert \right\}  .$
\end{proofof}

\bigskip


\subsubsection{Heat fluxes for homoenergetic flows for simple shear solutions}

\label{HeatFluxes}

We discuss in this section a phenomenon discussed in \cite{TM} concerning the
onset of nontrivial heat fluxes for homoenergetic solutions. Suppose that a self-similar solution of \eqref{W1} exists for $K$ sufficiently large. The heat fluxes
in gases described by means of the Boltzmann equation are given by:%
\begin{equation}
q=\int_{\mathbb{R}^{3}}\left\vert w\right\vert ^{2}wg\left(  dw\right).
\label{V2E4}%
\end{equation}

In the solutions obtained in Theorems \ref{ThSimpShear} and
\ref{ThSimpShearArbK} we can assume that they satisfy the symmetry condition:%
\begin{equation}
G\left(  w\right)  =G\left(  -w\right)  \label{V2E3}%
\end{equation}

This is due to the fact that the space of measures satisfying (\ref{V2E3}) is
invariant under the evolution semigroup $\mathcal{S}\left(  t\right)  .$
Therefore, for such self-similar solutions the heat flux $q$ given by
(\ref{V2E4}) is zero. This is seemingly in contrast with a computation made in
\cite{TM} where the evolution of the third moments tensor $M_{j,k,\ell}%
=\int_{\mathbb{R}^{3}}w_{j}w_{k}w_{\ell}g\left(  dw\right)  $ has been
computed and it has been seen there that for generic solutions the third
moments tensor increases exponentially. In particular the heat flux $q$ in
(\ref{V2E4}) can be computed in terms of the third moments tensor and it also
increases exponentially if $\left\vert K\right\vert $ is sufficiently large.

\bigskip

It is not clear if the self-similar solutions constructed in this paper yield
the same distribution of velocities associated to the evolution of the moments
in \cite{TM} because we have not proved neither uniqueness of the self-similar
solutions or stability. However, the fact that the evolution of the second
moments tensor $M_{j,k}$ for the solutions obtained in this paper \ growth
exponentially with the same exponent obtained in \cite{TM} strongly suggests
that the type of solutions considered in this paper and those suggested in
\cite{TM} are related. However, the exponential growth of the third moments
tensor obtained in \cite{TM} raises doubts about the stability of the
solutions obtained in this paper. We will argue now that the values of the
exponents obtained in \cite{TM} support the following scenario for large
values of $K:$ The self-similar solutions in Theorem \ref{ThSimpShearArbK}, if
they exist (i.e. if condition (\ref{V1E4}) holds) are stable under small
perturbations, but the eigenmode associated to the leading eigenvalue of the
problem obtained linearizing around the self-similar solutions does not
satisfy the symmetry condition (\ref{V2E3}).

In order to justify this scenario we will use the notation in \cite{TM}. The
exponential growth of the second moments tensor is $e^{At}$ where $A$ is the
root of the following equation with the largest real part (cf. (XIV.4) in
\cite{TM}):%
\begin{equation}
A\left(  A+1\right)  ^{2}=\frac{2}{3}\mathbf{T}^{2} . \label{V2E5}%
\end{equation}

The parameter $\mathbf{T}$ plays a role equivalent to $\frac{K}{b}$ in Theorem
\ref{ThSimpShearArbK}.

On the other hand, the exponential growth for the heat fluxes $q$ is given by
$e^{Rt}$ with $R$ is the root with the largest part of one of one of the
following equations (cf. (XIV.29), (XIV.31) in \cite{TM}):
\begin{align}
\left(  R+\frac{3}{2}\right)  ^{2}\left(  R+\frac{2}{3}\right)   &  =\frac
{1}{3}\mathbf{T}^{2}\label{V2E6}\\
\left(  R+\frac{3}{2}\right)  ^{2}\left(  R+\frac{2}{3}\right)  ^{2}  &
=2\mathbf{T}^{2}\left(  R+\frac{31}{36}\right) .\nonumber
\end{align}

In order to study the stability of the self-similar solution $G$ it is more
convenient to represent it using the variable $\xi=\frac{w}{e^{\beta t}}$
where $e^{\beta t}$ yields the characteristic velocity $w$ of the particles
for self-similar solutions. Since the second moments tensor increases as
$e^{At}$ and in the simple shear case the total mass is preserved, we would
have $\beta=\frac{A}{2}.$ Therefore, the eigenvalues obtained by means of
(\ref{V2E6}) would be associated to small perturbations of the self-similar
solution $G\left(  \xi\right)  $ if the largest root of the equations
(\ref{V2E6}) satisfies
\begin{equation}
R<3\beta=\frac{3A}{2}. \label{V2E7}%
\end{equation}

In order to prove (\ref{V2E7}) we introduce a new variable $R=3\chi.$ Then
$\chi$ is the solution with the largest real part of one of the equations
\begin{align}
\left(  \chi+\frac{1}{2}\right)  ^{2}\left(  \chi+\frac{2}{9}\right)   &
=\frac{1}{81}\mathbf{T}^{2}\label{V2E8}\\
\left(  \chi+\frac{1}{2}\right)  ^{2}\left(  \chi+\frac{2}{9}\right)  ^{2}  &
=\frac{2\mathbf{T}^{2}}{27}\left(  \chi+\frac{31}{108}\right)  .\nonumber
\end{align}

On the other hand we can rewrite (\ref{V2E5}) using that $A=2\beta$ as %
\begin{equation}
\beta\left(  \beta+\frac{1}{2}\right)  ^{2}=\frac{1}{12}\mathbf{T}^{2}.
\label{V2E9}%
\end{equation}

We need to prove that the root of (\ref{V2E8}) with the largest real part
satisfies $\chi<\beta,$ where $\beta$ is the root of (\ref{V2E9}) with the
largest real part. This result would follow proving the following inequalities
for $x>0$%
\begin{align*}
12x\left(  x+\frac{1}{2}\right)  ^{2}  &  <81\left(  x+\frac{1}{2}\right)
^{2}\left(  x+\frac{2}{9}\right) \\
12x\left(  x+\frac{1}{2}\right)  ^{2}  &  <\frac{27}{2}\frac{\left(
x+\frac{1}{2}\right)  ^{2}\left(  x+\frac{2}{9}\right)  ^{2}}{\left(
x+\frac{31}{108}\right)  }%
\end{align*}
which reduce to
\begin{align*}
12x  &  <81\left(  x+\frac{2}{9}\right)  \ \ \text{for }x\geq0\\
12x  &  <\frac{27}{2}\frac{\left(  x+\frac{2}{9}\right)  ^{2}}{\left(
x+\frac{31}{108}\right)  }\ \ \text{for }x\geq0 .
\end{align*}
The first of these inequalities is obviously satisfied, and the second one is
equivalent to%
\[
24x\left(  x+\frac{31}{108}\right)  <27\left(  x+\frac{2}{9}\right)
^{2}\ \ \text{for }x\geq0
\]
or equivalently %
\[
27\left(  x+\frac{2}{9}\right)  ^{2}-24x\left(  x+\frac{31}{108}\right)
=3x^{2}+\frac{46}{9}x+\frac{4}{3}>0
\]
which is obviously satisfied.

Therefore the desired instability follows. These inequalities suggest the
scenario mentioned above concerning the stability of $G\left(  \xi\right)  .$
More precisely, the asymptotic behavior of small perturbations of $G$ would
yield solutions with the form
\[
\bar{G}\left(  t,\xi\right)  =G\left(  \xi\right)  +e^{\lambda t}%
\varphi\left(  \xi\right)
\]
where $\lambda=R-\frac{3A}{2}<0$ and $\varphi\left(  \xi\right)  \neq
\varphi\left(  -\xi\right)  .$ Nevertheless in order to prove this scenario a
more careful analysis of the linearized problem would be needed.

\subsection{Planar shear for Maxwell molecules}  \label{planarshear}

\bigskip

In this subsection we consider the self-similar solutions for homoenergetic
flows (\ref{B1_0}), (\ref{B4_0}) with $L\left(  t\right)  $ as in (\ref{T1E3})
with $K\neq0.$ Then $g$ solves (\ref{D1_0}). We first check using dimensional
analysis that the terms $-L\left(  t\right)  w\cdot\partial_{w}g$ and
$\mathbb{C}g$ can be expected to have the same order of magnitude as
$t\rightarrow\infty$ if the homogeneity of the collision kernel $B$ is
$\gamma=0,$ i.e. for Maxwell molecules.

Using (\ref{T1E3}) in (\ref{D1_0}) we obtain that $g$ solves
\begin{equation}
\partial_{t}g-\frac{K}{t}w_{3}\partial_{w_{2}}g-\frac{1}{t}w_{3}%
\partial_{w_{3}}g=\mathbb{C}g\left(  w\right)  . \label{S9E8}%
\end{equation}

We have ignored the term $O\left(  \frac{1}{t^{2}}\right)  $ in (\ref{T1E3})
because this term is integrable, and it just produce a factor of order one in
the evolution of the characteristic curves in the space $w$ as $t\rightarrow
\infty.$

In order to find a reformulation with a conserved mass we need to compute the
evolution of the density $\rho\left(  t\right)  .$ We have $\operatorname*{tr}%
\left(  L\left(  t\right)  \right)  =1$ and then (\ref{S8E5}) implies 
\begin{equation}
\rho\left(  t\right)  =\frac{\rho\left(  1\right)  }{t} . \label{S9E9}%
\end{equation}

Suppose that the homogeneity of the kernel $B$ is $\gamma.$ Then, using
(\ref{S9E9}) we can see that the scaling properties of the four terms in
(\ref{S9E8}) are given by
\[
\frac{\left[  g\right]  }{\left[  t\right]  },\ \frac{\left[  g\right]
}{\left[  t\right]  },\ \frac{\left[  g\right]  }{\left[  t\right]  }%
,\ \frac{\left[  w\right]  ^{\gamma}\left[  g\right]  }{\left[  t\right]  }.
\]

Therefore, all the terms have the same order of magnitude even if the
temperature increases if $\gamma=0,$ i.e. for Maxwell molecules. We will
restrict to this case in this subsection. In order to transform (\ref{S9E8})
to a form with conserved density we use the change of variables
\[
g\left(  t,w\right)  =\frac{1}{t}\bar{g}\left(  \tau,w\right)  \ \ ,\ \ \tau
=\log\left(  t\right)
\]
whence
\begin{equation}
\partial_{\tau}\bar{g}-Kw_{3}\partial_{w_{2}}\bar{g}-\partial_{w_{3}}\left(
w_{3}\bar{g}\right)  =\mathbb{C}\bar{g}\left(  w\right)  . \label{T4E9}%
\end{equation}

We remark that
\begin{equation}
\partial_{\tau}\left(  \int_{\mathbb{R}^{3}}\bar{g}\left(  \tau,dw\right)
\right)  =0 . \label{T5E1}%
\end{equation}

We now look for similar solutions of (\ref{T4E9}). The conservation property
(\ref{T5E1}) suggests to look for self-similar solutions with the form
\begin{equation}
\bar{g}\left(  \tau,w\right)  =e^{-3\beta\tau}G\left(  \xi\right)
\ \ ,\ \ \xi=\frac{w}{e^{\beta\tau}}. \label{T5E2}%
\end{equation}
Therefore
\begin{equation}
-\beta\partial_{\xi}\left(  \xi\cdot G\right)  -K\partial_{\xi_{2}}\left(
\xi_{3}G\right)  -\partial_{\xi_{3}}\left(  \xi_{3}G\right)  =\mathbb{C}%
G\left(  w\right) . \label{T5E3}%
\end{equation}

This equation is a particular case of (\ref{S4E5}) with $\alpha=\beta$%
\begin{equation}
L=\left(
\begin{array}
[c]{ccc}%
0 & 0 & 0\\
0 & 0 & K\\
0 & 0 & 1
\end{array}
\right)  . \label{T5E3a}%
\end{equation}

Theorem \ref{th:ssprof} will then imply the existence of nontrivial solutions
of (\ref{S5E3}). It is worth to write in detail the eigenvalue problem
yielding $\beta.$ We recall that $\beta$ is the solution $\alpha$ of the
eigenvalue problem (\ref{S6E5}), (\ref{S6E6}) with the largest real part. We
use (\ref{T5E3a}) to write
\[
L_{j,k}=K\delta_{j,2}\delta_{k,3}+\delta_{j,3}\delta_{k,3}.
\]
Then, the eigenvalue problem (\ref{S6E5}), (\ref{S6E6}) becomes
\begin{align*}
\frac{\alpha}{b}\Gamma_{j,k}+\frac{1}{2b}\left(  \left[  K\delta_{j,2}%
+\delta_{j,3}\right]  \Gamma_{k,3}+\left[  K\delta_{k,2}+\delta_{k,3}\right]
\Gamma_{j,3}\right)   &  =-\left(  \Gamma_{j,k}-\Gamma\delta_{j,k}\right)
\ ,\ \ j,\ k=1,2,3,\ \ \\
\Gamma_{j,k} &  =\Gamma_{k,j}\\
\Gamma &  =\frac{1}{3}\left(  \Gamma_{1,1}+\Gamma_{2,2}+\Gamma_{3,3}\right)
\end{align*}
or in more detailed form
\begin{align}
\left(  \frac{\alpha}{b}+1\right)  \Gamma_{1,1}  &  =\Gamma\ \ ,\ \ \left(
\frac{\alpha}{b}+1\right)  \Gamma_{1,2}+\frac{K}{2b}\Gamma_{1,3}%
=0\ \ ,\ \ \left(  \frac{\alpha}{b}+1\right)  \Gamma_{1,3}+\frac{1}{2b}%
\Gamma_{1,3}=0\label{T5E4}\\
\left(  \frac{\alpha}{b}+1\right)  \Gamma_{2,2}+\frac{K}{b}\Gamma_{2,3}  &
=\Gamma\ \ ,\ \ \left(  \frac{\alpha}{b}+1\right)  \Gamma_{2,3}+\frac{K}%
{2b}\Gamma_{3,3}+\frac{1}{2b}\Gamma_{2,3}=0\ ,\ \ \left(  \frac{\alpha}%
{b}+1\right)  \Gamma_{3,3}+\frac{1}{b}\Gamma_{3,3}=\Gamma. \label{T5E5}%
\end{align}

We then have the following result.

\begin{theorem}
\ \label{Th1DilplusShear} Suppose that $B$ in (\ref{eq:collboltzgen}) is
homogeneous of order zero and suppose that $b$ is as in (\ref{S6E1}). There
exists $b_{0}>0$ large and $k_{0}>0$ small such that, for any $\zeta\geq0$ and
any $b\geq b_{0}$ and any $K\in{\mathbb{R}}$ such that $\frac{K}{b}\leq k_{0}$
there exists $\beta\in\mathbb{R}$ and $G\in\mathcal{M}_{+}\left(
\mathbb{R}_{c}^{3}\right)  $ which solves (\ref{S9E8}) and satisfies the
normalization conditions
\begin{equation}
\int_{\mathbb{R}^{3}}G\left(  dw\right)  =1,\ \ \ \int_{\mathbb{R}^{3}}%
w_{j}G\left(  dw\right)  =0,\ \int_{\mathbb{R}^{3}}\left\vert w\right\vert
^{2}G\left(  dw\right)  =\zeta\ . \label{T5E6}%
\end{equation}

Moreover, the following asymptotics holds for $\beta:$%
\begin{equation}
\beta\sim\frac{K^{2}-2b}{6b}+O\left(  \frac{1}{b}\right)  \ \ \text{if\ \ }%
K=O\left(  \sqrt{b}\right)  \ \text{as\ }b\rightarrow\infty. \label{T5E7}%
\end{equation}

\end{theorem}

\begin{remark}
Notice that the exponent $\beta$ might have positive or negative values. This
depends on the value of $K.$ In homoenergetic flows described by (\ref{T1E3})
(equivalently (\ref{T5E3a})) there are two competing effects. The dilatation
term tends to decrease the average energy of the molecules (which we will
think as a temperature in spite of the fact that the velocity distribution is
not close to a Maxwellian). On the contrary, the shear term tends to increase
the temperature of the system. The exponent $\beta$ is positive if the effect
of the shear is more important than the one due to dilatation, and as a
consequence the temperature of the molecules increases. On the contrary, if
the effect of the shear is small compared with the one of dilatation, $\beta$
is negative and the temperature of the system decreases, as it might be expected.
\end{remark}

\begin{remark}
In the original set of variables the self-similar solution has the form
\[
g\left(  t,w\right)  =\frac{1}{t^{4}}G\left(  \frac{w}{t}\right)
\]
with $G$ as in Theorem \ref{Th1DilplusShear}.
\end{remark}

\begin{proof}
The existence of a real number $\beta$ and a measure $G$ satisfying
(\ref{T5E6}) and solving (\ref{T5E3}), if $b$ is sufficiently large and
$\frac{K}{b}$ is sufficiently small, is a straightforward consequence of Theorem
\ref{th:ssprof} since, under these assumptions, $\left\Vert L\right\Vert $ in
(\ref{T5E3a}) is small.

It only remains to prove the asymptotics (\ref{T5E7}). To this end we describe
in detail the solutions of the eigenvalue problem (\ref{T5E4}), (\ref{T5E5}).
We denote $\lambda=\frac{\alpha}{b}+1.$ The problem (\ref{T5E4}), (\ref{T5E5})
has five eigenvalues, one of them with multiplicity two, namely
\[
\lambda=-\frac{1}{2b},\ \text{with\ eigenfunction }\Gamma_{1,2}=K\Gamma
_{1,3},\ \Gamma_{1,1}=\Gamma_{2,2}=\Gamma_{3,3}=\Gamma_{2,3}=0
\]%
\[
\lambda=0\ \ \text{with 2d subspace of eigenfunctions }\Gamma_{3,3}%
=\Gamma_{1,3}=\Gamma_{2,3}=\Gamma_{1,1}+\Gamma_{2,2}=0.
\]
The three roots of the equation
\begin{equation}
-\frac{2}{3}A^{2}+2B^{2}\lambda-\frac{4}{3}B^{2}+3B\lambda^{2}-\frac{7}%
{3}B\lambda+\lambda^{3}-\lambda^{2}=0 \label{T5E8}%
\end{equation}
with $A=\frac{K}{2b}$ and $B=\frac{1}{2b}.$ If $A$ and $B$ are small one of
the roots of (\ref{T5E8}) would be close to $\lambda=1$ and the other two
would be close to zero. Since we are interested in the root with the largest
real part we compute the asymptotics of the root close to one. Using the
Implicit Function Theorem we obtain
\[
\lambda-1=\frac{2}{3}A^{2}-\frac{2}{3}B+O\left(  A^{4}+B^{2}\right)
\]
as $\left(  A,B\right)  \rightarrow0,$ whence (\ref{T5E7}) follows.
\end{proof}

\subsubsection{Planar shear with $K=0$. } 

We now consider self-similar solutions for homoenergetic flows (\ref{B1_0}),
(\ref{B4_0}) with $L\left(  t\right)  $ as in (\ref{T1E3}) with $K=0.$
Actually this case can be considered a limit case of the one considered in the
previous subsection (namely $K\rightarrow0$), but we discuss it separately
because the competition between dilatation and shear effects does not take
place. In this case $g$ solves (\ref{D1_0}) which in this case becomes,
ignoring 
the term $O\left(  \frac{1}{t^{2}}\right)  $ as in the previous
subsection,
\begin{equation}
\partial_{t}g-\frac{1}{t}w_{3}\partial_{w_{3}}g=\mathbb{C}g\left(  w\right) .
\label{T5E9}%
\end{equation}
Using (\ref{S8E5}) we obtain
\begin{equation}
\rho\left(  t\right)  =\frac{\rho\left(  1\right)  }{t}. \label{T5E9a}%
\end{equation}

A dimensional analysis argument similar to the one in the previous subsection
shows that the balance between the hyperbolic term $-\frac{1}{t}w_{3}%
\partial_{w_{3}}g$ and the collision term $\mathbb{C}g\left(  w\right)  $
takes place for kernels $B$ with homogeneity $\gamma=0.$ We will restrict our
analysis to that case.

We change variables in order to obtain a problem with conserved "mass". We
define
\[
g\left(  t,w\right)  =\frac{1}{t}\bar{g}\left(  \tau,w\right)  \ \ ,\ \ \tau
=\log\left(  t\right) .
\]
Then
\begin{equation}
\partial_{\tau}\bar{g}-\partial_{w_{3}}\left(  w_{3}\bar{g}\right)
=\mathbb{C}\bar{g}\left(  w\right)  . \label{T6E1}%
\end{equation}
Thus $\partial_{\tau}\left(  \int_{\mathbb{R}^{3}}\bar{g}\left(
\tau,dw\right)  \right)  =0.$ Taking this into account we look for
self-similar solutions of (\ref{T6E1}) with the form
\begin{equation}
\bar{g}\left(  \tau,w\right)  =e^{-3\beta\tau}G\left(  \xi\right)
\ ,\ \ \xi=\frac{w}{e^{\beta\tau}} \label{T6E1a}%
\end{equation}
where $G$ solves
\begin{equation}
-\beta\partial_{\xi}\left(  \xi\cdot G\right)  -\partial_{\xi_{3}}\left(
\xi_{3}G\right)  =\mathbb{C}G\left(  w\right)  . \label{T6E2}%
\end{equation}

This equation is a particular case of (\ref{S4E5}) with $\alpha=\beta$ and
\begin{equation}
L=\left(
\begin{array}
[c]{ccc}%
0 & 0 & 0\\
0 & 0 & 0\\
0 & 0 & 1
\end{array}
\right)  . \label{T6E2a}%
\end{equation}

We remark also that (\ref{T6E2}), (\ref{T6E2a}) are analogous to (\ref{T5E3}),
(\ref{T5E3a}) with $K=0.$ We will prove the existence of nontrivial solutions
of (\ref{T6E2}) for some suitable $\beta,$ using Theorem \ref{th:ssprof}. We
recall that $\beta$ is obtained by means of the solution of an eigenvalue
problem (cf. (\ref{S6E5}), (\ref{S6E6})). Actually, this eigenvalue problem in
the case of $L$ given by (\ref{T6E2a}) takes the form (\ref{T5E4}),
(\ref{T5E5}) with $K=0$:%

\begin{align}
\left(  \frac{\alpha}{b}+1\right)  \Gamma_{1,1}  &  =\Gamma\ \ ,\ \ \left(
\frac{\alpha}{b}+1\right)  \Gamma_{1,2}=0\ \ ,\ \ \left(  \frac{\alpha}%
{b}+1\right)  \Gamma_{1,3}+\frac{1}{2b}\Gamma_{1,3}=0\label{T6E3a}\\
\left(  \frac{\alpha}{b}+1\right)  \Gamma_{2,2}  &  =\Gamma\ \ ,\ \ \left(
\frac{\alpha}{b}+1\right)  \Gamma_{2,3}+\frac{1}{2b}\Gamma_{2,3}%
=0\ ,\ \ \left(  \frac{\alpha}{b}+1\right)  \Gamma_{3,3}+\frac{1}{b}%
\Gamma_{3,3}=\Gamma. \label{T6E3b}%
\end{align}

We then have the following result.

\begin{theorem}
\ \label{Th1Dil} Suppose that $B$ in (\ref{eq:collboltzgen}) is homogeneous of
order zero and suppose that $b$ is as in (\ref{S6E1}). There exists $b_{0}>0$
large such that, for any $\zeta\geq0$ and any $b\geq b_{0}$ there exists
$\beta\in\mathbb{R}$ and $G\in\mathcal{M}_{+}\left(  \mathbb{R}_{c}%
^{3}\right)  $ which solves (\ref{S9E8}) and satisfies the normalization
conditions
\begin{equation}
\int_{\mathbb{R}^{3}}G\left(  dw\right)  =1,\ \ \ \int_{\mathbb{R}^{3}}%
w_{j}G\left(  dw\right)  =0,\ \int_{\mathbb{R}^{3}}\left\vert w\right\vert
^{2}G\left(  dw\right)  =\zeta\ . \label{T6E4}%
\end{equation}
Moreover,
\begin{equation}
\beta=b\left[  \frac{1}{2}\left[  -\left(  \frac{1}{b}+1\right)
+\sqrt{\left(  \frac{1}{b}-1\right)  ^{2}+\frac{8}{3}\frac{1}{b}}\right]
\right] . \label{T6E4a}%
\end{equation}%

\end{theorem}

\begin{remark}
It might be readily seen that $\beta<0$ for any $b>0.$ Moreover $\beta
=\beta\left(  b\right)  $ is a decreasing function of $b.$ (See Figure
\ref{fig:graph}). Therefore, the temperature of this system decreases as $t$
increases. Notice that this solution reduces to the one obtained in Theorem
\ref{Th1DilplusShear} if we take $K=0.$ We obtain $\beta\rightarrow-\frac
{1}{3} $ as $b\rightarrow\infty.$ Thus, the temperature decreases faster
as $b$ increases.

\begin{figure}[th]
\centering \caption{Graphic of the function $\beta
=\beta\left(  b\right)  $. \label{fig:graph} }
\includegraphics [scale=0.5]{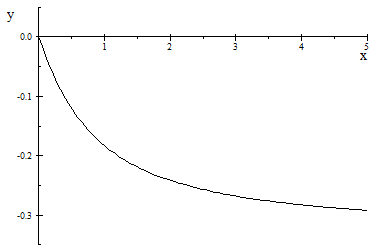}
\end{figure}

\end{remark}

\begin{remark}
In the original set of variables the self-similar solution has the form
\[
g\left(  t,w\right)  =\frac{1}{t^{4}}G\left(  \frac{w}{t}\right)
\]
with $G$ as in Theorem \ref{Th1DilplusShear}.
\end{remark}

\begin{proof}
The existence of a real number $\beta$ and a measure $G$ satisfying
(\ref{T5E6}) and solving (\ref{T5E3}) if $b$ is sufficiently large and
$\frac{K}{b}$ is sufficiently small is just a consequence of Theorem
\ref{th:ssprof} since under these assumptions $\left\Vert L\right\Vert $ in
(\ref{T5E3a}) is small.

It only remains to prove the asymptotics (\ref{T5E7}). To this end we describe
in detail the solutions of the eigenvalue problem (\ref{T5E4}), (\ref{T5E5}).
We denote $\lambda=\frac{\alpha}{b}+1.$ The problem (\ref{T6E3a}),
(\ref{T6E3b}) has the following eigenvalues and eigenvectors
\[
\lambda=0\ :\ \Gamma_{1,3}=\Gamma_{2,3}=\Gamma_{3,3}=0\ ;\ \Gamma_{1,1}%
+\Gamma_{2,2}=0\ ,\ \ \Gamma_{1,2}\text{ arbitrary}%
\]%
\[
\lambda=-\frac{1}{2b}\ :\ \Gamma_{1,2}=\Gamma_{3,3}=\Gamma_{1,1}=\Gamma
_{2,2}=0\ ;\ \Gamma_{1,3},\ \Gamma_{2,3}\text{ arbitrary} .
\]

Notice that the subspaces of eigenvectors of each of these eigenvalues have
dimension two. The last remaining eigenvectors are
\begin{align*}
\lambda_{1}  &  =\frac{1}{2}\left[  -\left(  B-1\right)  +\sqrt{\left(
B-1\right)  ^{2}+\frac{8}{3}B}\right] \\
\lambda_{2}  &  =\frac{1}{2}\left[  -\left(  B-1\right)  -\sqrt{\left(
B-1\right)  ^{2}+\frac{8}{3}B}\right]  .
\end{align*}

Since $\beta$ is given by the eigenvalue of (\ref{T6E3a}), (\ref{T6E3b}) with
the largest real part, i.e. $\lambda_{1}$ we obtain (\ref{T6E4a}).
\end{proof}

\section{Conjectures on the non-self-similar behavior}
\label{conj}

We recall that, the collision operator in \eqref{D1_0} is quadratic. It rescales like:
\begin{equation}
\rho\left(  t\right)  \left[  w\right]  ^{\gamma}\left[  g\right]
\ \label{eq:colltermscaling}%
\end{equation}
where $\left[  w\right]  $ is the order of magnitude of $w$, $\gamma$ is the homogeneity of the collision kernel $B$ (cf. \eqref{S8E7}) and $\left[
g\right]  $ the order of magnitude of $g$

The term $L\left(  t\right)  $ can yield different behaviors as
$t\rightarrow\infty.$ We denoted  the term $L\left(  t\right)  w\cdot\partial_{w}g$ as hyperbolic term. It can be constant, or it can
behave like a power law (increasing or decreasing). As we pointed out in the introduction, the key idea is that there
are three possibilities depending on the value of  the homogeneity $\gamma$ and the function
yielding the scaling of $\left[  w\right].$ Either the hyperbolic term is
larger than the collision term as $t\rightarrow\infty$, either the collision
term is larger or either the hyperbolic term and the collision term have the
same order of magnitude. 
Suppose that $L\left(  t\right)  $ scales like a function $\eta\left(
t\right)  .$ The hyperbolic term scales then like $\eta\left(  t\right)
\left[  g\right]  $ and the collision term scales as in (\ref{eq:colltermscaling}). Therefore,
we need to compare the terms: 
$
\eta\left(  t\right)  \text{ and }\rho\left(  t\right)  \left[  w\right]
^{\gamma}.
$

In this section we give conjectures for the cases in which the hyperbolic term and the collision term do not balance.  
More precisely, the cases for which
the hyperbolic terms dominate, and those in which the collision term dominates.  We believe, based on formal calculations, 
that the latter can be handled by the Hilbert expansion, but using as  small parameter $1/t$.  These lengthy formal calculations are presented elsewhere \cite{JMNV},
and here we give conjectures based on these calculations to complete most of the cases classified in 
Section \ref{ss:classeqsol}.  See Table \ref{sec:tableresults} for these conjectures.

The cases in which the hyperbolic terms dominate have two subcases.  In one subcase, 
according to a simplified model (presented in \cite{JMNV}), the collisions term is
formally very small as $t\rightarrow\infty,$ but has a huge effect on the
particle distributions.  We do not make conjectures about this interesting subcase here. 
In the other subcase the hyperbolic terms are so dominant that the collisions have no
effect on the asymptotic behavior of the solution (``frozen collisions'').  

We describe a few details on these formal calculations below.

\bigskip


\subsection{Collision-dominated behavior}

As we discussed at the beginning of this section, we can have three different asymptotic behaviors  for the solutions of \eqref{D1_0} depending on the value of the homogeneity $\gamma$ and the function
yielding the scaling of $\left[  w\right].$ 
Here we focus on the case in which, for some values of $\gamma$, the
collision term dominates the hyperbolic term. From now on, we refer to this case as the ``collision-dominated behavior'' case. 

For collision-dominated behavior we have computed the asymptotics of the velocity dispersion 
using a suitable Hilbert expansion around the Maxwellian equilibrium.   To formulate our conjecture
based on this expansion we define for $t>0$
\begin{equation}
\mu(t) = \left\{ \begin{array}{ll}
1, &  {\rm simple\ shear}, \\
1/t, &  {\rm planar\ shear,\ or\ 2d\ dilatation,\ or\ combined\ shear}, \\
\end{array} \right. \label{mu(t)}
\end{equation}
In the long time asymptotics, the solutions behave like a Maxwellian distribution with increasing or decreasing temperature depending on the sign of the homogeneity parameter $\gamma$.

\vspace{2mm}
\noindent {\bf Conjecture}.
{\it \label{th:genHilbexp} Let $g\left(  \cdot\right)  \in C\left(  \left[
0,\infty\right]  :\mathcal{M}_{+}\left(  \mathbb{R}_{c}^{3}\right)  \right)  $
be a mild solution in the sense of Definition \ref{mildSol} of the Boltzmann
equation  with cross-section $B$ and let $\mu$ be defined in the various cases by
(\ref{mu(t)}). Then, for $t\to\infty$, the solution behaves like
a Maxwellian distribution, i.e.
\begin{equation}
\label{eq:MaxwHilb}
g(w,t) \to \tilde{C} \beta(t)^{\frac 3 2} e^{-\beta(t) \left\vert w \right\vert ^{2}}\quad\text{in}\quad
L^{2}\left(  \mathbb{R}^{3};e^{-\left\vert w
\right\vert ^{2}}dw\right)  .
\end{equation}
where $\tilde{C}=
\frac{1}{(2\pi)^{\frac 3 2}}$. 
More precisely, we have the following cases.
\begin{itemize}
\item[1)] Assume that
\begin{equation}
\label{eq:case1Hilb}\Tr(L)\neq 0
\end{equation}
with $L$ as in \eqref{B7_0}.
We define
$a:=\frac{2}{3} \Tr(L)$.  If $\mu(t)e^{-\frac{\gamma}{2}at}\to\infty$ the asymptotic behavior is given by a Maxwellian distribution
\eqref{eq:MaxwHilb} with 
\begin{equation}
\label{eq:beta1Hilb}\beta(t)=C\,e^{at} \quad {\rm as} \  \  t \to \infty,
\end{equation}
where $C>0$ is a numerical constant.
\item[2)] Let $\gamma>0$ and assume that
\begin{equation}\label{eq:assmu}
\int_{0}^{+\infty}\frac{ds}{\mu(s)}=\infty
\end{equation}
and
\begin{equation}
\label{eq:case2Hilb}\Tr(L)=0
\end{equation}
with $L$ as in \eqref{B7_0}.
Then the asymptotic behavior is given by a Maxwellian distribution
\eqref{eq:MaxwHilb} with increasing temperature where $\beta(t)$ satisfies

\begin{equation}
\label{eq:beta2Hilb}\beta(t) \simeq \left( \gamma\, b \int_{0}^{t} \frac{ds}{\mu(s)}\right)^{-\frac{2}{\gamma}} \quad \text{as}\quad t\to\infty
\end{equation}
with $b=-\left\langle \xi \cdot L \xi, (\mathbb{L})^{-1}[ (\xi\cdot L\xi)] ) \right\rangle_w\,$ (Green-Kubo formula).
\end{itemize}
}

Further details about this conjecture can be found in \cite{JMNV}. We emphasize that in the case 1) with $\mu(t)=1$,  in order to obtain a dynamics dominated by collisions, we must choose the homogeneity $\gamma$ satisfying the condition $a\cdot \gamma<0$.

\subsection{Hyperbolic-dominated behavior}

As mentioned at the beginning of this section we focus on the case of frozen collisions, i.e, that the collisions term
becomes so small that the effect of collisions is irrelevant as $t \to \infty$.  The formal argument underlying our 
conjectures is based on control of collision rate (gain term) for molecular densities that satisfy the asymptotic 
first order hyperbolic equation $\partial_t g + \partial_w \cdot (L(t) w g) = 0$.  If the resulting collision rate is decreasing
in time, we refer to this case as hyperbolic-dominated behavior.  More precisely, our terminology {\it frozen collisions}
refers to the case of exponentially decreasing behavior of the collision rate as $t \to \infty$.  

For frozen collisions we conjecture the following behavior. 
\begin{enumerate}
\item {\bf Homogeneous dilatation:}   
for $\gamma > -2$ the solution $g(t,w)$ converges in the sense of measures
to a limit distribution that depends on the initial datum.  
\item {\bf Cylindrical dilatation:} 
for $\gamma > -2$ the solution $g(t,w)$ converges in the sense of measures
to a limit distribution that depends on the initial datum.  
\item {\bf Simple shear:} 
for $\gamma < -1$ the solution $g(t,w)$ converges in the sense of measures
to a limit distribution that depends on the initial datum.  
\end{enumerate}
Note that these regimes are complementary to those given by the formal Hilbert expansion and the self-similar
profile, except the case of simple shear, for which there is a gap $-1 \le \gamma < 0$.  In this gap the collision
rate is small but it still plays a significant role in the formal asymptotic behavior of the Boltzmann equation. 
A detailed justification for these conjectures can be found in \cite{JMNV}.


\section{Entropy formulas}

\label{sec:entropy}

\bigskip

Homoenergetic solutions are characterized by constant values in space of the
particle density $\rho=\rho\left(  t\right)  $ and internal energy
$\varepsilon=\varepsilon\left(  t\right)  .$ We now discuss the form of
another relevant thermodynamic magnitude, namely the entropy. We identify for
the Boltzmann equation the entropy with minus the $H-$function. Then if the
velocity distribution is given by $f=f\left(  t,x,v\right)  ,$ we obtain the
following entropy density for particle at a given point $x$
\[
\frac{s\left(  x,t\right)  }{\rho\left(  t\right)  }=-\frac{1}{\rho\left(
t\right)  }\int f\left(  t,x,v\right)  \log\left(  f\left(  t,x,v\right) .
\right)  d^{3}v
\]

It then readily follows, using (\ref{B1_0}) the entropy density for particle
is independent of $x$ and given by
\begin{equation}
\frac{s\left(  t\right)  }{\rho\left(  t\right)  }=-\frac{1}{\rho\left(
t\right)  }\int_{\mathbb{R}^{3}}g\left(  t,w\right)  \log\left(  g\left(
t,w\right)  \right)  d^{3}w . \label{U1E8}%
\end{equation}

It is interesting to notice that in several of the solutions discussed above,
the formulas for entropy for particle have many analogies with the
corresponding formulas for equilibrium distributions, in spite of the fact
that the distributions obtained in this paper deal with nonequilibrium situations.

\bigskip

The case in which the analogy between the entropy formulas for the equilibrium
case and the considered solutions is the largest, nonsurprisingly, if the
particle distribution is given by a Hilbert expansion (see details in \cite{JMNV}). However, there is also
a large analogy between the entropy formulas of equilibrium distributions and
self-similar solutions. This is due to the fact that to a large extent, the
entropy formulas depend on the scaling properties of the distributions.
Indeed, notice that both in the cases of solutions given by Hilbert
distributions or self-similar solutions we can approximate $g\left(
t,w\right)  $ as
\begin{equation}
g\left(  w,t\right)  \sim\frac{1}{a\left(  t\right)  }G\left(  \frac
{w}{\lambda\left(  t\right)  }\right)  \text{ as }t\rightarrow\infty
\label{U1E7}%
\end{equation}
for suitable functions $a\left(  t\right)  ,\ \lambda\left(  t\right)  $ which
are related to the particle density and the average energy of the particles.
In the case of solutions given by Hilbert expansions the distribution $G$ is a
Maxwellian, which can be assumed to be normalized to have density one and
temperature one. Moreover, we will assume also that the mass of the particles is normalized to $m=2$ in order to get simpler formulas. 
This implies that the Maxwellian distribution takes the form $G_{M}(\xi)=
{\pi^{-\frac{3}{2}}}{e^{-|\xi|^2}}.$

In the case of the self-similar solutions considered in this
paper, $G$ is a non Maxwellian distribution.

We define the energy for particle $e\left(  t\right)  $ as
\[
\rho\left(  t\right)  e\left(  t\right)  =\varepsilon\left(  t\right)
=\int_{\mathbb{R}^{3}}\left\vert w\right\vert ^{2}gdw.
\]
Then, using the approximation (\ref{U1E7})
\[
\rho=\frac{\lambda^{3}}{a}\int G\left(  \xi\right)  d\xi\ ,\ \ e=\lambda
^{2}\frac{\int\left\vert \xi\right\vert ^{2}Gd\xi}{\int G\left(  \xi\right)
d\xi}.
\]
Therefore
\[
\frac{e^{\frac{3}{2}}}{\rho}=a\frac{\left(  \int\left\vert \xi\right\vert
^{2}Gd\xi\right)  ^{\frac{3}{2}}}{\left(  \int G\left(  \xi\right)
d\xi\right)  ^{\frac{5}{2}}}
\]
and
\[
\log\left(  \frac{e^{\frac{3}{2}}}{\rho}\right)  =\log\left(  a\right)
+\log\left[  \frac{\left(  \int\left\vert \xi\right\vert ^{2}Gd\xi\right)
^{\frac{3}{2}}}{\left(  \int G\left(  \xi\right)  d\xi\right)  ^{\frac{5}{2}}%
}\right] .
\]
On the other hand (\ref{U1E8}) yields
\[
\frac{s}{\rho}=\log\left(  a\right)  -\frac{\int G\log\left(  G\right)  d\xi
}{\int G\left(  \xi\right)  d\xi}.
\]
Then
\begin{equation}
\frac{s}{\rho}=\log\left(  \frac{e^{\frac{3}{2}}}{\rho}\right)  +C_{G}%
\label{U1E9}%
\end{equation}
where $C_{G}$ is
\begin{equation}
C_{G}=-\frac{\int G\log\left(  G\right)  d\xi}{\int G\left(  \xi\right)  d\xi
}-\log\left[  \frac{\left(  \int\left\vert \xi\right\vert ^{2}Gd\xi\right)
^{\frac{3}{2}}}{\left(  \int G\left(  \xi\right)  d\xi\right)  ^{\frac{5}{2}}%
}\right] . \label{U2E1}%
\end{equation}

The formula (\ref{U1E9}) has the same form as the usual formula of the entropy
for the equilibrium case, except for the value of the constant $C_{G}.$ In the case of
solutions given by Hilbert expansions the value of $C_{G}$ is the same as the
one in the formula of the entropy for the equilibrium case. Therefore, in the case of
the solutions obtained in this paper which can be approximated by Hilbert
expansions, the asymptotic formula for the entropy by particle is the same as
the one for the equilibrium case.

In the case of the self-similar solutions the value of the constant $C_{G}$
differs from the corresponding value for the one for the equilibrium case. Since the
entropy tends to a maximum for a given value of the particle density and
energy, it follows that $G_{G}<C_{M},$ where $C_{M}$ is the corresponding
value of the constant for a Maxwellian distribution with density one and
temperature one and it takes the value $C_{M}=\frac{3}{2}\big[1-\log\big(\frac{3}{2}\big)\big]$.

In the case of hyperbolic-dominated behavior the formula of the
entropy for the corresponding solutions does not necessarily resemble the formula of the entropy for the equilibrium case, because in general the scaling properties of the particle distributions
are very different from the ones taking place in the case of gases described
by Maxwellian distributions. For further discussions in this direction we refer to \cite{JMNV}. 

\bigskip

\newpage

\section{Table of results}

\label{sec:tableresults}

We collect here all the results obtained in this paper and in \cite{JMNV}.

\begin{itemize}
\item \textbf{Simple shear.}

The critical homogeneity corresponds to $\gamma=0$, i.e. to Maxwell molecules.
\bigskip%

\begin{tabular}
[c]{cc}%
\parbox{7 cm}{\hspace{1cm}\underline{Critical case} \vspace{1mm}
\hspace{1.4cm} ($\gamma=0$) } &
\parbox{7cm}{\hspace{0.7cm}\underline{Supercritical case}  \vspace{1mm}
\hspace{1.2cm}($\gamma>0$)}
\end{tabular}

\vspace{4mm}%

\begin{tabular}
[c]{c|c}%
\parbox{7cm}{
Self-similar solutions \vspace{1.3 mm}
with increasing temperature
} & \parbox{7 cm}{
Maxwellian distribution with \vspace{1.3 mm}
time dependent temperature \vspace{1.3mm}
\hspace{1.5mm}
(Hilbert expansion)
}
\end{tabular}

\bigskip

\item \textbf{Homogeneous dilatation.}

The critical homogeneity corresponds to $\gamma=-2$. \bigskip%

\begin{tabular}
[c]{cc}%
\parbox{7cm}{\hspace{1cm}\underline{Critical case} \vspace{1mm}
\hspace{1.4cm} ($\gamma=-2$) } &
\parbox{7cm}{\hspace{0.7cm}\underline{Subcritical case}  \vspace{1mm}
\hspace{1.2cm}($\gamma<-2$)}
\end{tabular}

\vspace{4mm}%

\begin{tabular}
[c]{c|c}%
\parbox{7cm}{
Maxwellian distribution with \vspace{1.3 mm}
time dependent temperature \vspace{1.3mm}
\hspace{1.5mm}
(Hilbert expansion)
} & \parbox{7cm}{
Maxwellian distribution with \vspace{1.3 mm}
time dependent temperature \vspace{1.3mm}
\hspace{1.5mm}
(Hilbert expansion)
}
\end{tabular}
\bigskip

\item \textbf{Planar shear.}

The critical homogeneity corresponds to $\gamma=0$, i.e. to Maxwell molecules.
\bigskip%

\begin{tabular}
[c]{cc}%
\parbox{7cm}{\hspace{1cm}\underline{Critical case} \vspace{1mm}
\hspace{1.4cm} ($\gamma=0$) } &
\parbox{7cm}{\hspace{0.7cm}\underline{Subcritical case}  \vspace{1mm}
\hspace{1.2cm}($\gamma<0$)}
\end{tabular}

\vspace{4mm}%

\begin{tabular}
[c]{c|c}%
\parbox{7cm}{ \hspace{1.8mm}
Self-similar solutions \vspace{1.3 mm}
} & \parbox{7cm}{
Maxwellian distribution with \vspace{1.3 mm}
time dependent temperature \vspace{1.3mm}
\hspace{1.5mm}
(Hilbert expansion)
}
\end{tabular}
\bigskip

\item \textbf{Planar shear with $K=0$. 
 }

The critical homogeneity corresponds to $\gamma=0$, i.e. to Maxwell molecules.
\bigskip%

\begin{tabular}
[c]{cc}%
\parbox{7cm}{\hspace{1cm}\underline{Critical case} \vspace{1mm}
\hspace{1.4cm} ($\gamma=0$) } &
\parbox{7cm}{\hspace{0.7cm}\underline{Subcritical case}  \vspace{1mm}
\hspace{1.2cm}($\gamma<0$)}
\end{tabular}

\vspace{3mm}%

\begin{tabular}
[c]{c|c}%
\parbox{7cm}{ \hspace{1.8mm}
Self-similar solutions \vspace{1.3 mm}
} & \parbox{7cm}{
Maxwellian distribution with \vspace{1.3 mm}
time dependent temperature \vspace{1.3mm}
\hspace{1.5mm}
(Hilbert expansion)
}
\end{tabular}
\bigskip

\item \textbf{Cylindrical dilatation.}

In this case we have two critical homogeneities: $\gamma=-\frac{3}{2}$ and
$\gamma=-{2}$. \bigskip%

\begin{tabular}
[c]{cc}%
\parbox{7cm}{\hspace{1cm}\hspace{0.8cm} ($\gamma> -{2}$) } &
\parbox{7cm}{\hspace{0.7cm}\hspace{1.2cm}($\gamma<-\frac{3}{2}$)}
\end{tabular}

\vspace{3mm}%

\begin{tabular}
[c]{c|c}%
\parbox{7cm}{ \hspace{1.8mm}
Frozen collisions \vspace{1.3 mm}
} & \parbox{7cm}{
Maxwellian distribution with \vspace{1.3 mm}
time dependent temperature \vspace{1.3mm}
\hspace{1.5mm}
(Hilbert expansion)
}
\end{tabular}
\bigskip

\item \textbf{Combined shear in orthogonal directions $(K_{1}, K_{2}, K_{3})$
with $K_{1}K_{3}\neq0$.}

The critical homogeneity corresponds to $\gamma=0$, i.e. to Maxwell molecules.
\bigskip%

\begin{tabular}
[c]{cc}%
\parbox{7cm}{\hspace{1cm}\underline{Critical case} \vspace{1mm}
\hspace{1.4cm} ($\gamma=0$) } &
\parbox{7cm}{\hspace{0.7cm}\underline{Supercritical case}  \vspace{1mm}
\hspace{1.2cm}($\gamma>0$)}
\end{tabular}

\vspace{3mm}%

\begin{tabular}
[c]{c|c}%
\parbox{7cm}{ \hspace{1.3mm}
Non Maxwellian distribution \vspace{1.3 mm}
} & \parbox{7cm}{
Maxwellian distribution with \vspace{1.3 mm}
time dependent temperature \vspace{1.3mm}
\hspace{1.5mm}
(Hilbert expansion)
}
\end{tabular}
\bigskip

\end{itemize}

\section{Conclusions}

\label{sec:conclusions}

We have obtained several examples of long time asymptotics for
homoenergetic flows of the Boltzmann equation. These flows yield a very rich
class of possible behaviors. Homoenergetic flows can be characterized by a
matrix $L\left(  t\right)  $ which describes the deformation taking place in
the gas. The behavior of the solutions obtained in this paper depends on the
balance between the hyperbolic terms of the equation, which are proportional
to $L\left(  t\right)  $ and the homogeneity of the collision kernel. Roughly
speaking the flows can be classified in three different types, which
correspond to the situations in which the hyperbolic terms are the largest ones
as $t\rightarrow\infty,$ the collision terms are the dominant ones and both of
them have a similar order of magnitude respectively.

In this paper, we provided a rigorous proof of the existence of self-similar solutions
yielding a non Maxwellian distribution of velocities in the case in which the hyperbolic terms and the collisions balance. 
A distinctive feature of
these self-similar solutions is that the corresponding particle distribution
does not satisfy a detailed balance condition. In these solutions the particle
velocities are given by a subtle interplay between particle collisions and shear.

The solutions obtained in this paper yield interesting insights about the
mechanical properties of Boltzmann gases under shear. On the other hand, the
results of this paper suggest many interesting mathematical questions which
deserve further investigation. We have obtained in several cases critical
exponents for the homogeneity of the collision kernel. At the values of those
critical exponents we expect to have self-similar velocities distributions.
This has been proved rigorously in the cases in which the value of the
critical homogeneity is zero, i.e. for Maxwell molecules. New methods are
needed to prove the existence of self-similar solutions for critical
homogeneities different from zero, as for instance we could expect in the case of cylindrical dilatation for the critical value of the homogeneity, i.e. $\gamma=-2$. 

In the case of collision-dominated behavior and in the case of hyperbolic-dominated behavior we  proposed some conjectures for 
asymptotic formulas for the solutions based on formal computations presented in \cite{JMNV}. 
In the first case we have obtained that the corresponding distribution
of particle velocities for the associated homoenergetic flows can be
approximated by a family of Maxwellian distributions with a changing
temperature whose rate of change is obtained by means of a Hilbert expansion. It would be relevant to prove
rigorously the existence of those solutions and to understand their stability
properties. 

In the case in which the hyperbolic terms are much larger than the collision terms 
the resulting solutions yield much more complex behaviors than the ones that we
have obtained in the previous cases. 
The detailed understanding of the particle distributions is largely open and challenging.

Moreover, there are also homoenergetic flows yielding divergent densities or
velocities at some finite time. These flows seem to have also interesting
properties but we have not considered them in this paper.
\bigskip

\bigskip

\textbf{Acknowledgements. }We thank Stefan M\"uller, who motivated us to study this problem, for
useful discussions and suggestions on the topic. The work of R.D.J. was supported by ONR (N00014-14-1-0714), AFOSR (FA9550-15-1-0207), NSF (DMREF-1629026), and the MURI program (FA9550-12-1-0458, FA9550-16-1-0566).  A.N. and J.J.L.V. acknowledge support through the
CRC 1060 \textit{The mathematics of emergent effects }of the University of
Bonn that is funded through the German Science Foundation (DFG).

\end{document}